\documentclass{article}

\usepackage{microtype}
\usepackage{graphicx}
\usepackage{subfigure}
\usepackage{booktabs}
\usepackage[hidelinks]{hyperref}

\usepackage[accepted]{icml2023}

\usepackage{amsmath,amssymb,amsthm}
\usepackage[ruled,vlined]{algorithm2e}
\usepackage[utf8]{inputenc}

\usepackage[capitalize,noabbrev]{cleveref}

\newtheorem{Def}{Definition}[section]
\newtheorem{Thm}{Theorem}[section]
\newtheorem{Lem}{Lemma}[section]
\newtheorem{Cor}{Corollary}[section]

\DeclareMathOperator*{\argmin}{arg\,min}

\newcommand{\PolyLog}{\operatorname{polylog}}

\newcommand{\Lap}{\operatorname{Lap}}

\newcommand{\BigO}{\mathcal O}
\newcommand{\sign}{\operatorname{sign}}

\newcommand{\Gap}{\operatorname{Gap}}
\newcommand{\LSE}{\operatorname{LSE}}
\newcommand{\TV}{\operatorname{TV}}

\newcommand{\A}{\mathcal A}

\newcommand{\D}{\mathcal D}
\newcommand{\E}{\mathbb E}
\newcommand{\Err}{\mathcal E}
\newcommand{\F}{\mathcal F}

\newcommand{\LogL}{\mathcal L}
\newcommand{\N}{\mathcal N}

\newcommand{\R}{\mathbb R}
\newcommand{\T}{\mathcal T}

\usepackage{scalerel,mathtools,color}
\let\svsqrt\sqrt
\newsavebox\Nsqrt
\def\sr#1{\ThisStyle{%
		\savebox\Nsqrt{\scalebox{.5}[1]{$\SavedStyle\svsqrt{\phantom{\cramped{#1#1}}}$}}%
		\ooalign{\usebox{\Nsqrt}\cr\kern.2pt\usebox{\Nsqrt}\cr\hfil$\SavedStyle\cramped{#1}$}}}

\usepackage{enumitem}
\usepackage{bm}
\def\*#1{\mathbf{#1}}

\usepackage{threeparttable}

\usepackage{mathabx}
\def\!{\vvvert}

\usepackage[textsize=tiny]{todonotes}

\icmltitlerunning{Learning-Augmented Private Algorithms for Multiple Quantile Release}

\begin{document}

\twocolumn[
\icmltitle{Learning-Augmented Private Algorithms for Multiple Quantile Release}




\begin{icmlauthorlist}
\icmlauthor{Mikhail Khodak}{cmu}
\icmlauthor{Kareem Amin}{goog}
\icmlauthor{Travis Dick}{goog}
\icmlauthor{Sergei Vassilvitskii}{goog}
\end{icmlauthorlist}

\icmlaffiliation{cmu}{Carnegie Mellon University; work done in part as an intern at Google Research - New York.}
\icmlaffiliation{goog}{Google Research - New York}
\icmlcorrespondingauthor{Mikhail Khodak}{khodak@cmu.edu}

\icmlkeywords{Machine Learning, ICML}

\vskip 0.3in
]



\printAffiliationsAndNotice{}  % leave blank if no need to mention equal contribution

\begin{abstract}
	When applying differential privacy to sensitive data, we can often improve performance using external information such as other sensitive data, public data, or human priors. 
	We propose to use the learning-augmented algorithms (or algorithms with predictions) framework---previously applied largely to improve time complexity or competitive ratios---as a powerful way of designing and analyzing privacy-preserving methods that can take advantage of such external information to improve utility.
	This idea is instantiated on the important task of multiple quantile release, for which we derive error guarantees that scale with a natural measure of prediction quality while (almost) recovering state-of-the-art prediction-independent guarantees.
	Our analysis enjoys several advantages, including minimal assumptions about the data, a natural way of adding robustness, and the provision of useful surrogate losses for two novel ``meta" algorithms that learn predictions from other (potentially sensitive) data. 
	We conclude with experiments on challenging tasks demonstrating that learning predictions across one or more instances can lead to large error reductions while preserving privacy.\looseness-1
\end{abstract}
%!TEX root = main.tex

\vspace{-2.5mm}
\section{Introduction}
\vspace{-1mm}

The differentially private (DP) release of statistics such as the quantile $q$ of a private dataset $\*x\in\R^n$ is an inevitably error-prone task because we are by definition precluded from revealing exact information about the instance at hand~\citep{dwork2014dp}.
However, DP instances rarely occur in a vacuum:
even in the simplest practical settings, we usually know basic information such as the fact that all individuals have a nonnegative age.
Often, the dataset we are considering is drawn from a similar population as a public dataset $\*z\in\R^N$ and should thus have similar quantiles, a case known as the {\em public-private} setting~\citep{liu2021leveraging,bie2022private}.
Alternatively, in what we call {\em sequential release}, we aim to release the quantiles of each of a sequence of datasets $\*x_1,\dots,\*x_T$ one-by-one.
These could be generated by a stationary or other process that allows information derived from prior releases to inform predictions of future releases.
In all of these settings, we might hope to incorporate external information to reduce error, but approaches for doing so tend to be {\em ad hoc} and assumption-heavy.\looseness-1

We propose that the framework of {\em learning-augmented algorithms}---a.k.a. {\em algorithms with predictions}~\citep{mitzenmacher2021awp}---provides the right tools for deriving DP algorithms in this setting, and instantiate this idea for multiple quantile release~\citep{gillenwater2021differentially,kaplan2022quantiles}.
Algorithms with predictions is an expanding field of algorithm design that constructs methods whose instance-dependent performance improves with the accuracy of some prediction about the instance.
The goal is to bound the cost $C_{\*x}(\*w)$ of running on instance $\*x$ given a prediction $\*w$ by some metric $U_{\*x}(\*w)$ of the {\em quality} of the prediction on that instance.
Motivated by practical success~\citep{liu2012renewable,kraska2018case} and as a type of beyond-worst-case analysis~\citep{roughgarden2020beyond}, such algorithms can target a wide variety of cost measures, e.g. competitive ratios in online algorithms~\citep{anand2020customizing,bamas2020primal,diakonilakis2021learning,dutting2021secretary,indyk2020online,yu2022competitive,christianson2023optimal,jiang2020online,kumar2018improving, lykouris2021competitive,rohatgi2020nearoptimal}, space complexity in streaming algorithms~\citep{du2021putting}, and time complexity in graph algorithms~\citep{dinitz2021duals,chen2022faster,sakaue2022dca} and distributed systems~\citep{lattanzi2020scheduling,lindermayr2022permutation,scully2022uniform}.
Departing from such work, we instead aim to design learning-augmented algorithms whose cost $C_{\*x}(\*w)$ captures the error of some statistic---in our case quantiles---computed privately on instance an $\*x$ given a prediction $\*w$. 
We are interested in bounding this cost in terms of the quality of the external information provided to our algorithm, $U_{\*x}(\*w)$.\looseness-1

While incorporating external information into DP is well-studied, c.f. public-private methods~\cite{bie2022private,liu2021leveraging} and private posterior inference~\cite{dimitrakakis2017differential,geumlek2017renyi,seeman2020private}, 
by deriving and analyzing a learning-augmented algorithm for multiple quantiles we show numerous comparative advantages, including:
\ifdefined\arxiv
\begin{enumerate}[itemsep=1pt]
\else
\begin{enumerate}[leftmargin=*,topsep=-2pt,noitemsep]\setlength\itemsep{1pt}
\fi
	\item Minimal assumptions about the data, in our case even fewer than needed by the unaugmented baseline.
	\item Existing tools for studying the robustness of algorithms to noisy predictions~\citep{lykouris2021competitive}.
	\item Co-designing algorithms with predictions together with methods for {\em learning} those predictions from data~\citep{khodak2022awp}, which we show is crucial for both the public-private and sequential release settings.
\end{enumerate}

As part of this analysis we derive a learning-augmented extension of the \texttt{ApproximateQuantiles} (AQ) method of \citet{kaplan2022quantiles} that (nearly) matches its worst-case guarantees while being much better if a natural measure $U_{\*x}(\*w)$ of prediction quality is small.
By studying $U_{\*x}$, we make the following contributions to multiple quantiles:\looseness-1
\ifdefined\arxiv
\begin{enumerate}[itemsep=1pt]
\else
\begin{enumerate}[leftmargin=*,topsep=-2pt,noitemsep]\setlength\itemsep{1pt}
\fi
	\item The first robust algorithm, even for one quantile, that avoids assuming the data is bounded on some interval, specifically by using a heavy-tailed prior.
	\item A provable way of ensuring robustness to poor priors, without losing the consistency of good ones.
	\item A novel connection between DP quantiles and censored regression that leads to (a)~a public-private release algorithm and (b)~a sequential release scheme, both with runtime and error guarantees.
\end{enumerate}
Finally, we integrate these techniques to significantly improve the accuracy of public-private and sequential quantile release on several real and synthetic datasets.
%!TEX root = main.tex

\vspace{-2mm}
\section{Related work}
\vspace{-1mm}

There has been significant work on incorporating external information to improve DP methods.
A major line of work is the public-private framework, where we have access to public data that is related in some way to the private data~\cite{liu2021leveraging,amid2022public,li2022private,bie2022private,bassily2022private}.
The use of public data can be viewed as using a prediction, but such work starts by making (often strong) distributional assumptions on the public and private data;
we instead derive instance-dependent upper bounds with minimal assumptions that we then apply to such public-private settings.
Furthermore, our framework allows us to ensure robustness to poor predictions without distributional assumptions, and to derive learning algorithms using training data that may itself be sensitive.
Another approach is to treat DP mechanisms (e.g. the exponential) as Bayesian posterior sampling~\cite{dimitrakakis2017differential,geumlek2017renyi,seeman2020private}. 
Our work can be viewed as an adaptation where we give explicit prior-dependent utility bounds. To our knowledge, no such guarantees exist in the literature. Moreover, while our focus is quantile estimation, the predictions-based framework that we advocate is much broader, as many DP methods---including for multiple quantiles---combine multiple queries that must be considered jointly.\looseness-1

Our approach for augmenting DP with external information centers the algorithms with predictions framework, where past work has focused on using predictions to improve metrics related to time, space, and communication complexity. 
We make use of existing techniques from this literature, including robustness-consistency tradeoffs~\cite{lykouris2021competitive} and the online learning of predictions~\cite{khodak2022awp}.
Tuning DP algorithms has been an important topic in private machine learning, e.g. for hyperparameter tuning~\cite{chaudhari2013stability} and federated learning~\cite{andrew2021differentially}, but these have not to our knowledge considered incorporating per-instance predictions.\looseness-1

The specific task we focus on is DP quantiles, a well-studied problem~\citep{gillenwater2021differentially,kaplan2022quantiles}, but we are not aware of work adding outside information.
We also make the important contribution of an effective method for removing data-boundedness assumptions.
Our algorithm builds upon the state-of-the-art work of~\citet{kaplan2022quantiles}, which is also our main source for empirical comparison.
%!TEX root = main.tex

\vspace{-2mm}
\section{Augmenting a private algorithm}\label{sec:algorithm}
\vspace{-1mm}

The basic requirement for a learning-augmented algorithm is that the cost $C_{\*x}(\*w)$ of running it on an instance $\*x$ with prediction $\*w$ should be upper bounded---usually up to constant or logarithmic factors---by a metric $U_{\*x}(\*w)$ of the quality of the prediction on the instance.
We denote this by $C_{\*x}\lesssim U_{\*x}$.
In our work the cost $C_{\*x}(\*w)$ will be the error of a privately released statistic, as compared to some ground truth.
We will use the following privacy notion:\looseness-1
\begin{Def}
	[\citet{dwork2014dp}]
	Algorithm $\A$ is {\bf $(\varepsilon,\delta)$-differentially private} if for all subsets $S$ of its range, $\Pr\{\A(\*x)\in S\}\le e^\varepsilon\Pr\{\A(\*{\tilde x})\in S\}+\delta$ whenever $\*x\sim\*{\tilde x}$ are {\bf neighboring}, i.e. they differ in at most one element.\looseness-1
\end{Def}
Using $\varepsilon$-DP to denote $(\varepsilon,0)$-DP, the broad goal of this work will be to reduce the error $C_{\*x}(\*w)$ of  $\varepsilon$-DP multiple quantile release while fixing the privacy level $\varepsilon$.

\vspace{-1mm}
\subsection{Problem formulation}\label{sec:formulation}
\vspace{-1mm}

A good guarantee for a learning-augmented algorithm will have several important properties that formally separate its performance from naive upper bounds $U_{\*x}\gtrsim C_{\*x}$.
The first, {\em consistency}, requires it to be a reasonable indicator of strong performance in the limit of perfect prediction:
\begin{Def}\label{def:consistency}
	A learning-augmented guarantee $C_{\*x}\lesssim U_{\*x}$ is {\bf $c_{\*x}$-consistent} if $C_{\*x}(\*w)\le c_{\*x}$ whenever $U_{\*x}(\*w)=0$.\looseness-1
\end{Def}
Here $c_{\*x}$ is a prediction-independent quantity that should depend weakly or not at all on problem difficulty (in the case of quantiles, the minimum separation between data points).
Consistency is often presented via a tradeoff with {\em robustness}~\citep{lykouris2021competitive}, which bounds how poorly the method can do when the prediction is bad, in a manner similar to a standard worst-case bound:
\begin{Def}\label{def:robustness}
	A learning-augmented guarantee $C_{\*x}\lesssim U_{\*x}$ is {\bf $r_{\*x}$-robust} if it implies $C_{\*x}(\*w)\le r_{\*x}$ for all predictions $\*w$.
\end{Def}
Unlike consistency, robustness usually depends strongly on the difficulty of the instance $\*x$, with the goal being to not do much worse than a prediction-free approach.
Note that the latter is trivially robust but not (meaningfully) consistent, since it ignores the prediction;
this makes clear the need for considering the two properties via tradeoff between them.\looseness-1

As discussed further in Section~\ref{sec:robustness}, this existing language for quantifying robustness is one of the advantages of using the framework of learning-augmented algorithms for incorporating external information into DP methods.
We report robustness-consistency trade-offs for our quantile release algorithms in the same section.\looseness-1

A last desirable property of the prediction quality measure $U_{\*x}(\*w)$ is that it should be useful for making good predictions.
One way to formalize this is to require $U_{\*x_t}$ to be {\em learnable} from multiple instances $\*x_t$.
For example, we could ask for {\em online} learnability, i.e. the existence of an algorithm that makes predictions $\*w_t\in W$ in some action space $W$ given instances $\*x_1,\dots,\*x_{t-1}$ whose {\em regret} is sublinear in $T$:\looseness-1
\begin{Def}
	The {\bf regret} of actions $\*w_1,\dots,\*w_T\in W$ on the sequence of functions $U_{\*x_1},\dots,U_{\*x_T}$ is $\max_{\*w\in W}\sum_{t=1}^TU_{\*x_t}(\*w_t)-U_{\*x_t}(\*w)$.
\end{Def}
Sublinear regret implies average prediction quality as good as that of the optimal prediction in hindsight, up to an additive term that vanishes as $T\to\infty$.
Since $U_{\*x_t}$ roughly upper-bounds the error $C_{\*x_t}$, this means that asymptotically the average error is governed by the average prediction quality $\min_{\*w\in W}\frac1T\sum_{t=1}^TU_{\*x_t}(\*w)$ of the optimal $\*w\in W$.
A crucial observation here is that sublinear regret can often be obtained by making the function $U_{\*x}$ amenable to familiar gradient-based online convex optimization methods such as online gradient descent~\citep{khodak2022awp}.
Doing so also enables instance-dependent linear prediction: 
setting $\*w_t$ using a learned function of some instance features $\*f_t$.\looseness-1

We demonstrate the usefulness of both learning and robustness-consistency analysis in two applications where it is reasonable to have external information about the sensitive dataset(s).
In the {\bf public-private} setting, the prediction $\*w$ is obtained from a public dataset $\*x'$ that is assumed to be similar to $\*x$ but is not subject to privacy-protection.
In {\bf sequential release}, we privately release information about each dataset in a sequence $\*x_1,\dots,\*x_T$;
the release at time $t$ can depend on $\*x_t$ and on a prediction $\*w_t$, which can be derived (privately) from past observations.
In Section~\ref{sec:applications} we show that sequential release can be posed directly as a private online learning problem, while the public-private setting can be approached via online-to-batch conversion~\citep{cesa-bianchi2004online2batch}.
Both are thus directly enabled by treating the prediction quality measures $U_{\*x_t}$ as surrogate objectives for the actual cost functions $C_{\*x}$ and applying standard optimization techniques~\citep{khodak2022awp}.\looseness-1

With these desiderata of algorithms with predictions guarantees in-mind, we now move to deriving them for quantile release. 
The robustness and learnability of the resulting prediction quality measures $U_{\*x}$ are discussed in Section~\ref{sec:advantages}.\looseness-1

\vspace{-1mm}
\subsection{Warm-up: Releasing one quantile}\label{sec:single}
\vspace{-1mm}

Given a quantile $q\in(0,1)$ and a sorted dataset $\*x\in\R^n$ of $n$ distinct points, we want to release $o\in[\*x_{[\lfloor qn\rfloor]},\*x_{[\lfloor qn\rfloor+1]})$, i.e. such that the proportion of entries less than $o$ is $q$.
As in prior work~\cite{kaplan2022quantiles}, the error of $o$ will be the number of points between it and the desired interval:
\begin{equation}\label{eq:gap}
\Gap_q(\*x,o)
=||\{i:\*x_{[i]}<o\}|-\lfloor qn\rfloor|
=|\max_{\*x_{[i]}<o}i-\lfloor qn\rfloor|
\end{equation}
$\Gap_q(\*x,o)$ is constant on intervals $I_k=(\*x_{[k]},\*x_{[k+1]}]$ in the partition by $\*x$ of $\R$ (let $I_0=(-\infty,\*x_{[1]}]$ and $I_n=(\*x_{[n]},\infty)$), so we also say that $\Gap_q(\*x,I_k)$ is the same as $\Gap_q(\*x,o)$ for some $o$ in the interior of $I_k$.\looseness-1

For single quantile release we choose perhaps the most natural way of specifying a prediction for a DP algorithm:
via the base measure $\mu:\R\mapsto\R_{\ge0}$ of the exponential mechanism:\looseness-1
\begin{Thm}[\citet{mcsherry2007mechanism}]
	If the {\bf utility} $u(\*x,o)$ of an outcome $o$ of a query over dataset $\*x$ has {\bf sensitivity} $\max_{o,\*x\sim\*{\tilde x}}|u(\*x,o)-u(\*{\tilde x},o)|\le\Delta$ then the {\bf exponential mechanism}, which releases $o$ w.p. $\propto\exp(\frac\varepsilon{2\Delta}u(\*x,o))\mu(o)$ for some base measure $\mu$, is $\varepsilon$-DP.\looseness-1
\end{Thm}

The utility function we use is $u_q=-\Gap_q$, so since this is constant on each interval $I_k$ the mechanism here is equivalent to sampling $k$ w.p. $\propto\exp(\varepsilon u_q(\*x,I_k)/2)\mu(I_k)$ and then sampling $o$ from $I_k$ w.p. $\propto\mu(o)$.
While the idea of specifying a prior for EM is well-known, the key idea here is to obtain a prediction-dependent bound on the error that reveals a useful measure of the {\em quality} of the prediction.
In particular, we can show (c.f.~Lemma~\ref{lem:quantile}) that running EM in this way yields $o$ that w.p. $\ge1-\beta$ satisfies\vspace{-1mm}
\begin{equation}\label{eq:basic}
\Gap_q(\*x,o)\le\frac2\varepsilon\log\frac{1/\beta}{\Psi_{\*x}^{(q,\varepsilon)}(\mu)}\le\frac2\varepsilon\log\frac{1/\beta}{\Psi_{\*x}^{(q)}(\mu)}\vspace{-1mm}
\end{equation}
where the quantity $\Psi_{\*x}^{(q,\varepsilon)}=\int\exp(-\frac\varepsilon2\Gap_q(\*x,o))\mu(o)do$ is the inner product between the prior and the EM score while $\Psi_{\*x}^{(q)}=\lim_{\varepsilon\to\infty}\Psi_{\*x}^{(q,\varepsilon)}=\mu((\*x_{[\lfloor qn\rfloor]},\*x_{[\lfloor qn\rfloor+1]}])$ is the probability that the prior assigns to the optimal interval.

This suggests two metrics of prediction quality:
the negative log-inner-products $U_{\*x}^{(q,\varepsilon)}(\mu)=-\log\Psi_{\*x}^{(q,\varepsilon)}(\mu)$ and $U_{\*x}^{(q)}(\mu)=-\log\Psi_{\*x}^{(q)}(\mu)$.
Both make intuitive sense:
we expect predictions $\mu$ that assign a high probability to intervals that the EM score weighs heavily to perform well, and EM assigns the most weight to the optimal interval.
There are also many ways that these metrics are useful.
For one, in the case of perfect prediction---i.e. if $\mu$ assigns probability one to the optimal interval $I_{\lfloor qn\rfloor}$---then $\Psi_{\*x}^{(q,\varepsilon)}(\mu)=\Psi_{\*x}^{(q)}(\mu)=1$, yielding an upper bound on the error of only $\frac2\varepsilon\log\frac1\beta$.
Secondly, as we will see, both are also amenable for analyzing robustness (the mechanism's sensitivity to \emph{incorrect} priors) and learning.
A final and important quality is that the guarantees using these metrics hold under no extra assumptions. Between the two, the first metric provides a tighter bound on the utility loss while the second does not depend on $\varepsilon$, which may be desirable.

It is also fruitful to analyze the metrics for specific priors.
When $\*x$ is in a bounded interval $(a,b)$ and $\mu(o)=\frac{1_{o\in(a,b)}}{b-a}$ is the uniform measure, then $\Psi_{\*x}^{(q)}(\mu)\ge\frac{\psi_{\*x}}{b-a}$, where $\psi_{\*x}$ is the minimum distance between entries;
thus we recover past bounds, e.g. \citet[Lemma~A.1]{kaplan2022quantiles}, that implicitly use this measure to guarantee $\Gap_q(\*x,o)\le\frac2\varepsilon\log\frac{b-a}{\beta\psi_{\*x}}$. 
Here the support of the uniform distribution is correct by assumption as the data is assumed bounded. 
However, analyzing $\Psi_{\*x}^{(q)}$ also yields a novel way of removing this assumption: 
if we suspect the data lies in $(a,b)$, we set $\mu$ to be the Cauchy prior with location $\frac{a+b}2$ and scale $\frac{b-a}2$.
Even if we are wrong about the interval, there exists an $R>0$ s.t. the data lies in the interval $(\frac{a+b}2\pm R)$, so using the Cauchy yields $\Psi_{\*x}^{(q)}\ge\frac{2(b-a)\psi_{\*x}/\pi}{(b-a)^2+4R^2}$ and thus the following guarantee:\looseness-1
\begin{Cor}[of Lem.~\ref{lem:quantile}]\label{cor:cauchy}
	If the data lies in the interval $(\frac{a+b}2\pm R)$ and $\mu$ is the Cauchy measure with location $\frac{a+b}2$ and scale $\frac{b-a}2$ then the output of the exponential mechanism satisfies $\Gap_q(\*x,o)\le\frac2\varepsilon\log\left(\pi\frac{b-a+\frac{4R^2}{b-a}}{2\beta\psi_{\*x}}\right)$ w.p. $\ge1-\beta$.\looseness-1
\end{Cor}
If $R=\frac{b-a}2$, i.e. we get the interval right, then the bound is only an additive factor $\frac2\varepsilon\log\pi$ worse than before, but if we are wrong then performance degrades as $\BigO(\log(1+R^2))$, unlike the $\BigO(R)$ error of the uniform prior.
Note our use of a heavy-tailed distribution here:
a sub-exponential density decays too quickly and leads to error $\BigO(R)$ rather than $\BigO(\log(1+R^2))$.
We can also adapt this technique if we know only a single-sided bound, e.g. if values must be positive, by using an appropriate half-Cauchy distribution.

\vspace{-1mm}
\subsection{Releasing multiple quantiles}
\vspace{-1mm}

\ifdefined\arxiv
\begin{figure}[!t]
	\centering
	\includegraphics[width=.495\linewidth]{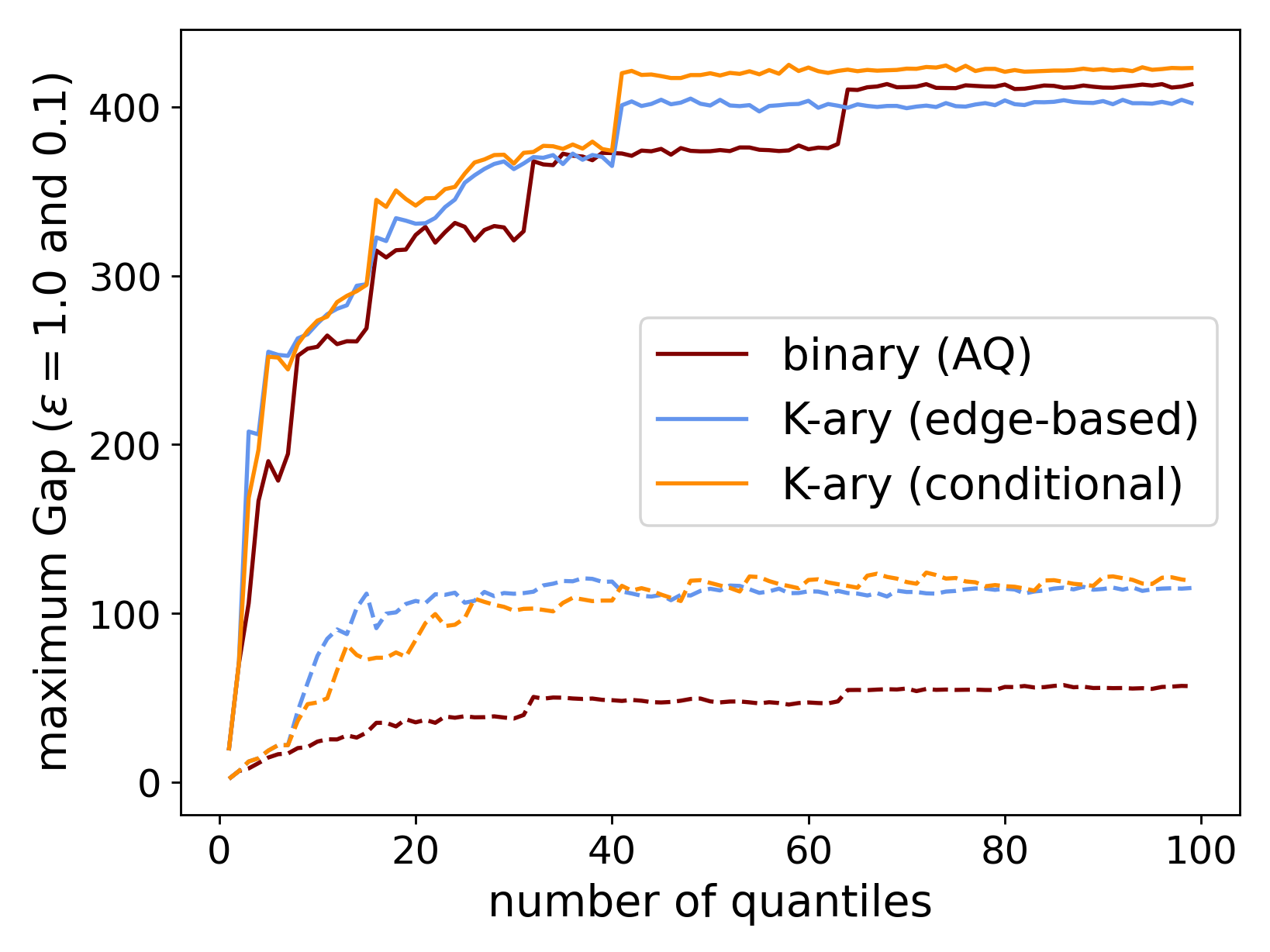}
	\hfill
	\includegraphics[width=.495\linewidth]{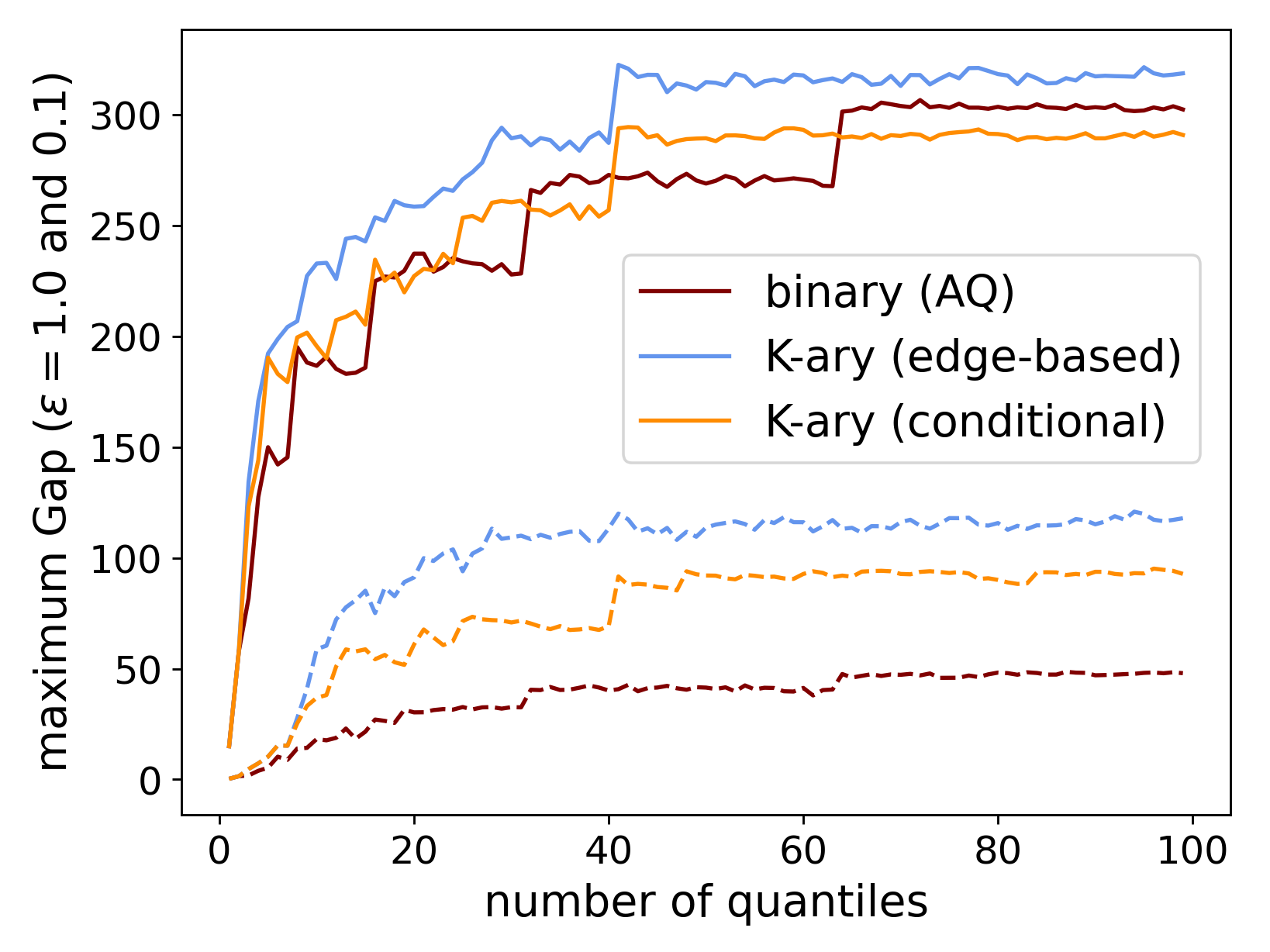}\vspace{-4mm}
	\caption{\label{fig:mplots}
		Maximum gap as a function of $m$ for different variants of AQ when using the Uniform prior, evaluated on 1000 samples from a standard Gaussian (left) and the Adult ``age" dataset (right).
		The dashed and solid lines correspond to $\varepsilon=1$ and 0.1, respectively.\looseness-1
	}
\end{figure}
\else
\fi

To simultaneously estimate quantiles $q_1,\dots,q_m$ we adapt the \texttt{ApproximateQuantiles} method of \cite{kaplan2022quantiles}, which assigns each $q_i$ to a node in a binary tree and, starting from the root, uses EM with the uniform prior to estimate a quantile before sending the data below the outcome $o$ to its left child and the data above $o$ to its right child.
Thus each entry is only involved in $\lceil\log_2m\rceil$ exponential mechanisms, and so for data in $(a,b)$ the maximum $\Gap_{q_i}$ across quantiles is $\BigO\left(\frac{\log^2m}\varepsilon\log\frac{m(b-a)}{\beta\psi_{\*x}}\right)$, which is much better than the naive bound of a linear function of $m$.\looseness-1

Given one prior $\mu_i$ for each $q_i$, a naive extension of~\eqref{eq:basic} gets a similar $\PolyLog(m)$ bound (c.f. Lem~\ref{lem:empirical-quantiles}); 
notably we extend the Cauchy-unboundedness result to multiple quantiles (c.f. Cor.~\ref{cor:multi-cauchy}).
However the upper bound is not a deterministic function of $\mu_i$, as it depends on restrictions of $\*x$ and $\mu_i$ to subsets $(o_j,o_k)$ of the domain induced by the outcomes of EM for quantiles $q_j$ and $q_k$ earlier in the tree.
It thus does not encode a direct relationship between the prediction and instance data and is less amenable for learning.\looseness-1

We instead want guarantees depending on a more natural metric, e.g. one aggregating $\Psi_{\*x}^{(q_i,\varepsilon_i)}(\mu_i)$ from the previous section across pairs $(q_i,\mu_i)$.
The core issue is that the data splitting makes the probability assigned by a prior $\mu_i$ to data outside the interval $(o_j,o_k)$ induced by the outcomes of quantiles $q_j$ and $q_k$ earlier in the tree not affect the distribution of $o_i$. 
One way to handle this is to assign this probability mass to the edges of $(o_j,o_k)$, rather than the more natural conditional approach of \texttt{ApproximateQuantiles}.
We refer to this as ``edge-based prior adaptation" and use it to bound $\Gap_{\max}=\max_i\Gap_{q_i}(\*x,o_i)$ via the harmonic mean $\Psi_{\*x}^{(\varepsilon)}$ of the inner products $\Psi_{\*x}^{(q_i,\varepsilon_i)}(\mu_i)$:\looseness-1
\begin{Thm}[c.f. Thm.~\ref{thm:binary}]
	If $m=2^k-1$ for some $k$, quantiles $q_1,\dots,q_m$ are uniformly spaced, and for each we have a prior $\mu_i:\R\mapsto\R_{\ge0}$, then running \texttt{ApproximateQuantiles} with edge-based prior adaptation (c.f. Algorithm~\ref{alg:quantiles}) is $\varepsilon$-DP, and w.p. $\ge1-\beta$\looseness-1\vspace{-2mm}
	\ifdefined\arxiv
	\begin{equation}
		\Gap_{\max}
		\le\frac2\varepsilon\phi^{\log_2(m+1)}\lceil\log_2(m+1)\rceil\log\frac{m/\beta}{\Psi_{\*x}^{(\varepsilon)}}
		\qquad\textrm{for}\qquad\Psi_{\*x}^{(\varepsilon)}=\left(\sum_{i=1}^m\frac{1/m}{\Psi_{\*x}^{(q_i,\varepsilon_i)}(\mu_i)}\right)^{-1}
	\end{equation}
	\else
	\begin{align}
		\begin{split}
			\Gap_{\max}
			&\le\frac2\varepsilon\phi^{\log_2(m+1)}\lceil\log_2(m+1)\rceil\log\frac{m/\beta}{\Psi_{\*x}^{(\varepsilon)}}\\
			\qquad\quad\textrm{for}\quad&\Psi_{\*x}^{(\varepsilon)}=\left(\sum_{i=1}^m\frac{1/m}{\Psi_{\*x}^{(q_i,\varepsilon_i)}(\mu_i)}\right)^{-1}
		\end{split}
	\end{align}
	\fi
	Here $\varepsilon_i=\frac\varepsilon{\lceil\log_2(m+1)\rceil}$ and $\phi=\frac{1+\sqrt 5}2$ is the golden ratio.\looseness-1
\end{Thm}
The golden ratio is due to a Fibonacci-type recurrence bounding the maximum $\Gap_{q_i}$ at each depth of the tree.
$\Psi_{\*x}^{(\varepsilon)}$ depends only on $\*x$ and predictions $\mu_i$, and it yields a nice error metric $U_{\*x}^{(\varepsilon)}
=-\log\Psi_{\*x}^{(\varepsilon)}=\log\sum_{i=1}^me^{U_{\*x}^{(q_i,\varepsilon_i)}}$.
However, the dependence of the error on $m$ is worse than of \texttt{ApproximateQuantiles}, as $\phi^{\log_2m}$ is roughly $\BigO(m^{0.7})$.
The bound is still sublinear and thus better than the naive baseline of running EM $m$ times.\looseness-1

The $\tilde\BigO(\phi^{\log_2m})$ dependence results from error compounding across depths of the tree, so we can try to reduce depth by going from a binary to a $K$-ary tree.
This involves running EM $K-1$ times at each node---and paying $K-1$ more in budget---to split the data into $K$ subsets;
the resulting estimates may also be out of order.
However, by showing that sorting them back into order does not increase the error and then controlling the maximum $\Gap_{q_i}$ at each depth via another recurrence relation, we prove the following:
\begin{Thm}[c.f. Thm.~\ref{thm:kary}]\label{thm:kary-main}
	For any $q_1,\dots,q_m$, using $K=\lceil\exp(\sqrt{\log2\log(m+1)})\rceil$ and edge-based adaptation guarantees $\varepsilon$-DP and w.p. $\ge1-\beta$ has \ifdefined\arxiv
	\begin{equation}
		\Gap_{\max}\le\frac{2\pi^2}\varepsilon\exp\left(2\sqrt{\log(2)\log(m+1)}\right)\log\frac{m/\beta}{\Psi_{\*x}^{(\varepsilon)}}
	\end{equation}
	\else $\Gap_{\max}\le\frac{2\pi^2}\varepsilon\exp\left(2\sqrt{\log(2)\log(m+1)}\right)\log\frac{m/\beta}{\Psi_{\*x}^{(\varepsilon)}}$.\looseness-1
	\fi
\end{Thm} 
The rate in $m$ is both sub-polynomial and super-poly-logarithmic ($o(m^\alpha)$ and $\omega(\log^\alpha m)~\forall~\alpha>0$);
while asymptotically worse than the prediction-free original result~\cite{kaplan2022quantiles}, 
for almost any practical value of $m$ (e.g. $m\in [3,10^{12}]$)
 it does not exceed a small constant (e.g. nine) times $\log^3m$.
Thus if the error $-\log\Psi_{\*x}^{(\varepsilon)}$ of the prediction is small---i.e. the inner products between priors and EM scores are large on (harmonic) average---then we may do much better with this approach.\looseness-1

\ifdefined\arxiv\else
\begin{figure}[!t]
	\centering
	\includegraphics[width=.495\linewidth]{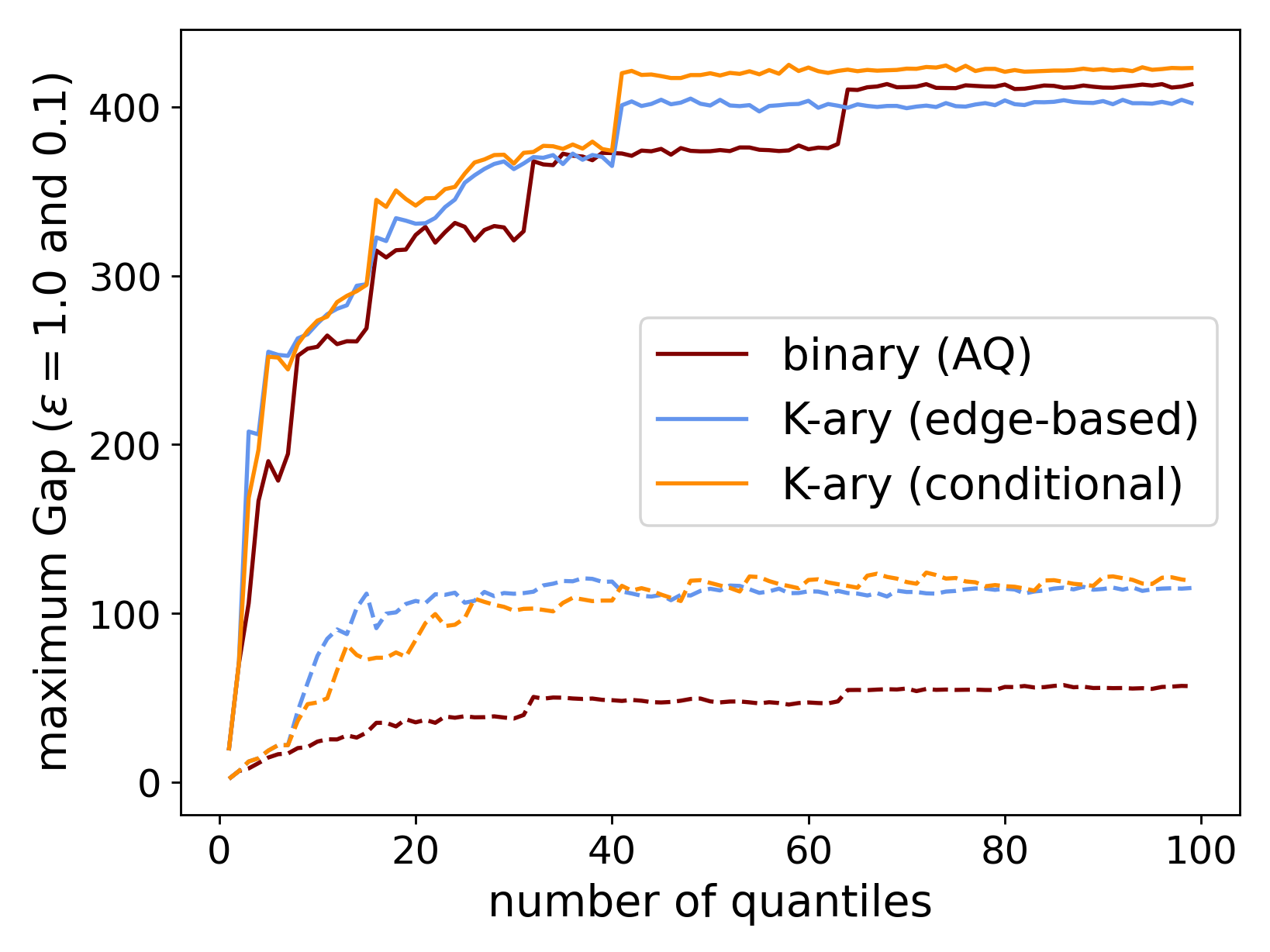}
	\hfill
	\includegraphics[width=.495\linewidth]{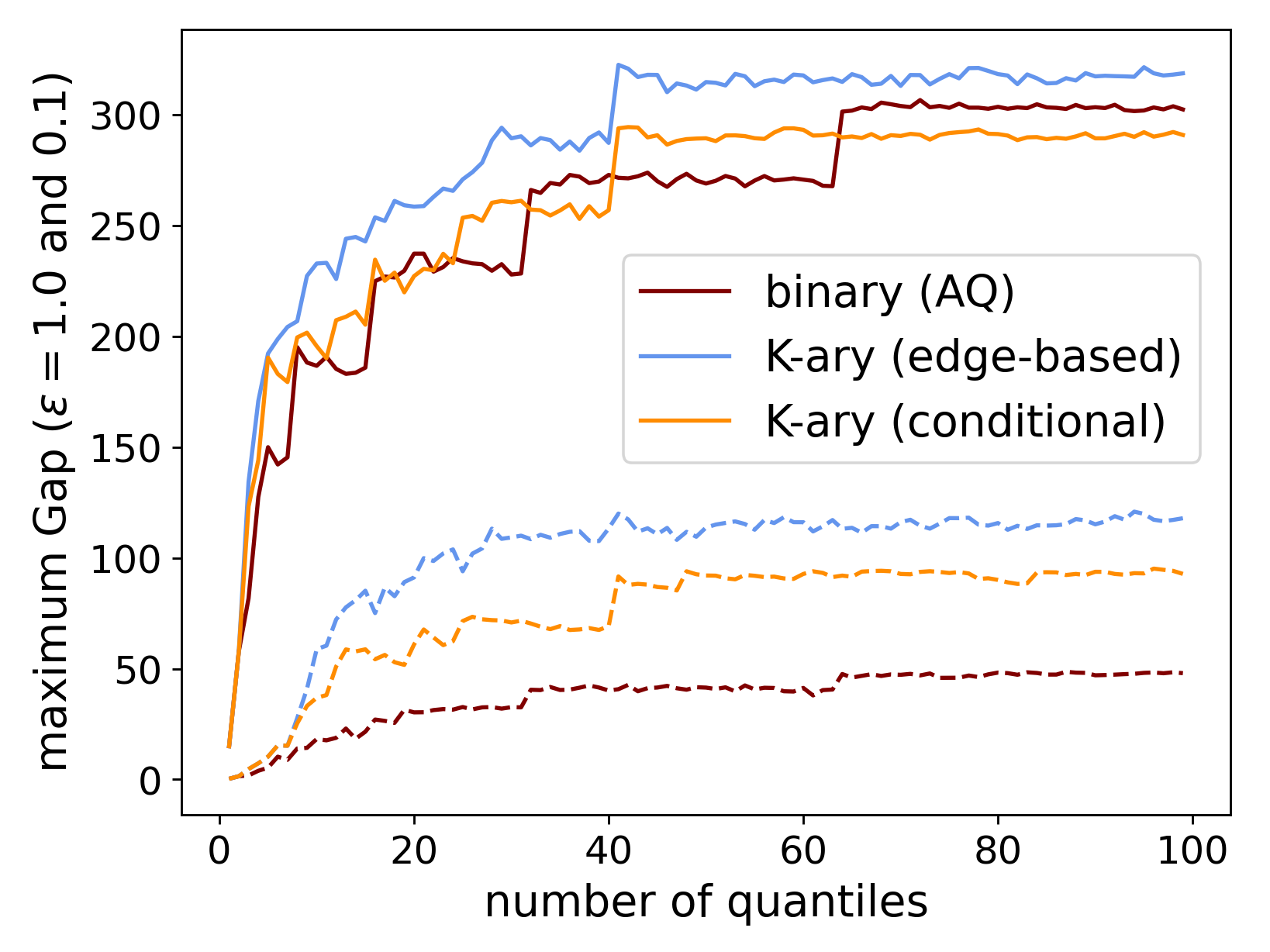}\vspace{-4mm}
	\caption{\label{fig:mplots}
		Maximum gap as a function of $m$ for different variants of AQ when using the Uniform prior, evaluated on 1000 samples from a standard Gaussian (left) and the Adult ``age" dataset (right).
		The dashed and solid lines correspond to $\varepsilon=1$ and 0.1, respectively.\looseness-1
	}
\end{figure}
\fi

We compare K-ary AQ with edge-based adaptation to regular AQ on two datasets in Figure~\ref{fig:mplots}.
The original is better at higher $\varepsilon$ but similar or worse at higher privacy.
We also find that conditional adaptation is only better on discretized data that can have repetitions, a case where neither method provides guarantees.
Overall, we find that our prior-dependent analysis covers a useful algorithm, but for consistency with past work and due to its better performance at high $\varepsilon$ we will focus on the original binary approach in experiments.
%!TEX root = main.tex

\vspace{-2mm}
\section{Utility of learning-augmented algorithms}\label{sec:advantages}
\vspace{-1mm}

In the previous section we derived a data-dependent function $U_{\*x}^{(\varepsilon)}=-\log\Psi_{\*x}^{(\varepsilon)}$ that upper bounds the error of quantile release using priors $\mu_1,\dots,\mu_m$.
As in the single-quantile case, we can construct a looser, $\varepsilon$-independent upper bound
\begin{equation}
	U_{\*x}
	=-\log\Psi_{\*x}
	=\log\sum_{i=1}^me^{U_{\*x}^{(q_i)}}
	\ge U_{\*x}^{(\varepsilon)}
\end{equation}
using the harmonic mean $\Psi_{\*x}$ of $\Psi_{\*x}^{(q_i)}$.
We next summarize the usefulness of these upper bounds for understanding and applying DP methods with external information.
Note that all three aspects below are crucial in our experiments.

\vspace{-1mm}
\subsection{Minimal assumptions and new insights}
\vspace{-1mm}

Our guarantees require no extra data assumptions:
in-fact, the first outcome of our analysis was {\em removing} a boundedness assumption.
This contrasts with past public-private work~\citep{liu2021leveraging,bie2022private}, which makes distributional assumptions, and is why we can apply these results to two very distinct settings in Section~\ref{sec:applications}.\looseness-1

\vspace{-1mm}
\subsection{Ensuring robustness}\label{sec:robustness}
\vspace{-1mm}

While we incorporate external information into DP-algorithms because we hope to improve performance, if not done carefully it may lead to worse results.
For example, a quantile prior concentrated away from the data may have error depending linearly on the distance to the optimal interval.
Ideally an algorithm that uses a prediction will be robust, i.e. revert back to worst-case guarantees if the prediction is poor, without significantly sacrificing consistency, i.e. performing well if the prediction is good.

Using the formalization of these properties in Definitions~\ref{def:consistency} and~\ref{def:robustness}, algorithms with predictions provides a convenient way to deploy them by {\em parameterizing} the robustness-consistency tradeoff, in which methods are designed to be $r_{\*x}(\lambda)$-robust and $c_{\*x}(\lambda)$-consistent for a user-specified parameter $\lambda\in[0,1]$~\citep{bamas2020primal,lykouris2021competitive}.
For quantiles, we can obtain an elegant parameterized tradeoff by interpolating prediction priors with a ``robust" prior.
In particular, since $\Psi_{\*x}^{(q,\varepsilon)}$ is linear we can pick $\rho$ to be a trusted prior such as the uniform or Cauchy and for any prediction $\mu$ use $\mu^{(\lambda)}=(1-\lambda)\mu+\lambda\rho$ instead.
Setting $\Psi_{\*x}^{(q,\varepsilon)}(\mu^{(\lambda)})=(1-\lambda)\Psi_{\*x}^{(q,\varepsilon)}(\mu)+\lambda\Psi_{\*x}^{(q,\varepsilon)}(\rho)$ in \eqref{eq:basic} yields:\looseness-1
\begin{Cor}[of Lem.~\ref{lem:quantile}; c.f. Cor.~\ref{cor:quantile}]
	For quantile $q$, applying EM with prior $\mu^{(\lambda)}=(1-\lambda)\mu+\lambda\rho$ is $\left(\frac2\varepsilon\log\frac{1/\beta}{\lambda\Psi_{\*x}^{(q,\varepsilon)}(\rho)}\right)$-robust and $\left(\frac2\varepsilon\log\frac{1/\beta}{1-\lambda}\right)$-consistent.
\end{Cor}
Thus w.h.p. error is simultaneously at most $\frac2\varepsilon\log\frac1\lambda$ worse than that of only using the robust prior $\rho$ and we only have error $\frac2\varepsilon\log\frac{1/\beta}{1-\lambda}$ if the prediction $\mu$ is perfect, i.e. if it is only supported on the optimal interval.
This is easy to extend to the multiple-quantile metric $-\log\Psi_{\*x}^{(\varepsilon)}$.
In fact, we can even interpolate between the $\PolyLog(m)$ prediction-free guarantee of past work and our learning-augmented guarantee with the worse dependence on $m$; 
thus if the prediction is not good enough to overcome this worse rate we can still ensure that we do not do much worse than the original guarantee.\looseness-1
\begin{Cor}[of Lem.~\ref{lem:empirical-quantiles} \& Thm.~\ref{thm:binary}; c.f. Cor.~\ref{cor:multiple}]\label{cor:cauchy-robust}
	If we run binary AQ on data in the interval $(\frac{a+b}2\pm R)$ for unknown $R>0$ and use the prior $\mu_i^{(\lambda)}=(1-\lambda)\mu+\lambda\rho$ for each $q_i$, where $\rho$ is Cauchy $(\frac{a+b}2,\frac{b-a}2)$, then the algorithm is  $\left(\frac2\varepsilon\lceil\log_2m\rceil^2\log\left(\pi m\frac{b-a+\frac{4R^2}{b-a}}{2\lambda\beta\psi_{\*x}}\right)\right)$-robust and $\left(\frac2\varepsilon\phi^{\log_2m}\lceil\log_2m\rceil\log\frac{m/\beta}{1-\lambda}\right)$-consistent.
\end{Cor}

These results show the advantage of our framework in designing algorithms that make robust use of possibly noisy predictions.
Notably, related public-private work that studies robustness still assumes source and target data are Gaussian~\cite{bie2022private}, whereas we make no distributional assumptions.
We demonstrate the importance of this robustness technique throughout our experiments in Section~\ref{sec:applications}.

\vspace{-1mm}
\subsection{Learning}
\vspace{-1mm}

A last important use for prior-dependent bounds is as surrogate objectives for optimization.
As we show in Section~\ref{sec:applications}, being able to learn across upper bounds $U_{\*x_1},\dots,U_{\*x_T}$ of a sequence of (possibly sensitive) datasets $\*x_t$ is useful for both the public-private and sequential release.
Algorithms with predictions guarantees are often sufficiently nice to do this using off-the-shelf online learning~\citep{khodak2022awp}, a property that largely holds for our upper bounds as well.\looseness-1

Most saliently, the bound $U_{\*x}^{(q,\varepsilon)}\hspace{-.5mm}=\hspace{-.5mm}-\log\Psi_{\*x}^{(q,\varepsilon)}$ is a convex function of an inner product $\Psi_{\*x}^{(q,\varepsilon)}$ between the EM score and the prior $\mu$;
thus by discretizing one can learn over a large family of piecewise-constant priors, which themselves Lipschitz priors over a bounded domain.
The same is true of the multiple quantile bound $U_{\*x}^{(\varepsilon)}$ because it is the log-sum-exp over $U_{\*x}^{(q_i,\varepsilon_i)}$ and thus also convex.
Thus in theory we can (privately) online learn the sequence $U_{\*x_t}^{(\varepsilon)}$ with low-regret w.r.t. any set of $m$ Lipschitz priors (c.f.~Theorem~\ref{thm:quantile-regret}).
However, in-practice we may not want to learn in the high dimensions needed by the discretization, and rather than fixed priors we may wish to learn a mapping from dataset-specific features.
In Section~\ref{sec:applications} we thus focus on learning the less-expressive family of location-scale models.\looseness-1
%!TEX root = main.tex

\vspace{-2mm}
\section{Applications}\label{sec:applications}
\vspace{-1mm}

We now consider our two applications from the introduction---public-private and sequential release---using a specific class of location-scale priors, which for some measure $f:\R\mapsto\R_{\ge0}$ have form $\mu_{\nu,\sigma}(x)=\frac1\sigma f\left(\frac{x-\nu}\sigma\right)$ for $\nu\in\R$ and $\sigma>0$.
Such families allows us to model both the location of a quantile using $\nu=\langle\*w,\*f\rangle$---where $\*w\in\R^d$ is a linear model from public features $\*f\in\R^d$ about the dataset $\*x\in\R^n$---and our uncertainty about it using $\sigma$, all while staying in reasonable dimensions.
Note that in this section we use only the $\varepsilon$-independent bound $U_{\*x}$, as $U_{\*x}^{(\varepsilon)}$ does not yield a convex objective;
furthermore, while we mainly discuss the single-quantile bound $U_{\*x}^{(q)}$ for simplicity, the general results (c.f.~Section~\ref{app:applications}) extend naturally to the case of $m>1$ because it is the log-sum-exp of the former.\looseness-1

\vspace{-1mm}
\subsection{Convexity vs. robustness of location-scale models}
\vspace{-1mm}

We must first determine which location-scale family to use, as this include Gaussians with mean $\nu$ and variance $\sigma^2$, Laplace with mean $\nu$ and scale $\sigma$, Cauchy with location $\nu$ and scale $\sigma$, and more.
To make this decision, we consider two desiderata:
(1) the prior should be robust in the way the Cauchy is robust, i.e. being wrong about the data location should not harm us too much, and (2) it should be easy to learn the parameters $\nu$ and $\sigma$, e.g. by optimizing $U_{\*x}^{(q)}(\mu_{\nu,\sigma})$.\looseness-1

While not necessary, one way of ensuring (2) is convexity of $U_{\*x}^{(q)}$, which we focus on as it enables efficient algorithms.
Here we make use of a connection between these upper bounds and the likelihood of {\bf censored regression} \citep{pratt1981concavity}, which for noise $\xi_i\in\R$ models a relationship between features $\*f_i\in\R^d$ and a variable $y_i=\langle\*w,\*f_i\rangle+\xi_i$ when information about $y_i$ is only provided in terms of an interval $[a_i,b_i)$ containing it (e.g. an individual's income bracket, not their exact income).
If $\xi_i$ is from a location-scale distribution with $\nu=0$ the log-likelihood given datapoints $(a_i,b_i,\*f_i)$ is\looseness-1\vspace{-2mm}
\begin{equation}\label{eq:censored}
	\ifdefined\arxiv
	\LogL_{\{a_i,b_i,\*f_i\}_{i=1}^n}(\*w,\sigma)
	=\hspace{-1mm}\sum_{i=1}^n\log\int_{a_i}^{b_i}\frac1\sigma f\left(\frac{y-\langle\*w,\*f_i\rangle}\sigma\right)dy
	\else
	\LogL_{\{a_i,b_i,\*f_i\}_{i=1}^n}\hspace{-.5mm}(\*w,\sigma)\hspace{-.5mm}
	=\hspace{-1mm}\sum_{i=1}^n\hspace{-.5mm}\log\hspace{-.5mm}\int_{a_i}^{b_i}\hspace{-.5mm}\frac1\sigma f\hspace{-.5mm}\left(\hspace{-.5mm}\frac{y-\langle\*w,\*f_i\rangle}\sigma\hspace{-.5mm}\right)\hspace{-.5mm}dy
	\fi
\end{equation}
Observe that for $a=\*x_{[\lfloor qn\rfloor]}$ and $b=\*x_{[\lfloor qn\rfloor]+1}$ we have\vspace{-1mm}
\ifdefined\arxiv
\begin{equation}
	U_{\*x}^{(q)}(\mu_{\langle\*w,\*f\rangle,\sigma})
	=-\log\mu_{\langle\*w,\*f\rangle,\sigma}((a,b])
	=-\log\int_a^b\frac1\sigma f\left(\frac{o-\langle\*w,\*f\rangle}\sigma\right)do
\end{equation}
\else
\begin{align}
	\begin{split}
		U_{\*x}^{(q)}(\mu_{\langle\*w,\*f\rangle,\sigma})
		&=-\log\mu_{\langle\*w,\*f\rangle,\sigma}((a,b])\\
		&=-\log\int_a^b\frac1\sigma f\left(\frac{o-\langle\*w,\*f\rangle}\sigma\right)do
	\end{split}
\end{align}
\fi
which is the negative of $\LogL_{a,b,\*f}(\*w,\sigma)$.
We thus adopt the reparameterization of \citet{burridge1981note}, who showed that  \eqref{eq:censored} is concave w.r.t. $(\*v,\phi)=(\frac{\*w}\sigma,\frac1\sigma)$ whenever $f$ is {\bf log-concave}, a property satisfied by the Gaussian and Laplace families but not the Cauchy.
Therefore, for such $f$ we have that $\ell_{\*x}^{(q)}(\langle\*v,\*f\rangle,\phi)=U_{\*x}^{(q)}(\mu_{\frac{\langle\*v,\*f\rangle}\phi,\frac1\phi})$ is convex w.r.t. $(\*v,\phi)$.\looseness-1

Unfortunately, we show that no log-concave $f$ is robust, in the sense that for any $R>0$ there exists a dataset of points in the interval $(\theta\pm R)^n$ s.t. $U_{\*x}^{(q)}(\mu_{\theta,1})=\Omega(R)$ (rather than $\BigO(\log(1+R^2))$ as shown for the Cauchy family in Corollary~\ref{cor:cauchy}).
On the other hand, log-concave location-scale families are the only ones for which $U_{\*x}^{(q)}$ is convex, both for the original parameterization and that of \citet{burridge1981note}.
We record these facts in the following theorem:\looseness-1
\begin{Thm}[c.f. Thm.~\ref{thm:impossibility}]
	Let $\mu_{\nu,\sigma}$ be a location-scale family associated with a continuous measure $f:\R\mapsto\R_{\ge0}$.
	\ifdefined\arxiv
	\begin{enumerate}[itemsep=1pt]
	\else
	\begin{enumerate}[leftmargin=*,topsep=-2pt,noitemsep]\setlength\itemsep{1pt}
	\fi
		\item If $f$ is log-concave then $\exists~a,b>0$ s.t. for any $R>0$, $\psi\in(0,\frac R{2n}]$, $q\ge\frac1n$, and $\theta\in\R$ there exists $\*x\in(\theta\pm R)^n$ with $\min_i\*x_{[i+1]}\hspace{-.25mm}-\hspace{-.25mm}\*x_{[i]}\hspace{-.25mm}=\hspace{-.25mm}\psi$ s.t. $U_{\*x}^{(q)}(\mu_{\theta,1})\hspace{-.5mm}=\hspace{-.5mm}aR+\log\frac b\psi$.
		\item If $f$ is not log-concave then there exists $\*x\in\R^n$ with $\min_i\*x_{[i+1]}\hspace{-.5mm}-\hspace{-.5mm}\*x_{[i]}\hspace{-.25mm}>\hspace{-.25mm}0$ s.t. $U_{\*x}^{(q)}(\mu_{\theta,1})$ is non-convex in $\theta$.\looseness-1
	\end{enumerate}
\end{Thm}
Note the latter dataset is not degenerate:
for $f$ strictly log-convex over $[a,b]$, any $\*x$ whose optimal interval has length $<\frac{b-a}2$ has non-convex $U_{\*x}^{(q)}(\mu_{\theta,1})\hspace{-.5mm}=\hspace{-.5mm}-\log\Psi_{\*x}^{(q)}(\mu_{\theta,1})$.
 
We must thus choose between having a robust location-scale family like the Cauchy or an easy-to-optimize log-concave one.
As we can ensure robustness of the learned prior {\em post-hoc} using the approach of Section~\ref{sec:robustness}, we choose the latter.
Specifically, we use the Laplace prior, as it is in some sense the most robust log-concave distribution (it has loss $\Theta(R)$ if $\*x\in(\theta\pm R)^n$, whereas e.g. the Gaussian has loss $\Theta(R^2)$) and because it yields a numerically stable closed-form expression~\eqref{eq:closed-form} for $\ell_{\*x}^{(q)}(\theta,\phi)$ (unlike e.g. the Gaussian).\looseness-1

\vspace{-1mm}
\subsection{Augmenting quantile release using public data}
\vspace{-1mm}

We turn to two applications that depend on optimizing upper bounds $\ell_{\*x}^{(q)}(\theta,\phi)$ on the performance of quantile release using the Laplace prior with scale $\frac1\phi$ and location $\frac\theta\phi$.
While our final objective is small $\Gap_q$, we will mainly discuss optimizing $\ell_{\*x}^{(q)}=U_{\*x}^{(q)}$, or its expectation if $\*x$ is drawn from some distribution.
In the former case this directly bounds (w.h.p.) the cost of multiple quantile release via the theoretical results in Section~\ref{sec:algorithm} because $U_{\*x}\ge-\log\Psi_{\*x}^{(\varepsilon)}$, while a bound on $\E_{\*x}U_{\*x}$ can bound $\E\Gap_{\max}$ by setting $\beta$.
 For example,
$\beta=\frac{2\pi^2}{\varepsilon n}\exp(2\sqrt{\log(2)\log(m+1)})$ in Theorem~\ref{thm:kary-main} implies $\Gap_{\max}$ has expectation at most\vspace{-2mm}
\begin{equation}\label{eq:expectation-bound}
	\BigO\left(\hspace{-.75mm}\exp\hspace{-.25mm}\left(2\sqrt{\log(2)\log(m\hspace{-.25mm}+\hspace{-.25mm}1)}\right)\hspace{-.5mm}\frac{\log(\varepsilon mn)\hspace{-.5mm}+\hspace{-.5mm}\E_{\*x}U_{\*x}}\varepsilon\hspace{-.5mm}\right)\vspace{-1mm}
\end{equation}
Our first application is the frequently studied setting where we have a large public dataset $\*x'\in\R^N$ and want to use it to improve the release of statistics of a smaller private dataset $\*x\in\R^n$.
To apply our quantile release method, we must use $\*x'$ to construct a prior $\mu'$ for each that makes $U_{\*x}^{(q)}(\mu')$ small.
If the entries of $\*x$ and $\*x'$ are sampled i.i.d. from similar distributions $\D$ and $\D'$, respectively, the convexity of $U_{\*x}^{(q)}$ suggests using stochastic optimization find a prior $\mu$ that approximately minimizes the expectation $\E_{\*z\sim{\D'}^n}U_{\*z}(\mu)$ using samples of size $n$ drawn from $\*x'$.
We provide a guarantee for a variant of this generic approach that runs online gradient descent (OGD) with separate learning rates for $\theta$ and $\phi$ on samples drawn without replacement from $\*x'$:
\begin{Thm}[c.f. Thm.~\ref{thm:pubpri}]
	If $\D$ and $\D'$ have bounded densities with bounded support then there exists an algorithm optimizing $U_{\*x_t'}$ over $T$ datasets $\*x_t'$ of size $n$ drawn from $\*x'\in\R^N$ without replacement that runs in time $\BigO(mN)$ and returns a set $\mu'$ of $m$ Laplace priors s.t. w.h.p.\vspace{-2mm}
	\ifdefined\arxiv
	\begin{equation}
		\E_{\*x\sim\D^n}U_{\*x}(\mu')
		\le\min_{\mu\in\Lap_{B,\sigma_{\min},\sigma_{\max}}^m}\E_{\*x\sim\D^n}U_{\*x}(\mu)
		+\tilde\BigO\left(\TV_q(\D,\D')+\sqrt{\frac{mn}N}\right)
	\end{equation}
	\else
	\begin{align}
	\begin{split}
		\E_{\*x\sim\D^n}U_{\*x}(\mu')
		&\le\min_{\mu\in\Lap_{B,\sigma_{\min},\sigma_{\max}}^m}\E_{\*x\sim\D^n}U_{\*x}(\mu)\\
		&\qquad+\tilde\BigO\left(\TV_q(\D,\D')+\sqrt{\frac{mn}N}\right)\vspace{-4mm}
	\end{split}
	\end{align}
	\fi
	where $\Lap_{B,\sigma_{\min},\sigma_{\max}}$ is the set of Laplace priors with locations in $[\pm B]$ and scales in $[\sigma_{\min},\sigma_{\max}]$ and $\TV_q(\D,\D')$ is the total variation distance between the joint distributions of the order statistics $\left\{(\*x_{[\lfloor q_in\rfloor]},\*x_{[\lfloor q_in\rfloor+1]})\right\}_{i=1}^m$ for $\*x\sim\D^n$ and $\left\{(\*x_{[\lfloor q_in\rfloor]}',\*x_{[\lfloor q_in\rfloor+1]}')\right\}_{i=1}^m$ for $\*x'\sim{\D'}^n$.
\end{Thm}
For $N\gg mn$, the suboptimality of $\mu'$ for the upper bound $U_{\*x}$ will depend on the statistical distance between the quantile intervals of $\D$ and $\D'$:
even if $\D$ and $\D'$ are dissimilar, similar order statistic distributions will ensure good performance.
Note, as in Section~\ref{sec:robustness}, we can hedge against large $\TV_q(\D,\D')$ by mixing the output $\mu'$ with a robust prior.\looseness-1

We evaluate this approach, which we call {\em Public Fit} or \texttt{PubFit}, on Adult~\citep{kohavi1996adult} and Goodreads~\citep{wan2018goodreads}, both used previously for DP quantiles~\citep{gillenwater2021differentially,kaplan2022quantiles}.
Because our guarantees improve with different step-sizes for $\theta$ and $\phi$, we use COCOB~\citep{orabona2017training}---an OGD variant that provably sets per-coordinate step-sizes without the need for tuning---as \texttt{PubFit}'s stochastic solver. 
We also test a robust version where its output is mixed with a half-Cauchy distribution, and three baselines:
the Uniform prior, just using the quantiles of the public data (\texttt{public quantiles}), and using the public quantiles to set the location parameters of $m$ Cauchy priors (\texttt{public Cauchy}).\looseness-1

\begin{figure}[!t]
	\centering
	\includegraphics[width=.495\linewidth]{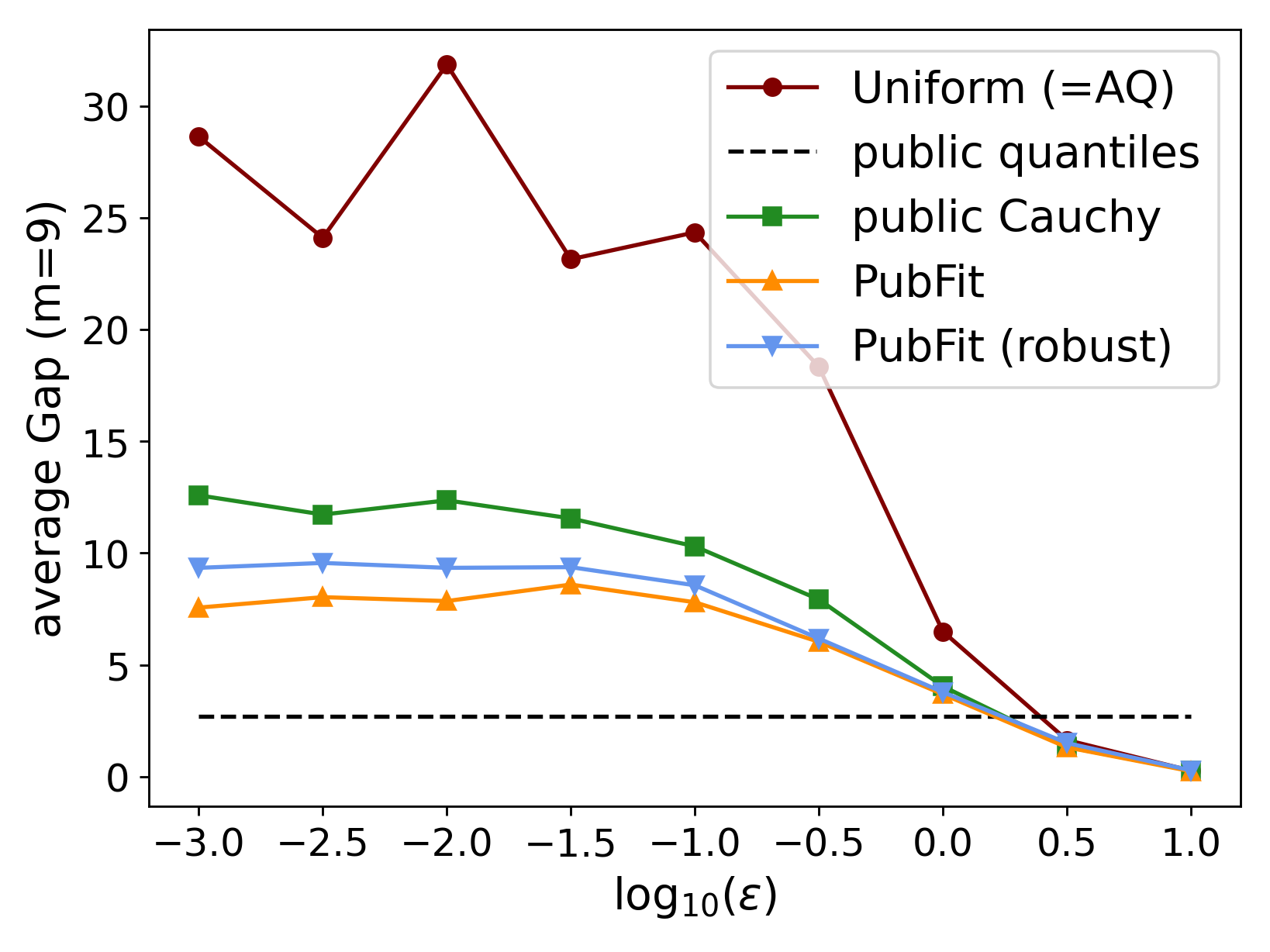}
	\hfill
	\includegraphics[width=.495\linewidth]{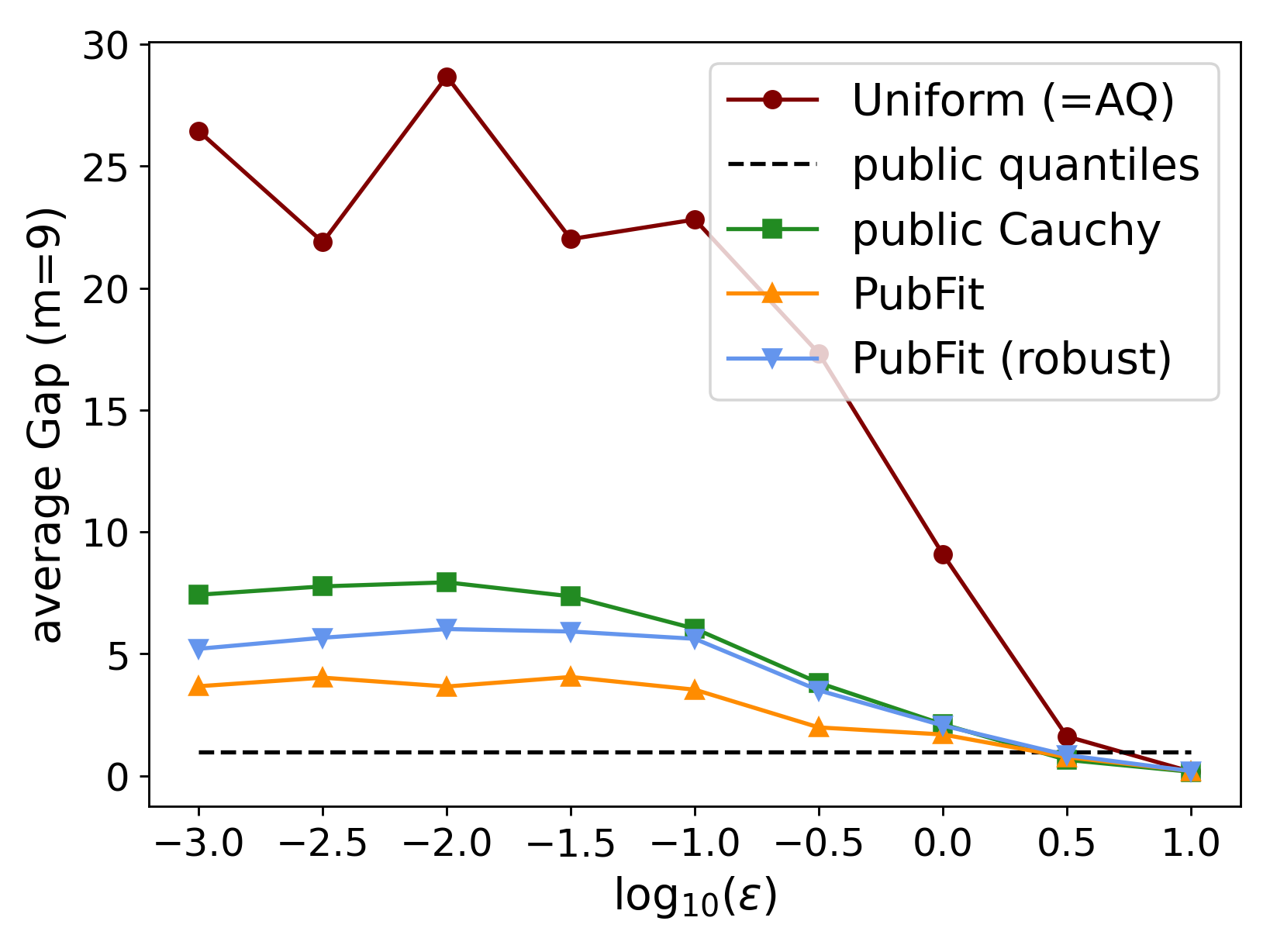}\vspace{-4mm}
	\caption{\label{fig:pubpri-logepsplots}
		Public-private release of nine quantiles using one hundred samples from the Adult age (left) and hours (right) datasets.
		The public data is the Adult training set while private data is test.\looseness-1
	}
\end{figure}

\ifdefined\arxiv\else
\begin{figure}[!t]
	\centering\vspace{-2mm}
	\includegraphics[width=.495\linewidth]{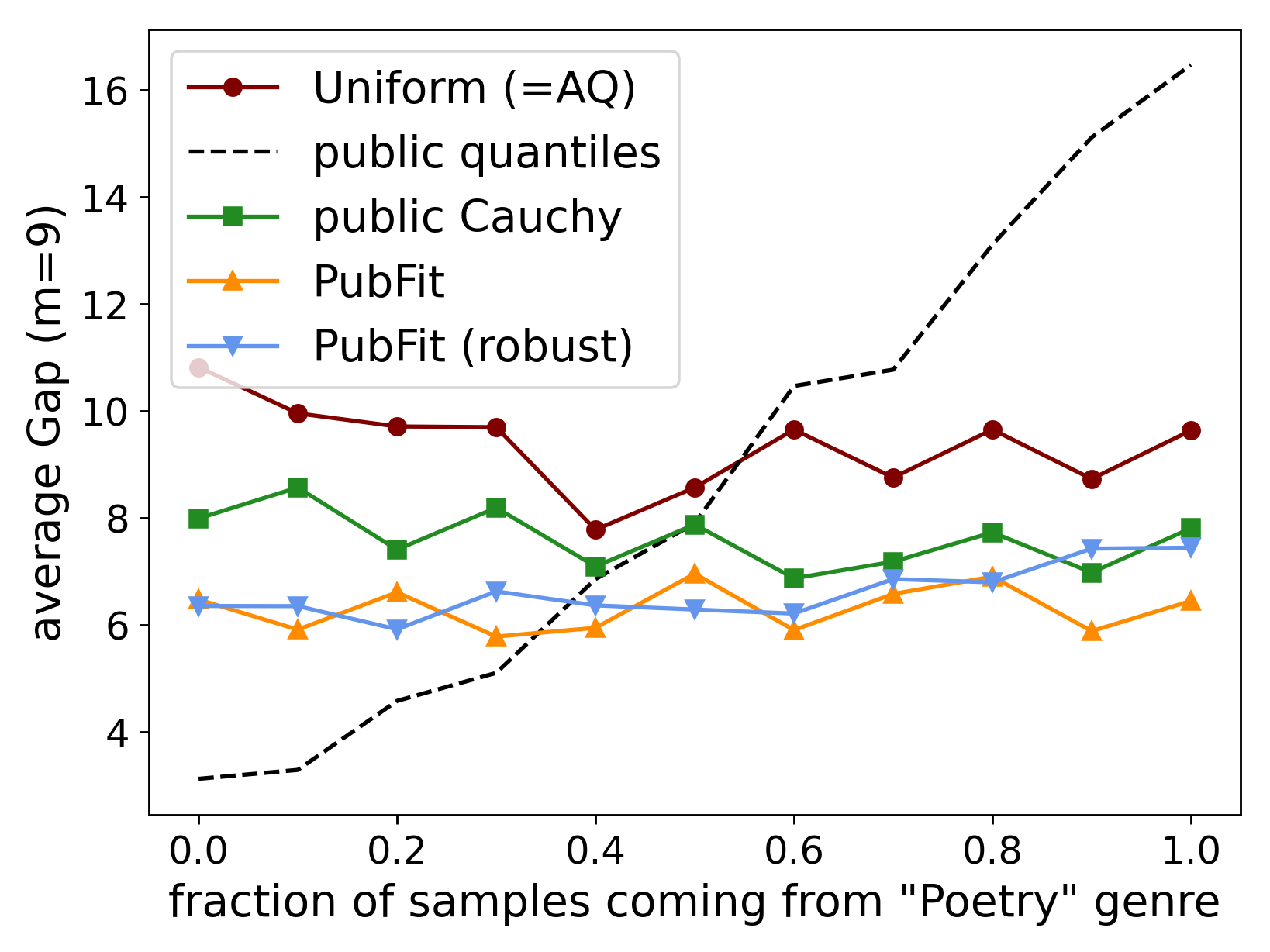}
	\hfill
	\includegraphics[width=.495\linewidth]{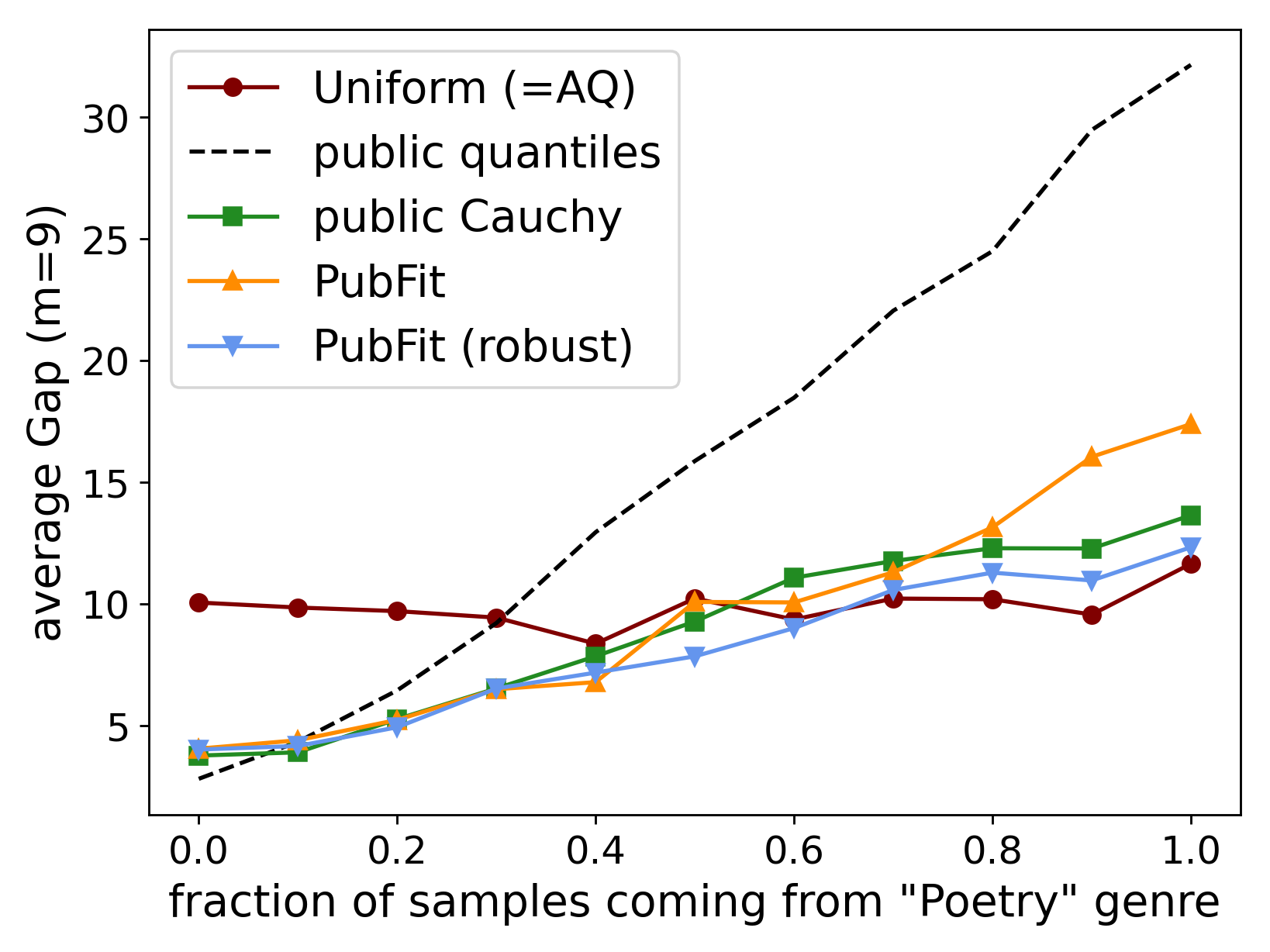}\vspace{-4mm}
	\caption{\label{fig:mixplots}
		Public-private release of nine quantiles on one hundred samples from the Goodreads rating (left) and page count (right) datasets, with $\varepsilon=1$.
		The public data is the ``History" genre while private data is sampled from a mixture of it and ``Poetry."\looseness-1
	}
\end{figure}
\fi

Adult tests the $\D=\D'$ case, with its ``train" set the public dataset and a hundred samples from ``test" as private.
Figure~\ref{fig:pubpri-logepsplots} shows that \texttt{public quantiles} does best at small $\varepsilon$, as is expected with no distribution shift, but it cannot adapt to the empirical distribution of a small number of private points, and so is worse at $\varepsilon>1$.
Among the rest, \texttt{PubFit} is most similar to \texttt{public-quantiles} at small $\varepsilon$ but still does well at large $\varepsilon$.\looseness-1

We use the Goodreads ``History" and ``Poetry" genres to evaluate under distribution shift by fitting on all but a small fraction of data from the former and releasing quantiles of samples from varying mixtures of the two datasets. 
As expected, the performance of \texttt{public quantiles} deteriorates with more samples from ``Poetry."
For book ratings, \texttt{PubFit} is best among the remaining methods, but without much change with distribution shift, possibly due to an incomplete fit of the data.
For page counts, the \texttt{PubFit} methods and \texttt{public Cauchy} both do as well as \texttt{public-quantiles} when most data is from ``History," but \texttt{PubFit (robust)} deteriorates least---and much less than regular \texttt{PubFit}---as the distribution shifts. 
This highlights the importance of robustness analysis, and suggest the former as a good method to start with, as it takes advantage of similar public and private distributions (Fig.~\ref{fig:pubpri-logepsplots}) while never doing much worse than the default method (Uniform) when the the distributions are dissimilar (Fig.~\ref{fig:mixplots}).\looseness-1

\vspace{-1mm}
\subsection{Sequentially setting priors using past sensitive data}
\vspace{-1mm}

\ifdefined\arxiv
\begin{figure}[!t]
	\centering\vspace{-2mm}
	\includegraphics[width=.495\linewidth]{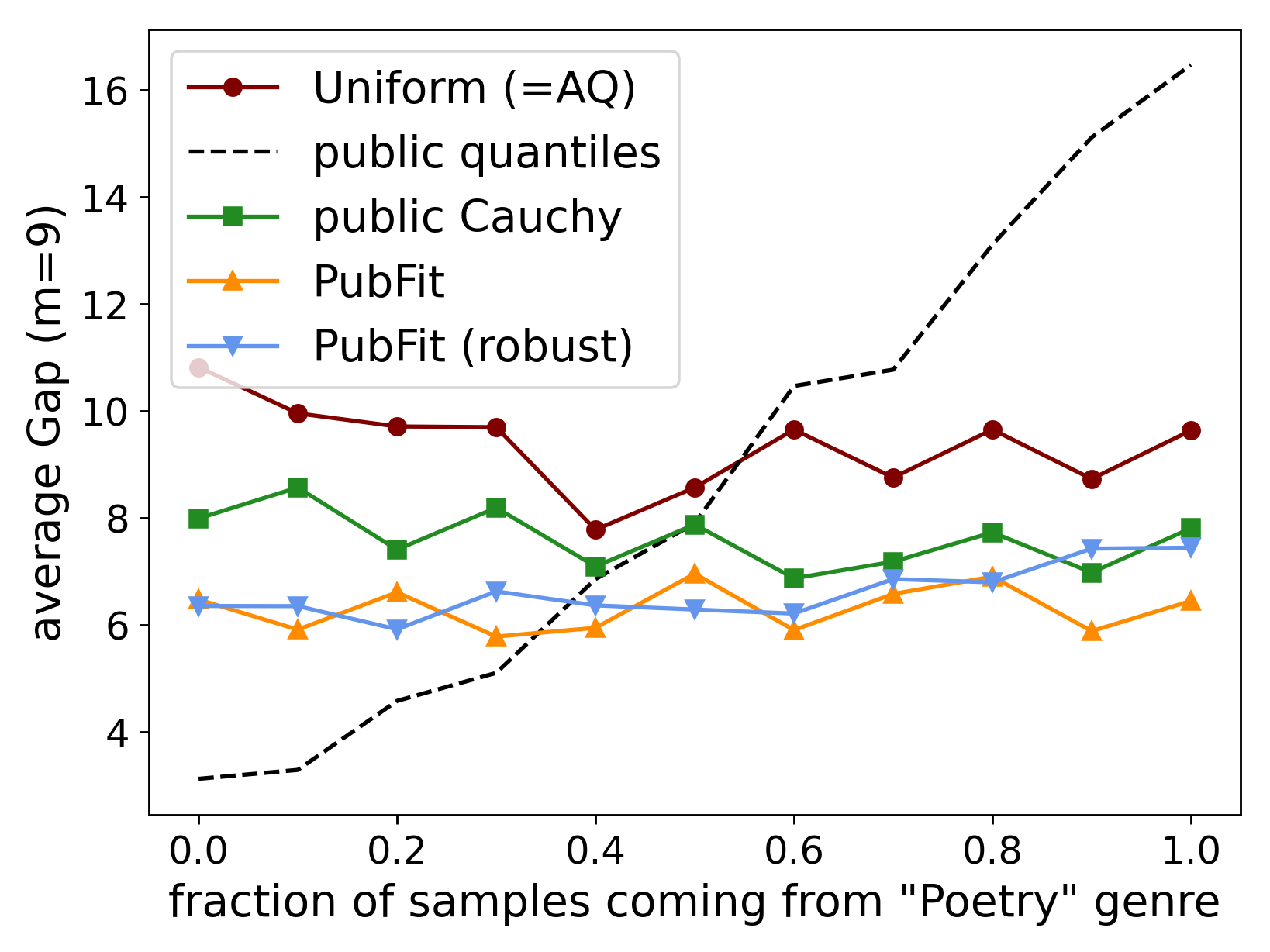}
	\hfill
	\includegraphics[width=.495\linewidth]{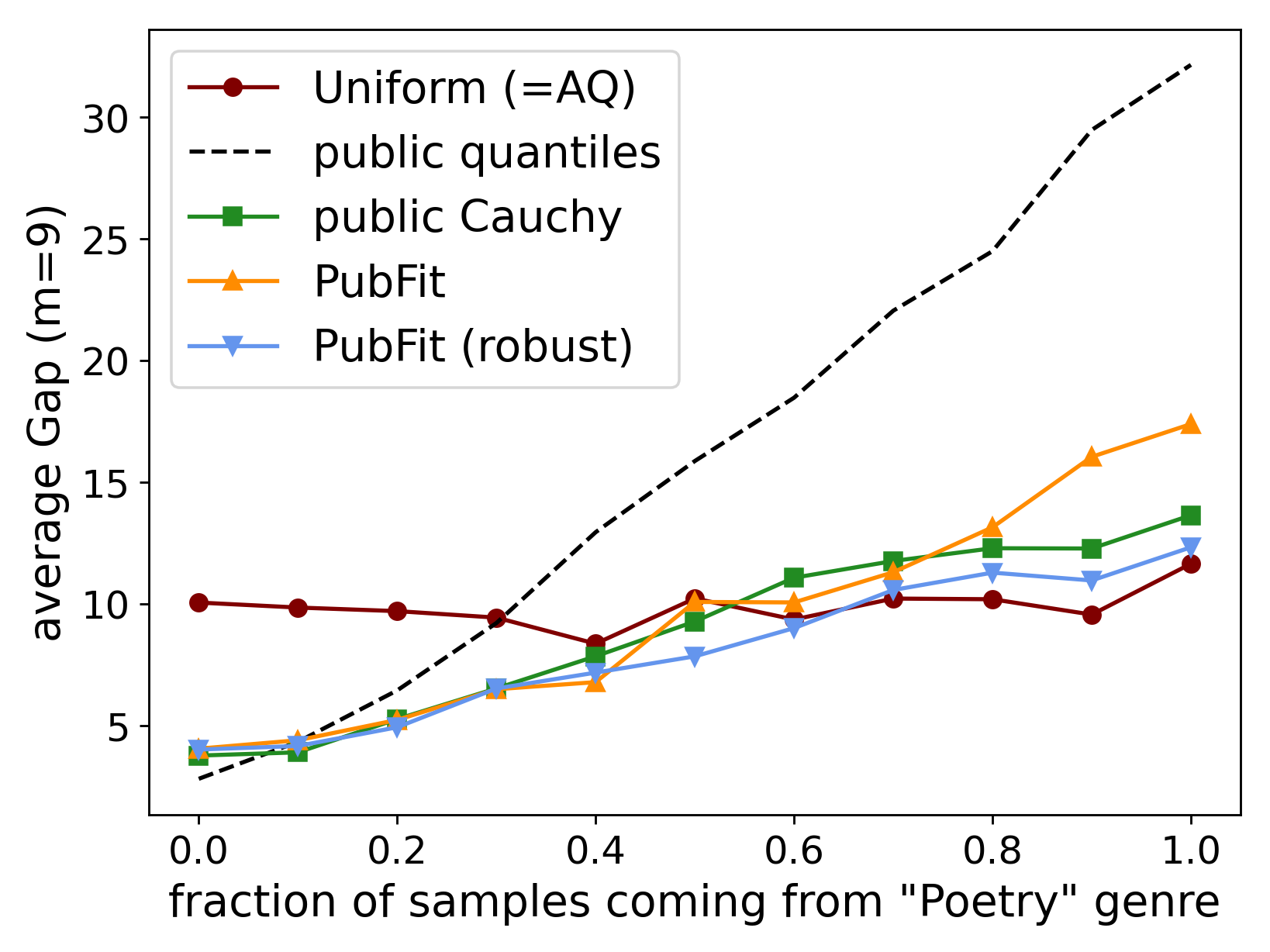}\vspace{-4mm}
	\caption{\label{fig:mixplots}
		Public-private release of nine quantiles on one hundred samples from the Goodreads rating (left) and page count (right) datasets, with $\varepsilon=1$.
		The public data is the ``History" genre while private data is sampled from a mixture of it and ``Poetry."\looseness-1
	}
\end{figure}
\else
\fi

Our second application is sequential release, which we do not believe has been studied, but arises naturally if e.g. we wish to release daily statistics from a continuous stream of data.
Here we have a sequence of datasets $\*x_1,\dots,\*x_T$, each with associated {\em public} features $\*f_1,\dots,\*f_T\in\R^d$ (e.g. day of the week), and we wish to minimize the average maximum gap $\frac1T\sum_{t=1}^T\max_i\Gap_{q_i}(\*x_t,o_{t,i})$, whose expectation can be bounded~\eqref{eq:expectation-bound} in terms of $\frac1T\sum_{t=1}^TU_{\*x_t}$.
For simplicity, we assume individuals do not occur in multiple datasets $\*x_t$, e.g. we are releasing the median age of new users of a service.
Note the natural way to avoid this assumption is to compose the privacy budgets at each time;
empirically our methods are especially useful in the low privacy regime this entails.\looseness-1

\ifdefined\arxiv\newpage\else\fi
Our analysis suggests that we can apply online learning here, e.g. doing the following at each $t$ starting with a prior $\mu_1$:\looseness-1
\ifdefined\arxiv
\begin{enumerate}[itemsep=1pt]
\else
\begin{enumerate}[leftmargin=*,topsep=-2pt,noitemsep]\setlength\itemsep{1pt}
\fi
	\item release $o_t$ using the prior $\mu_t$ and suffer $\Gap_q(\*x_t,o_t)$
	\item update to $\mu_{t+1}$ using online learning on the loss $\ell_{\*x_t}^{(q)}$
\end{enumerate}
Because $\ell_{\*x_t}^{(q)}(\theta,\phi)=U_{\*x_t}^{(q)}(\mu_{\frac\theta\phi,\frac1\phi})$ is convex for Laplace priors, online convex optimization (OCO)~\citep{shalev-shwartz2011oco} lets us compete with the best prior in hindsight according to the upper bounds $U_{\*x_t}^{(q)}(\mu_t)$, or with the best linear map $\*w$ to locations $\langle\*w,\*f_t\rangle$.
We can again hedge against poor predictions by mixing with a constant robust distribution.\looseness-1

However, we face the difficulty that online learning on losses $\ell_{\*x_t}^{(q)}$ leaks information about $\*x_t$.
There are two natural solutions. 
One is to use part of the budget $\varepsilon'<\varepsilon$ on a DP online learner~\citep{jain2012dp,smith2013optimal} and hope that the reduction in budget allocated to quantile release is made up for by the improved priors.
Alternatively, we can replace $\ell$ with a {\em proxy} loss $\hat\ell$ that does not depend on the data and optimize it using regular OCO.
The first can be done with provable guarantees by applying DP-FTRL~\citep{kairouz2021practical}, again using two different step-sizes:\looseness-1
\begin{Thm}[c.f. Thm.~\ref{thm:sequential}]
	Consider a sequence of datasets $\*x_t\in[\pm B]^{n_t}$ with bounded features $\*f_t$ and suppose we set Laplace priors $\mu_{t,i}=\mu_{\frac{\langle\*v_{t,i},\*f_t\rangle}{\phi_{t,i}},\frac1{\phi_{t,i}}}$ via two DP-FTRL algorithms applied separately to the variables $\*v_i$ and $\phi_i$ of the losses $\ell_{\*x_t}(\langle\*v_i,\*f_t\rangle,\phi_i)$ with budgets $\frac{\varepsilon'}2$, with respective step-sizes $\tilde\Theta\left(\sqrt{\frac{\varepsilon'}{\sigma_{\min}^2T}\sqrt{\frac md}}\right)$ and $\tilde\Theta\left(\sqrt{\frac{\varepsilon'\sqrt m}{\sigma_{\min}^2\sigma_{\max}^2T}}\right)$.
	This is $(\varepsilon',\delta')$-DP and w.h.p. has regret\vspace{-1mm}
	\ifdefined\arxiv
	\begin{equation}
		\frac1T\sum_{t=1}^TU_{\*x_t}(\mu_t)-\min_{\begin{smallmatrix}\*w_i\in[\pm B]^d\\\sigma_i\in[\sigma_{\min},\sigma_{\max}]\end{smallmatrix}}\frac1T\sum_{t=1}^TU_{\*x_t}(\mu_{\langle\*w_i,\*f_t\rangle,\sigma_i})
		=\tilde\BigO\left(\frac{d^\frac34+\sigma_{\max}}{\sigma_{\min}}\sqrt{\frac m{\varepsilon'T}\sqrt{m\log\frac2{\delta'}}}\right)
	\end{equation}
	\else
	\begin{align}
	\begin{split}
	\frac1T\sum_{t=1}^TU_{\*x_t}(\mu_t&)-\hspace{-2mm}\min_{\begin{smallmatrix}\*w_i\in[\pm B]^d\\\sigma_i\in[\sigma_{\min},\sigma_{\max}]\end{smallmatrix}}\hspace{-.5mm}\frac1T\sum_{t=1}^TU_{\*x_t}(\mu_{\langle\*w_i,\*f_t\rangle,\sigma_i})\\
	&=\tilde\BigO\left(\frac{d^\frac34+\sigma_{\max}}{\sigma_{\min}}\sqrt{\frac m{\varepsilon'T}\sqrt{m\log\frac2{\delta'}}}\right)\vspace{-3mm}
	\end{split}
	\end{align}
	\fi
\end{Thm}
\ifdefined\arxiv\newpage\else\fi
Thus we can do as well as any sequence of Laplace priors $\mu_t$ with locations determined by a fixed linear map from $\*f_t$, up to a term that decreases at rate $\tilde\BigO(\frac1{\sqrt T})$.
Furthermore, running quantile release with budget $\varepsilon-\varepsilon'$ ensures $(\varepsilon,\delta')$-DP for each dataset $\*x_t$.
Note that using different step-sizes allows us to separate the difficulty of learning a $d$-dimensional linear map from the difficulty of learning a scale parameter of magnitude at most $\sigma_{\max}$.\looseness-1

Unfortunately, DP-FTRL is too noisy to learn competitive priors, except with a lot of stationary data~(c.f. Fig.~\ref{fig:timeplots} (left)).
One issue is that its DP guarantee is too strong, as it
it allows swapping out the entire dataset $\*x_t$ rather than a single entry.
It is unclear if a better sensitivity is possible for $U_{\*x_t}$, as changing an entry can flip the sign of the gradient while preserving magnitude.
We show (c.f. Lem.~\ref{lem:refined-sensitivity}) that it is possible for the $\varepsilon$-dependent bound $U_{\*x_t}^{(\varepsilon)}$ over piecewise-constant priors---remarkably sensitivity {\em decreases} with $\varepsilon$---but that upper bound is non-convex for location-scale families, which are preferable for model learning.\looseness-1

\ifdefined\arxiv
\begin{figure}[!t]
	\centering
	\includegraphics[width=.495\linewidth]{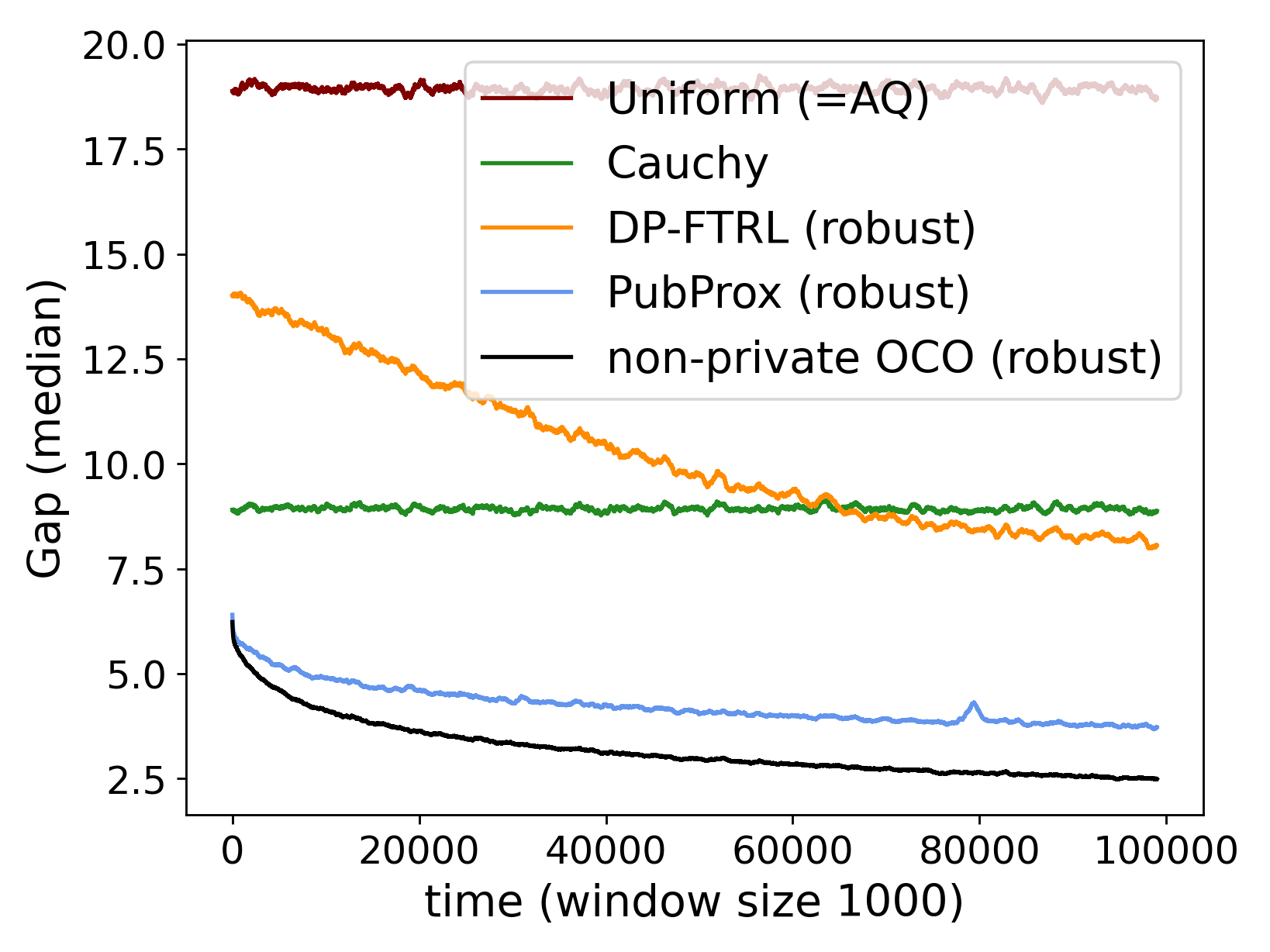}
	\hfill
	\includegraphics[width=.495\linewidth]{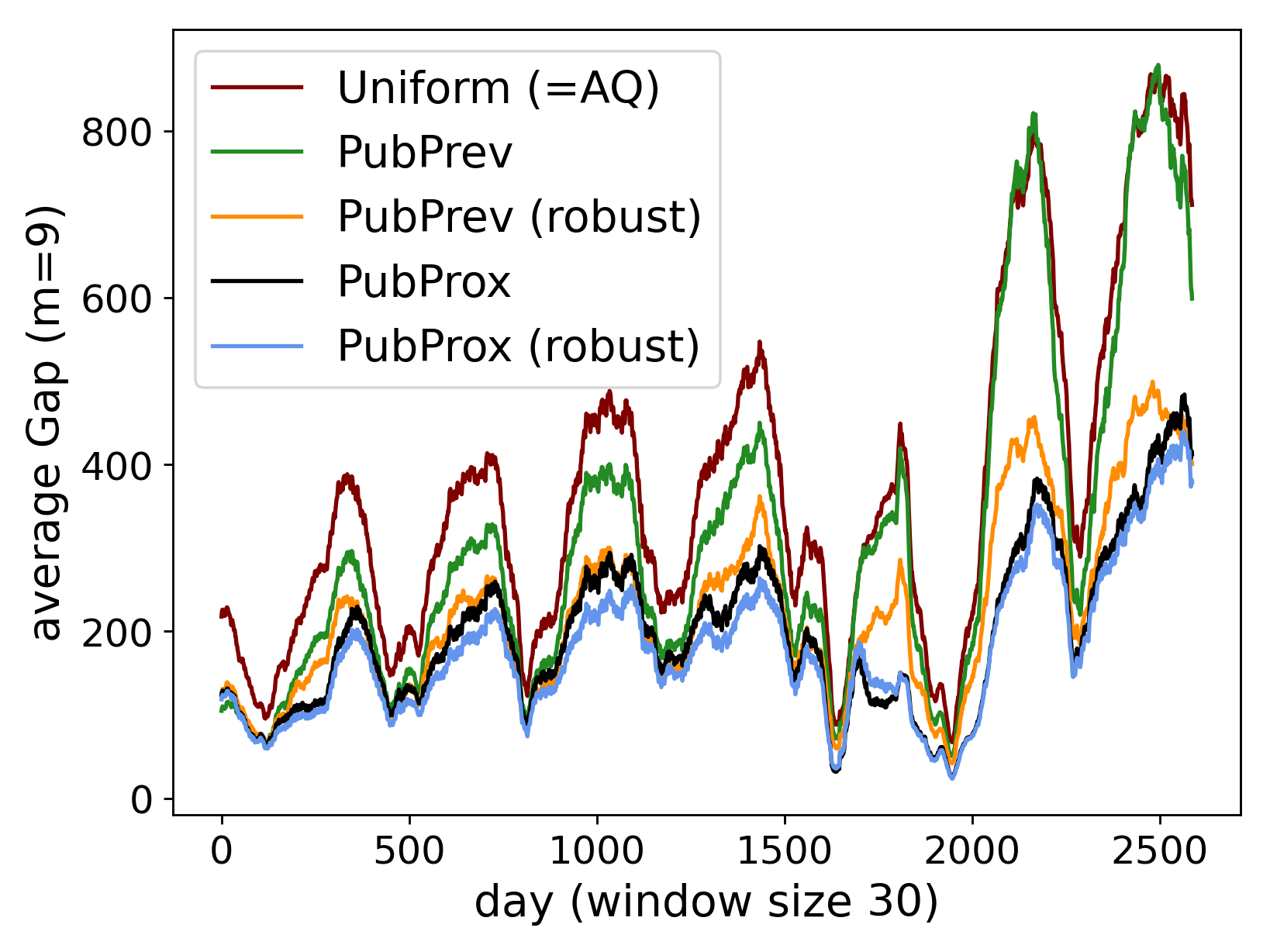}\vspace{-4mm}
	\caption{\label{fig:timeplots}
		Comparison of sequential release over time on Synthetic (left, $\log_{10}\varepsilon=-1/2$) and CitiBike (right, $\log_{10}\varepsilon=-2$) tasks.
	}
\end{figure}
\else
\fi

Our second solution involves recognizing that $U_{\*x_t}^{(q)}$ depends only on the optimal interval $[\*x_{t[\lfloor qn\rfloor]},\*x_{t[\lfloor qn\rfloor+1]})$, whose location and size we have (public) estimates for:
the former via the quantile estimate $o_t$ and the size is lower-bounded by the underlying data discretization, which we have access to in-practice (e.g. age is reported in years, bicycle trip length in seconds).
We use this information to construct {\em proxy} losses $\hat\ell_{o_t}^{(q)}(\langle\*v,\*f_t\rangle,\phi)$, which do not depend on $\*x_t$ and so be learned with (standard) OCO.
As our DP-FTRL analysis again showed the importance of different step-sizes, we again use the COCOB optimizer here.\looseness-1

We evaluate sequential release on three online tasks, each consisting of a sequence of datasets needing quantiles:
\ifdefined\arxiv
\begin{enumerate}[itemsep=1pt]
\else
\begin{enumerate}[leftmargin=*,topsep=-2pt,noitemsep]\setlength\itemsep{1pt}
\fi
	\item Synthetic: each dataset is generated such that the quantiles are fixed linear functions of a random Gaussian feature vector, plus noise.
	\item CitiBike: the data are the lengths of a day's bicycle trips, with the date and NYC weather information features.
	\item BBC: the data are the Flesch readability scores of the comments on a headline posted to Reddit's \texttt{worldnews} forum, with date and headline text information features.
\end{enumerate}

\ifdefined\arxiv\newpage\else\fi
In addition to the proxy approach, which we call \texttt{PubProx}, we evaluate static priors---the uniform, Cauchy, and half-Cauchy (if nonnegative)---and an approach we call \texttt{PubPrev}, which uses a Laplace prior centered around the previous step's released quantile.
Note that using the Uniform is equivalent to  \texttt{ApproximateQuantiles} (AQ).
For both \texttt{PubProx} and \texttt{PubPrev} we ensure robustness by mixing with a Cauchy (or half-Cauchy, if nonnegative) distribution with coefficient 0.1;
this nearly always improves performance for these methods, likely by ensuring their training data is not too noisy.
To see its effectiveness, note how in Figure~\ref{fig:timeplots} (right) both augmented methods are almost always better when made robust, especially \texttt{PubPrev};
in fact, non-robust \texttt{PubPrev} is unable to do better than Uniform after around day 1600, when the start of the COVID-19 pandemic significantly affects bicycle trips.

Our main comparisons is time-aggregated performance as a function of $\varepsilon$ (c.f. Figs.~\ref{fig:online-logepsplots} and~\ref{fig:worldnews-logepsplots}).
All except perhaps Synthetic demonstrate significant improvement by our methods over the Uniform (AQ) baseline, especially at small $\varepsilon$.
On Synthetic and CitiBike, both tasks with features for which a linear model should provide some benefit, we see in Figure~\ref{fig:online-logepsplots} that \texttt{PubProx} is indeed the best across all except perhaps the lowest privacy settings.
For BBC, Figure~\ref{fig:worldnews-logepsplots} reveals a large difference between mean and median performance (note the difference in y-axis scales), with \texttt{PubProx} doing best for the typical headline but the Cauchy doing better on-average due to better performance on headlines with many comments.
The result suggests that in highly noisy settings, the learning-based scheme should help, but it might not overcome the robustness of a static Cauchy prior in-expectation.\looseness-1

Overall, the results demonstrate the strength of the Cauchy and half-Cauchy priors, both as unbounded substitutes for the Uniform and as a means of robustifying learning-augmented algorithms.
They also demonstrate the utility of our upper bound in providing an objective for learning, albeit using proxy data rather the DP online learning:
\texttt{PubProx} usually does better than \texttt{PubPrev} despite using the same information.
Overall, \texttt{PubProx} performs the best at most privacy levels in all evaluation settings (Synthetic, CitiBike, and BBC) except when the mean is used as the metric for BBC (Fig.~\ref{fig:worldnews-logepsplots}, left), where it does almost as well as the best.
Narrowing the performance gap with non-private OCO (c.f. Fig.~\ref{fig:timeplots} (left), where we run COCOB directly on $\ell_{\*x_t}^{(q)}$)---remains an important research direction.\looseness-1

Code to reproduce our results is available at \url{https://github.com/mkhodak/private-quantiles}.

\ifdefined\arxiv\else
\begin{figure}[!t]
	\centering
	\includegraphics[width=.495\linewidth]{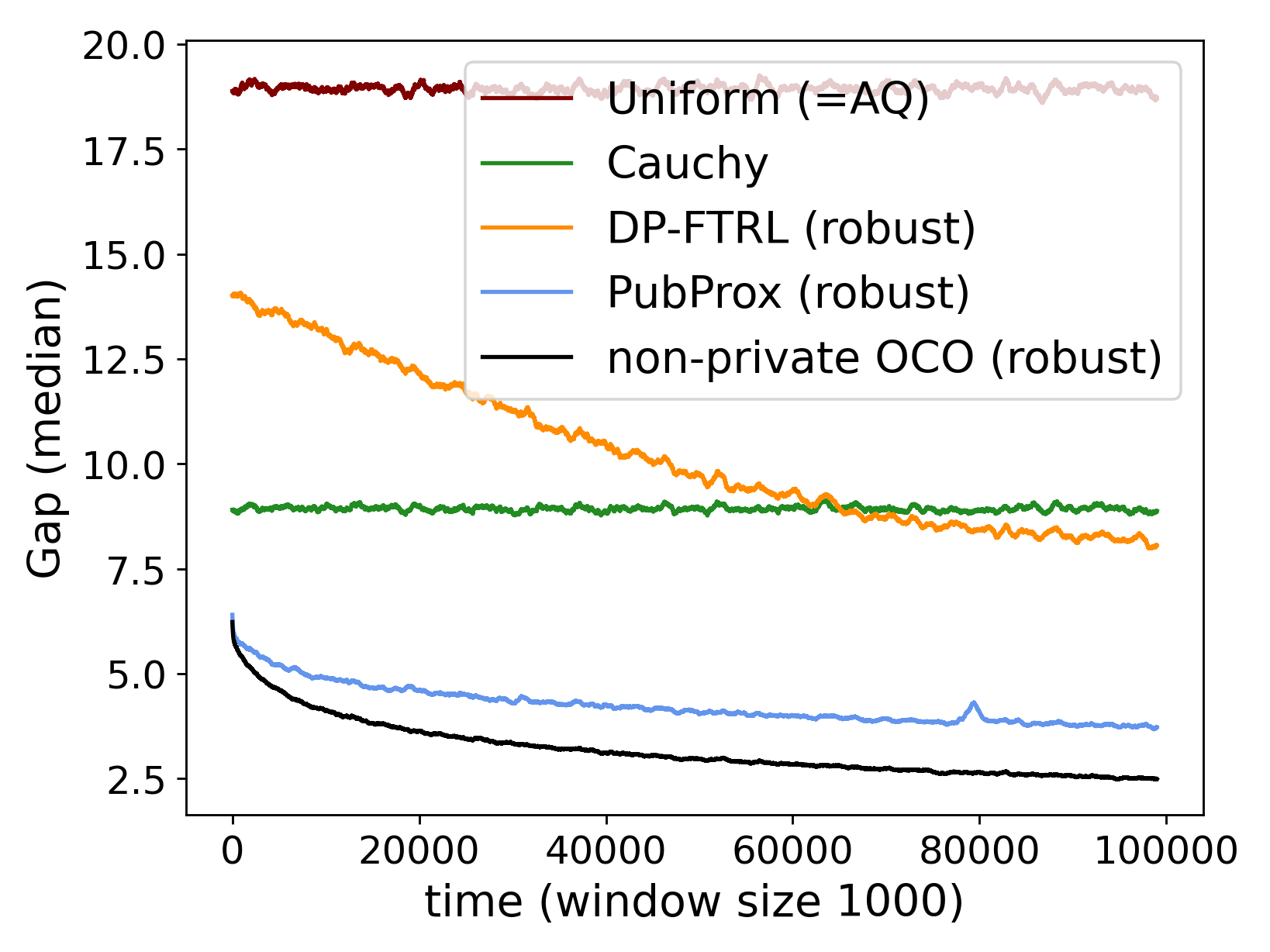}
	\hfill
	\includegraphics[width=.495\linewidth]{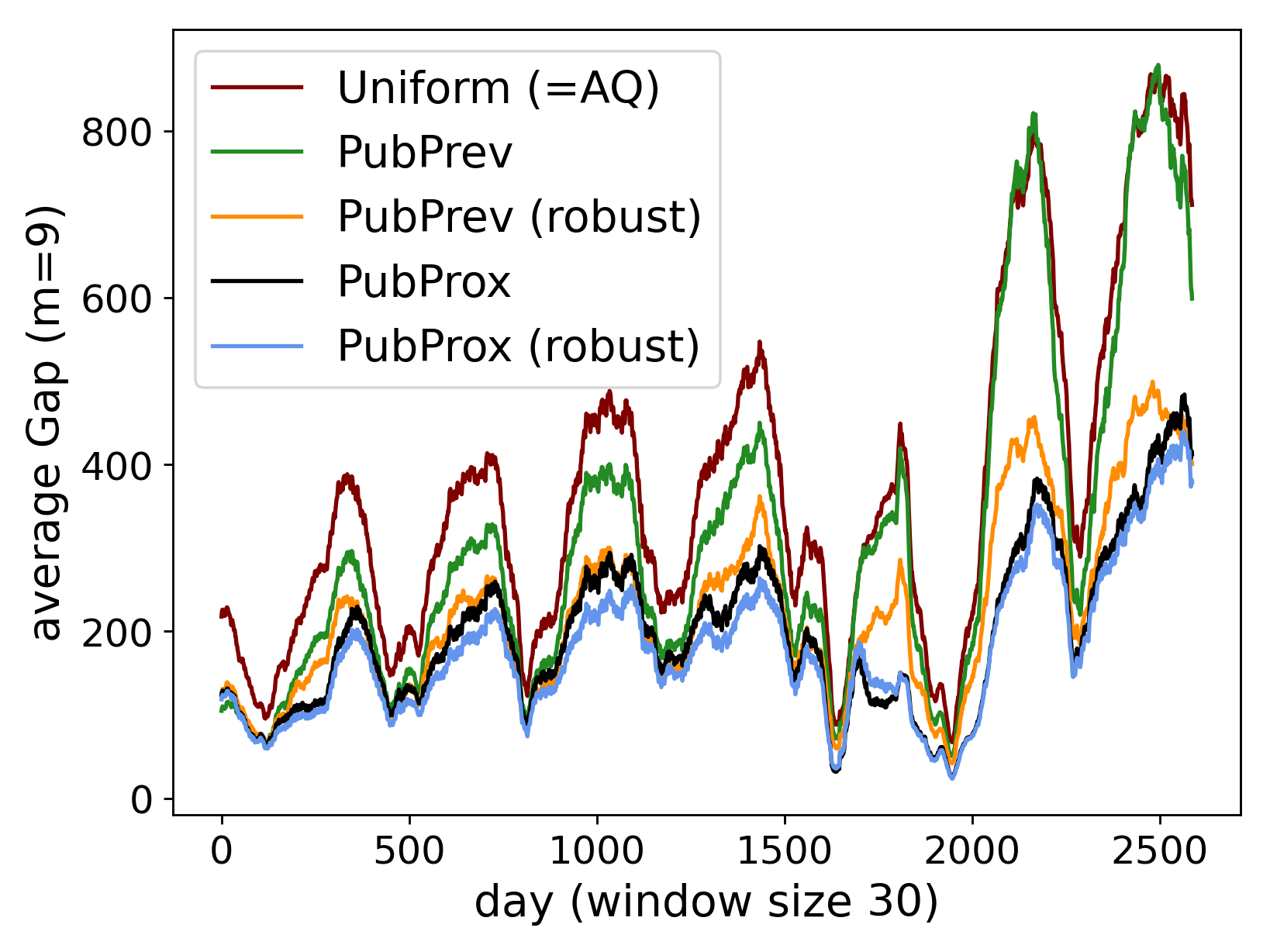}\vspace{-4mm}
	\caption{\label{fig:timeplots}
		Comparison of sequential release over time on Synthetic (left, $\log_{10}\varepsilon=-1/2$) and CitiBike (right, $\log_{10}\varepsilon=-2$) tasks.
	}
\end{figure}
\fi

\begin{figure}[!t]
	\centering\vspace{-2mm}
	\includegraphics[width=.495\linewidth]{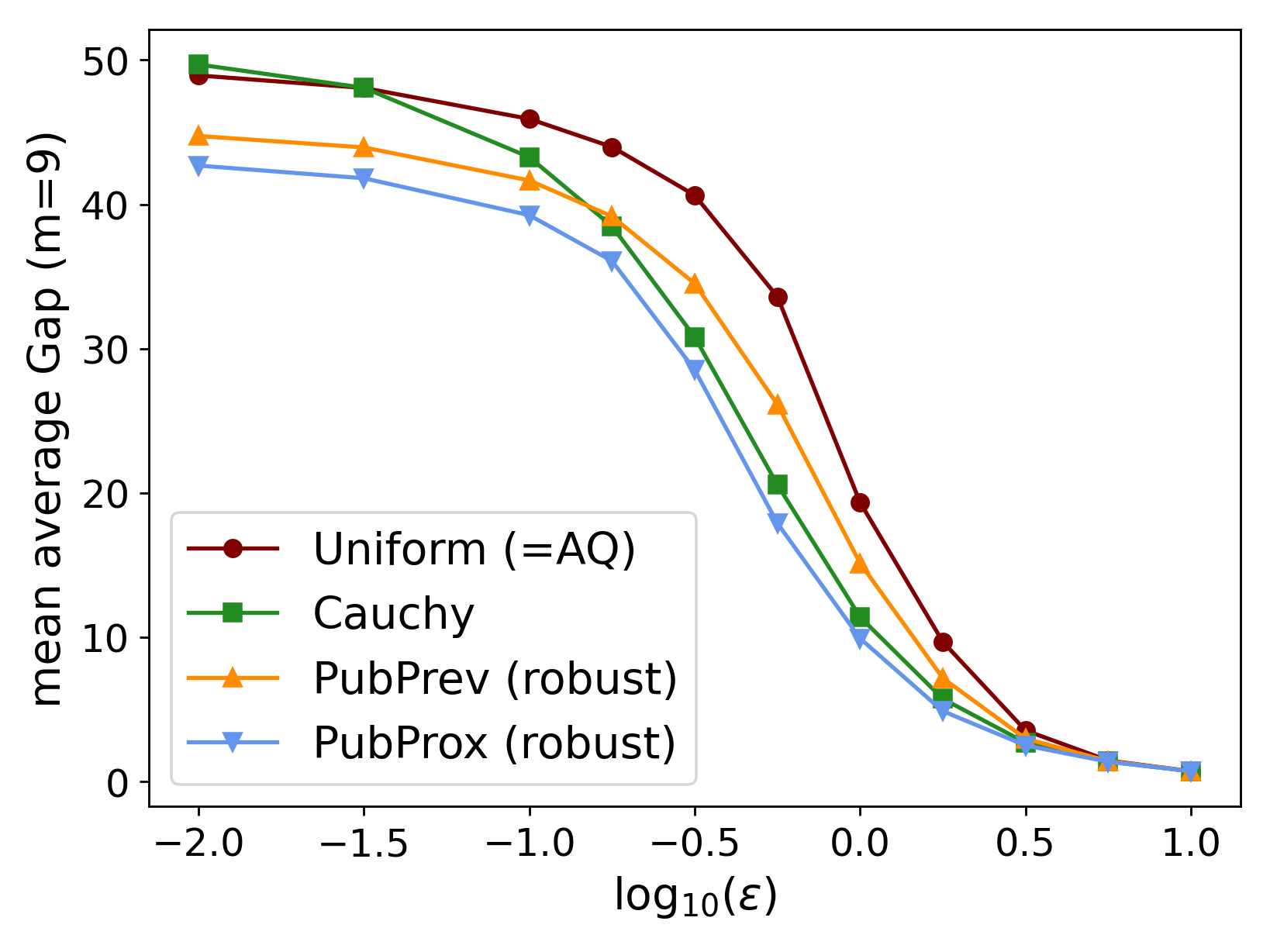}
	\hfill
	\includegraphics[width=.495\linewidth]{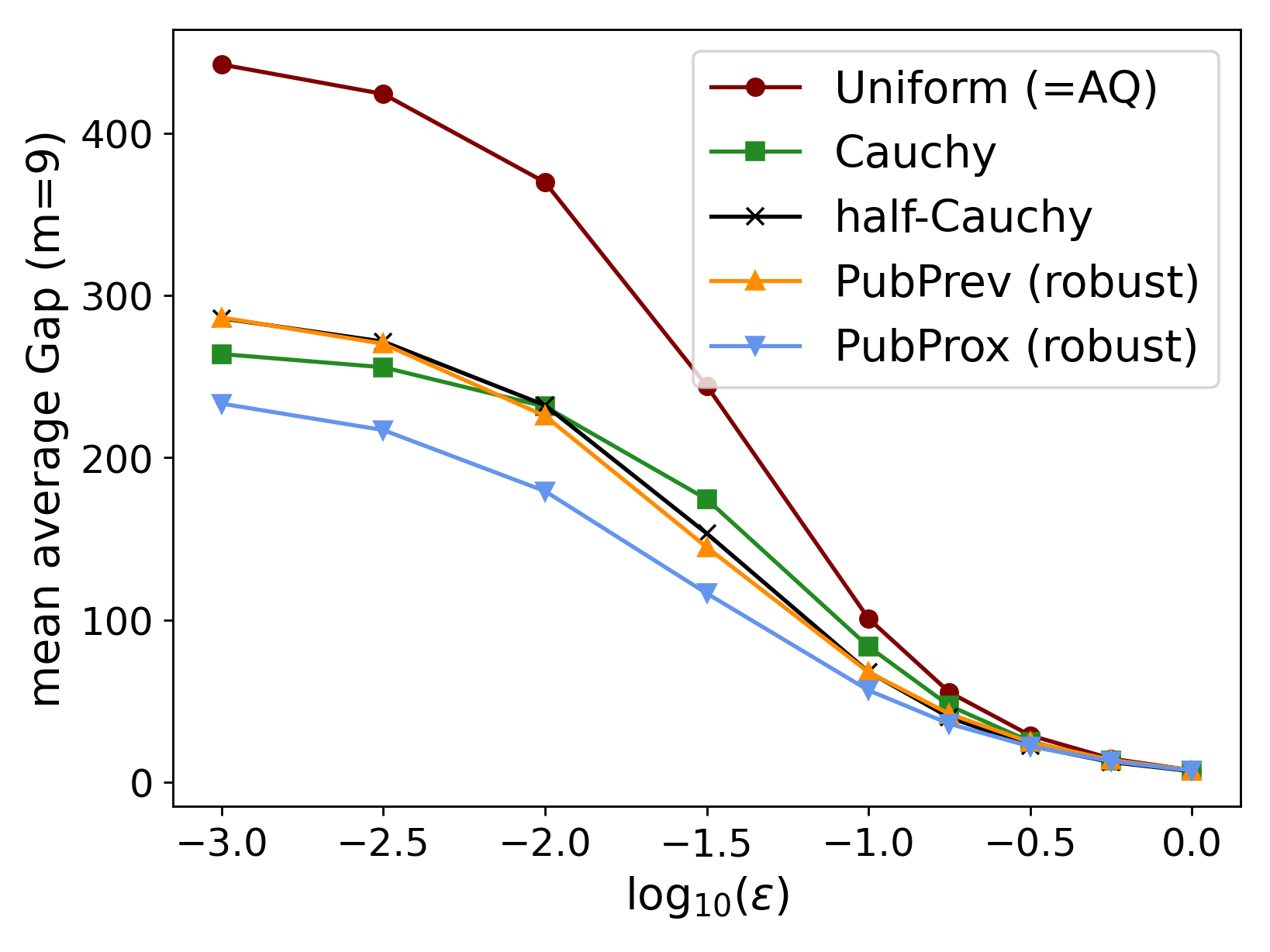}\vspace{-4mm}
	\caption{\label{fig:online-logepsplots}
		Time-averaged performance of the sequential release of nine quantiles on the Synthetic (left) and CitiBike (right) tasks. 
	}
\end{figure}

\begin{figure}[!t]
	\centering\vspace{-2mm}
	\includegraphics[width=.495\linewidth]{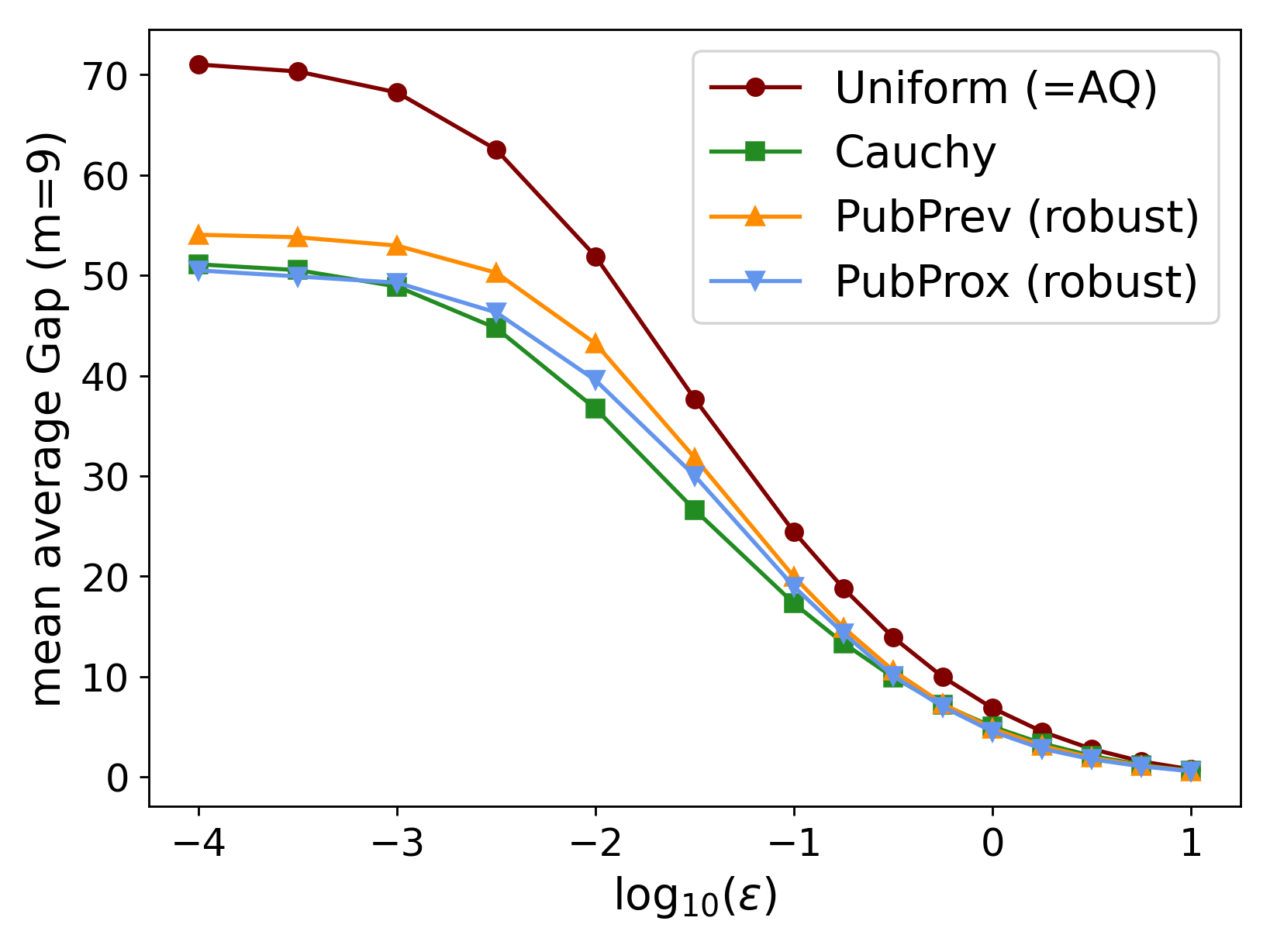}
	\hfill
	\includegraphics[width=.495\linewidth]{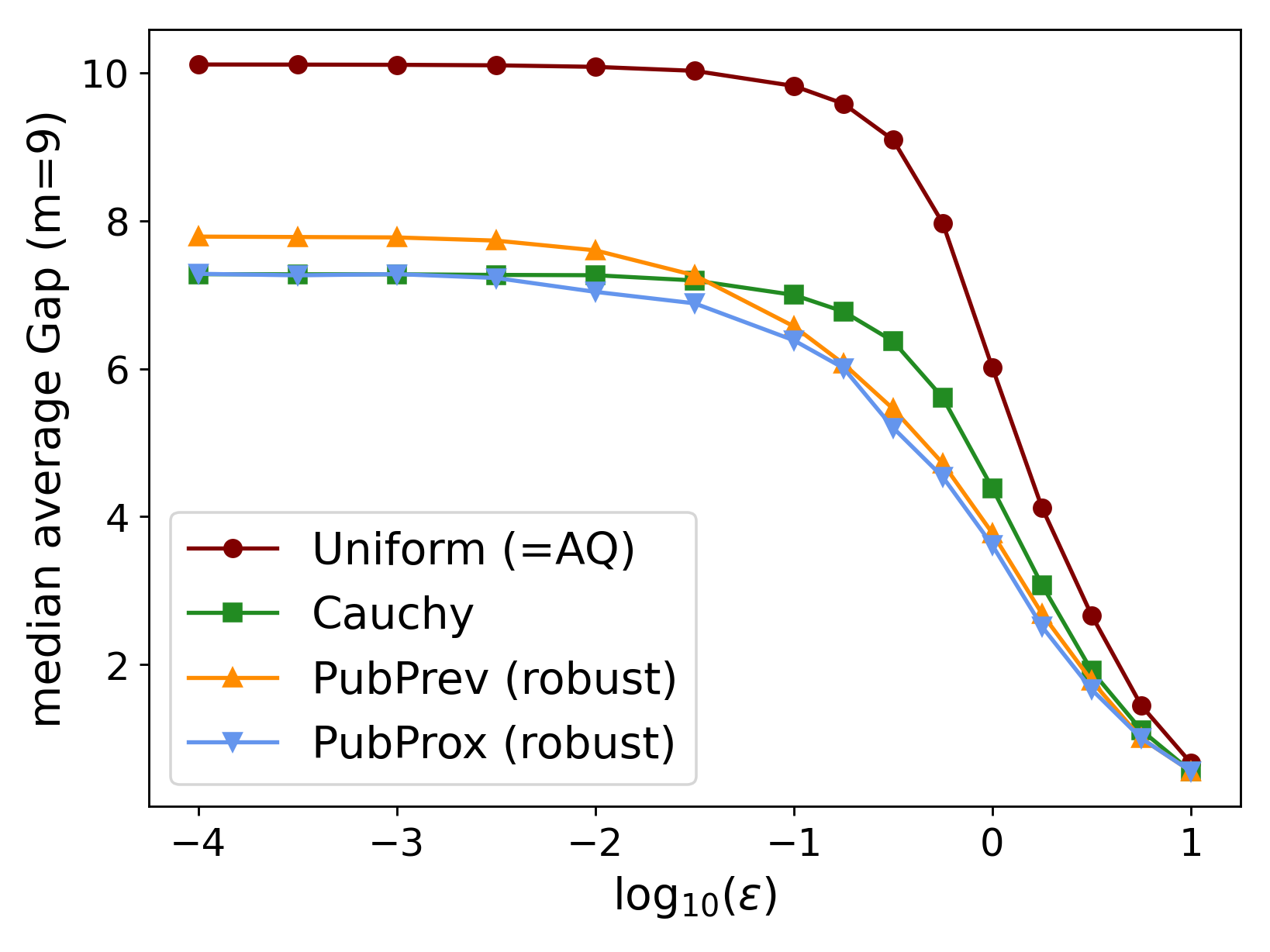}\vspace{-4mm}
	\caption{\label{fig:worldnews-logepsplots}
		Time-aggregated mean (left) and median (right) performance of sequential release of nine quantiles on the BBC task.
	}
\end{figure}
%!TEX root = main.tex

\section{Conclusion}

This work introduces the framework of private algorithms with private predictions, an extension of the algorithms with predictions setup to DP methods. 
Using the application of multiple quantile release, we provide strong theoretical and empirical evidence of its utility as a way of integrating external information in privacy-preserving algorithms.
We believe this way of studying DP methods is highly applicable and will see a great deal of future work in finding new applications for incorporating predictions or improving the approaches described here.
Some specific areas to explore include other important dataset statistics~\cite{biswas2020coinpress} and iterative data analysis methods~\cite{gupta2012iterative,hardt2010mw} and .
For multiple quantile release, our work directly suggests open questions, including obtaining algorithms with poly-logarithmic dependence on $m$, more natural prior adaptation schemes, and ways to use non-independent priors.

\section*{Acknowledgments}
This work was supported in part by a Facebook PhD Fellowship.\looseness-1

\bibliography{refs}
\bibliographystyle{icml2023}
\newpage
\appendix
\onecolumn
%!TEX root = main.tex

\section{Section~\ref{sec:algorithm} details}\label{app:algorithm}

\subsection{Quantile estimation via a prediction-dependent prior}

The base measure $\mu$ of DP mechanisms such as the exponential is the starting point of many approaches to incorporating external information, especially ones focused on Bayesian posterior sampling~\cite{dimitrakakis2017differential,geumlek2017renyi,seeman2020private};
while it is also our approach to single-quantile estimation with predictions, a key difference here is the focus on utility guarantees depending on both the prediction and instance, which is missing from this past work.
In the quantile problem, given a quantile $q$ and a sorted dataset $\*x\in\R^n$ of $n$ distinct points, the goal is to release a number $o$ that upper bounds exactly $\lfloor qn\rfloor$ of the entries.
A natural error metric, $\Gap_q(\*x,o)$, is the number of entries between the released number $o$ and $\lfloor qn\rfloor$, and we can show that prediction-dependent bound using astraightforward application of EM with utility $-\Gap_q$:
\begin{Lem}\label{lem:quantile}
	Releasing $o\in\R$ w.p. $\propto\exp(-\varepsilon\Gap_q(\*x,o)/2)\mu(o)$ is $\varepsilon$-DP, and w.p. $1-\beta$
	\begin{equation}
	\Gap_q(\*x,o)
	\le\frac2\varepsilon\left(\log\frac1\beta-\log\Psi_{\*x}^{(q,\varepsilon)}(\mu)\right)
	\le\frac2\varepsilon\left(\log\frac1\beta-\log\Psi_{\*x}^{(q)}(\mu)\right)
	\end{equation}
	where $\Psi_{\*x}^{(q,\varepsilon)}(\mu)=\sum_{i=0}^n\exp(-\varepsilon\Gap_q(\*x,I_i)/2)\mu(I_i)=\int\exp(-\varepsilon\Gap_q(\*x,o)/2)\mu(o)do$ is the inner product between $\mu$ and the exponential score while $\Psi_{\*x}^{(q)}(\mu)=\mu(I_{\lfloor qn\rfloor})$ is the measure of the optimal interval (note $\max_ku_q(\*x,I_k)=-\Gap_q(\*x,I_{\lfloor qn\rfloor})=0$ and so $\Psi_{\*x}^{(q)}(\mu)\le\Psi_{\*x}^{(q,\varepsilon)}(\mu)~\forall~\varepsilon>0$).\looseness-1
\end{Lem}
\begin{proof}
	$\varepsilon$-DP follows from $u_q$ having sensitivity one and the guarantee of EM with base measure $\mu$~\citep[Theorem~6]{mcsherry2007mechanism}.
	For the error, since we sample an interval $I_k$ and then sample $o\in I_k$ we have\looseness-1
	\begin{align}\label{eq:quantile}
	\begin{split}
	\Pr\{\Gap_q(\*x,o)\ge\gamma\}
	=\Pr\{u_q(\*x,I_k)\le-\gamma\}
	&=\sum_{j=0}^n\Pr\{k=j\}1_{u_q(\*x,I_j)\le-\gamma}\\
	&\le\sum_{j=0}^n\frac{\exp(-\frac{\varepsilon\gamma}2)\mu(I_j)}{\sum_{i=0}^n\exp(\frac\varepsilon2u_q(\*x,I_i))\mu(I_i)}
	\le\frac{\exp(-\frac{\varepsilon\gamma}2)}{\Psi_{\*x}^{(q,\varepsilon)}(\mu)}
	\end{split}
	\end{align}
	The result follows by substituting $\beta$ for the failure probability and solving for $\gamma$.
\end{proof}
We can also analyze the error metrics in this bound for specific measures $\mu$.
In particular, if the points are in a bounded interval $(a,b)$ and we use the uniform measure $\mu(o)=1_{o\in(a,b)}/(b-a)$ then $\Psi_{\*x}^{(q,\varepsilon)}(\mu)\ge\frac{\psi_{\*x}}{b-a}$, where $\psi_{\*x}=\min_k\*x_{[k+1]}-\*x_{[k]}$, and we exactly recover the standard bound of $\frac2\varepsilon\log\frac{b-a}{\beta\psi_{\*x}}$, e.g. the one in~\citep[Lemma~A.1]{kaplan2022quantiles} (indeed their analysis implicitly uses this measure).
However, our approach also allows us to remove the boundedness assumption, which itself can be viewed as a type of prediction, as one needs external information to assume that the data, or at least the quantile, lies within the interval $(a,b)$.
Taking this view, we can use the prediction to set the location $\nu\in\R$ and scale $\sigma>0$ of a Cauchy prior $\mu_{\nu,\sigma}(o)=\sigma/(\pi(\sigma^2+(o-\nu)^2))$ without committing to $(a,b)$ actually containing the data.
Since we know that the optimal interval $(\*x_{[\lfloor qn\rfloor]},\*x_{[\lfloor qn\rfloor+1]}]$ is a subset of $(\frac{a+b}2\pm R)$ for some $R>0$, setting $\nu=\frac{a+b}2$ and $\sigma=\frac{b-a}2$ yields
\begin{align}
\begin{split}
\Psi_{\*x}^{(q)}(\mu_{\nu,\sigma})
\ge\frac\sigma\pi\frac{\*x_{[\lfloor qn\rfloor+1]}-\*x_{[\lfloor qn\rfloor]}}{\sigma^2+\max_{k\in\{\lfloor qn\rfloor,\lfloor qn\rfloor+1\}}(\nu-\*x_{[k]})^2}
\ge\frac\sigma\pi\min_k\frac{\*x_{[k+1]}-\*x_{[k]}}{\sigma^2+R^2}
\ge\frac{2(b-a)\psi_{\*x}/\pi}{(b-a)^2+4R^2}
\end{split}
\end{align}
If $R=\frac{b-a}2$, i.e. we get the interval containing the data correct, then substituting the above into Lemma~\ref{lem:quantile} recovers the guarantee of the uniform prior up to an additive factor $\frac2\varepsilon\log\pi$.
However, whereas for the uniform prior we have no performance guarantees if the interval is incorrect, using the Cauchy prior the performance degrades gracefully as the error ($R$) grows.
While this first result can be viewed as designing a better prediction-free algorithm, it can also be viewed as making more robust use of the external information about the interval containing the data.
\looseness-1

\ifdefined\arxiv\else
\newpage
\fi
\subsection{Multiple-quantile release using multiple priors}

To estimate $m>1$ quantiles $q_1,\dots,q_m$ at once, we adapt the recursive approach of \cite{kaplan2022quantiles}, whose method \texttt{ApproximateQuantiles} implicitly constructs a binary tree with a quantile $q_i$ at each node and uses the exponential mechanism to compute the quantile $\tilde q_i=(q_i-\underline q_i)/(\overline q_i-\underline q_i)$ of the dataset $\*{\hat x}_i$ of points in the original dataset $\*x$ restricted to the interval $(\hat a_i,\hat b_i)$;
here $\underline q_i<q_i$ and $\overline q_i>q_i$ are quantiles appearing earlier in the tree whose respective estimates $\hat a_i$ and $\hat b_i$ determine the sub-interval (if there is no earlier quantile on the left and/or right of $q_i$ we use $\underline q_i=0,\hat a_i=a$ and/or $\overline q_i=1,\hat b_i=b$).
Because each datapoint only participates in $\BigO(\log_2 m)$ exponential mechanisms, the approach is able to run each mechanism with budget $\Omega(\varepsilon/\log_2m)$ and thus only suffer error logarithmic in the number of quantiles $m$, a significant improvement upon running one EM with budget $\varepsilon/m$ on the entire dataset for each quantile, which has error $\BigO(m)$ in the number of quantiles.

We can apply prior-dependent guarantees to \texttt{ApproximateQuantiles}---pseudocode for a generalized version of which is provided in Algorithm~\ref{alg:quantiles}---by recognizing that implicitly the method assigns a uniform prior $\mu_i$ to each quantile $q_i$ and then running EM with the {\em conditional} prior $\hat\mu_i$ restricted to the interval $[\hat a_i,\hat b_i]$ determined by earlier quantiles in the binary tree.
An extension of the argument in Equation~\ref{eq:quantile} (c.f. Lemma~\ref{lem:empirical-quantiles}) then yields a bound on the error of the estimate $o_i$ returned for quantile $q_i$ in terms of the prior-EM inner-product computed with this conditional prior $\hat\mu_i$ over the subset $\*{\hat x}_i$:
\begin{equation}\label{eq:empirical-probability}
\Pr\{\Gap_{q_i}(\*x,o_i)\ge\gamma\}\le\frac{\exp\left(\frac{\varepsilon_i}2(\hat\gamma_i-\gamma)\right)}{\Psi_{\*{\hat x}_i}^{(\tilde q_i,\varepsilon_i)}(\hat\mu_i)}\qquad\textrm{for}\qquad\hat\gamma_i=(1-\tilde q_i)\Gap_{\underline q_i}(\*x,\hat a_i)+\tilde q_i\Gap_{\overline q_i}(\*x,\hat b_i)
\end{equation}
Note that the error is offset by a weighted combination $\hat\gamma_i$ of the errors of the estimates of quantiles earlier in the tree.
Controlling this error allows us to bound the maximum error of any quantile via the harmonic mean of the inner products between the exponential scores and conditional priors:
\begin{Lem}\label{lem:empirical-quantiles}
	Algorithm~\ref{alg:quantiles} with $K=2$ and $\varepsilon_i=\varepsilon/\lceil\log_2m\rceil~\forall~i$ is $\varepsilon$-DP and w.p. $\ge1-\beta$ has
	\begin{equation}
	\max_i\Gap_{q_i}(\*x,o_i)\le\frac2\varepsilon\lceil\log_2m\rceil^2\log\frac m{\beta\hat\Psi_{\*x}^{(\varepsilon)}}\qquad\textrm{for}\qquad\hat\Psi_{\*x}^{(\varepsilon)}=\left(\sum_{i=1}^m\frac{1/m}{\Psi_{\*{\hat x}_i}^{(\tilde q_i,\varepsilon_i)}(\hat\mu_i)}\right)^{-1}
	\end{equation}
\end{Lem}
\begin{proof}
	The privacy guarantee follows as in \citep[Lemma~3.1]{kaplan2022quantiles}.
	Setting the above probability bound~\eqref{eq:empirical-probability} to $\frac{\beta\hat\Psi_{\*x}^{(\varepsilon)}}{m\Psi_{\tilde q_i}^{(\varepsilon_i)}(\*{\hat x}_i,\hat\mu_i)}$ for each $i$ we have w.p. $\ge1-\beta$ that $\Gap_{q_i}(\*x,o_i)\le\frac2{\bar\varepsilon}\log\frac m{\beta\hat\Psi_{\*x}^{(\varepsilon)}}+\hat\gamma_i~\forall~i$.
	Now let $k_i$ be the depth of quantile $q_i$ in the tree.
	If $k_i=1$ then $i$ is the root node so $\hat\gamma_i=0$ and we have $\Gap_{q_i}(\*x,o_i)\le\frac2{\bar\varepsilon}\log\frac m{\beta\hat\Psi_{\*x}^{(\varepsilon)}}$.
	To make an inductive argument, we assume $\Gap_{q_i}(\*x,o_i)\le\frac{2k}{\bar\varepsilon}\log\frac m{\beta\hat\Psi^{(\varepsilon)}}~\forall~i$ s.t. $k_i\le k$, and so for any $i$ s.t. $k_i=k+1$ we have that
	\begin{equation}
	\Gap_{q_i}(\*x,o_i)
	\le\frac2{\bar\varepsilon}\log\frac m{\beta\hat\Psi_{\*x}^{(\varepsilon)}}+(1-\tilde q_i)\Gap_{\underline q_i}(\*x,\hat a_i)+\tilde q_i\Gap_{\overline q_i}(\*x,\hat b_i)
	\le\frac{2(k+1)}{\bar\varepsilon}\log\frac m{\beta\hat\Psi_{\*x}^{(\varepsilon)}}
	\end{equation}
	Thus $\Gap_{q_i}(\*x,o_i)\le\frac{2k_i}{\bar\varepsilon}\log\frac m{\beta\hat\Psi_{\*x}^{(\varepsilon)}}~\forall~i$, so using $k_i\le\lceil\log_2m\rceil$ and $\bar\varepsilon=\frac\varepsilon{\lceil\log_2m\rceil}$ yields the result.\looseness-1
\end{proof}
Setting $\hat\mu_i$ to be uniform on $[\hat a_i,\hat b_i]$ exactly recovers both the algorithm and guarantee of \citep[Theorem~3.3]{kaplan2022quantiles}.
As before, we can also extend the algorithm to the infinite interval:
\begin{Cor}\label{cor:multi-cauchy}
	If all priors are Cauchy with location $\frac{a+b}2$ and scale $\frac{b-a}2$ and the data lies in the interval $(\frac{a+b}2\pm R)$ then w.p. $\ge1-\beta$ the maximum error is at most
	$\frac2\varepsilon\lceil\log_2m\rceil^2\log\left(\pi m\frac{b-a+\frac{4R^2}{b-a}}{2\beta\psi_{\*x}}\right)$.
\end{Cor}
However, while this demonstrates the usefulness of Lemma~\ref{lem:empirical-quantiles} for obtaining robust priors on infinite intervals, the associated prediction measure $\hat\Psi_{\*x}^{(\varepsilon)}$ is imperfect because it is non-deterministic: 
its value depends on the random execution of the algorithm, specifically on the data subsets $\*{\hat x}_i$ and priors $\hat\mu_i$, which for $i$ not at the root of the tree are affected by the DP mechanisms of $i$'s ancestor nodes.
In addition to not being given fully specified by the prediction and data, this makes $\hat\Psi^{(\varepsilon)}$ difficult to use as an objective for learning.
A natural more desirable prediction metric is the harmonic mean of the inner products between the exponential scores and {\em original} priors $\mu_i$ over the {\em original} dataset $\*x$, i.e. the direct generalization of our approach for single quantiles.

Unfortunately, the conditional restriction of $\mu_i$ to the interval $[\hat a_i,\hat b_i]$ removes the influence of probabilities assigned to intervals between points {\em not} in this interval.
To solve this, we propose a different {\em edge}-restriction of $\mu_i$ that assigns probabilities $\mu_i((-\infty,\hat a_i))$ and $\mu_i((\hat b_i,\infty))$ of being outside the interval $[\hat a_i,\hat b_i]$ to atoms on its edges $\hat a_i$ and $\hat b_i$, respectively.
Despite not using any information from points outside $\*{\hat x}_i$, this approach puts probabilities assigned to intervals outside $[\hat a_i,\hat b_i]$ to the edge closest to them, allowing us to extend the previous probability bound~\eqref{eq:empirical-probability} to depend on the original prior-EM inner-product (c.f. Lemma~\ref{lem:prior}):
\begin{equation}
\Pr\{\Gap_{q_i}(\*x,o_i)\ge\gamma\}
\le\exp(\varepsilon(\hat\gamma_i-\gamma/2))/\Psi_{\*x}^{(q_i,\varepsilon_i)}(\mu_i)
\end{equation}
However, the stronger dependence of this bound on errors $\hat\gamma_i$ earlier in the tree lead to an $\tilde\BigO(\phi^{\log_2m})=\BigO(m^{0.7})$ dependence on $m$, where $\phi=\frac{1+\sqrt5}2$ is the golden ratio:
\begin{Thm}\label{thm:binary}
	If the quantiles are uniform negative powers of two then Algorithm~\ref{alg:quantiles} with $K=2$, edge-based prior adaptation, and $\varepsilon_i=\varepsilon/\lceil\log_2(m+1)\rceil~\forall~i$ is $\varepsilon$-DP and w.p. $\ge1-\beta$ has
	\begin{equation}
	\max_i\Gap_{q_i}(\*x,o_i)
	\le\frac2\varepsilon\phi^{\log_2(m+1)}\lceil\log_2(m+1)\rceil\log\frac m{\beta\Psi_{\*x}^{(\varepsilon)}}\qquad\textrm{for}\qquad\Psi_{\*x}^{(\varepsilon)}=\left(\sum_{i=1}^m\frac{1/m}{\Psi_{\*x}^{(q_i,\varepsilon_i)}(\mu_i)}\right)^{-1}
	\end{equation}
\end{Thm}
\begin{proof}
	Since $\tilde q_i=1/2~\forall~i$, setting the new probability bound equal to $\frac{\beta\Psi_{\*x}^{(\varepsilon)}}{m\Psi_{\*x}^{(q_i\varepsilon_i)}(\mu_i)}$ yields that w.p. $\ge1-\beta$ 
	\begin{equation}
		\Gap_{q_i}(\*x,o_i)\le\frac2{\bar\varepsilon}\log\frac m{\beta\Psi_{\*x}^{(\varepsilon)}}+2\hat\gamma_i=\frac2{\bar\varepsilon}\log\frac m{\beta\Psi_{\*x}^{(\varepsilon)}}+\Gap_{\underline q_i}(\*x,\hat a_i)+\Gap_{\overline q_i}(\*x,\hat b_i)~\forall~i
	\end{equation}
	If for each $k\le\lceil\log_2m\rceil$ we define $E_k$ to be the maximum error of any quantile of at most depth $k$ in the tree then since one of $\underline q_i$ and $\overline q_i$ is at depth at least one less than $q_i$ and the other is at depth at least two less than $q_i$ we have $E_k\le\frac{2A_k}{\bar\varepsilon}\log\frac m{\beta\Psi_{\*x}^{(\varepsilon)}}$ for recurrent relation $A_k=1+A_{k-1}+A_{k-2}$ with $A_0=0$ and $A_1=1$.
	Since $A_k=F_{k+1}-1$ for Fibonacci sequence $F_j=\frac{\phi^j-(1-\phi)^j}{\sqrt 5}$, we have 
	\begin{equation}
		\max_i\Gap_{q_i}(\*x,o_i)=\max_kE_k\le\frac{2\phi^{\lceil\log_2(m+1)\rceil+1}}{\bar\varepsilon\sqrt5}\log\frac m{\beta\Psi_{\*x}^{(\varepsilon)}}=\frac{2\phi^{\lceil\log_2(m+1)\rceil+1}}{\varepsilon\sqrt5}\lceil\log_2(m+1)\rceil\log\frac m{\beta\Psi_{\*x}^{(\varepsilon)}}
	\end{equation}
\end{proof}
Thus while we have obtained a performance guarantee depending only on the prediction and the data via the harmonic mean $\Psi_{\*x}^{(\varepsilon)}$ of the true prior-EM inner-products, the dependence on $m$ is now polynomial.
Note that it is still sublinear, which means it is better than the naive baseline of running $m$ independent exponential mechanisms.
Still, we can do much better---in-fact asymptotically better than any power of $m$---by recognizing that the main issue is the compounding error induced by successive errors to the boundaries of sub-intervals.
We can reduce this by reducing the depth of the tree using a $K$-ary rather than binary tree and instead paying $K-1$ times the privacy budget at each depth in order to naively release values for $K-1$ quantiles.
This can introduce out-of-order quantiles, but by Lemma~\ref{lem:shuffle} swapping any two out-of-order quantiles does not increase the maximum error and so this issue can be solved by sorting the $K-1$ quantiles before using them to split the data.
We thus have the following prediction-dependent performance bound for multiple quantiles:\looseness-1
\begin{Thm}\label{thm:kary}
	If we run Algorithm~\ref{alg:quantiles} with $K=\lceil\exp(\sqrt{\log2\log(m+1)})\rceil$, edge-based adaptation, and $\varepsilon_i=\frac{\bar\varepsilon}{k_i^p}$ for some power $p>1$, $k_i$ the depth of $q_i$ in the $K$-ary tree, and $\bar\varepsilon=\frac\varepsilon{K-1}\left(\sum_{k=1}^{\lceil\log_K(m+1)\rceil}\frac1{k^p}\right)^{-1}$, then the result satisfies $\varepsilon$-DP and w.p. $\ge1-\beta$ we have 
	\ifdefined\arxiv
	\begin{equation}
		\max_i\Gap_{q_i}(\*x,o_i)
		\le\frac{2\pi^2}\varepsilon\exp\left(2\sqrt{\log(2)\log(m+1)}\right)\log\frac m{\beta\Psi_{\*x}^{(\varepsilon)}}
	\end{equation}
	\else
	$\max_i\Gap_{q_i}(\*x,o_i)
	\le\frac{2\pi^2}\varepsilon\exp\left(2\sqrt{\log(2)\log(m+1)}\right)\log\frac m{\beta\Psi_{\*x}^{(\varepsilon)}}$
	\fi 
	if $p=2$ and more generally $\max_i\Gap_{q_i}(\*x,o_i)
	\le\frac{c_p}\varepsilon\exp\left(2\sqrt{\log(2)\log(m+1)}\right)\log\frac m{\beta\Psi_{\*x}^{(\varepsilon)}}$, where $c_p$ depends only on $p$.
\end{Thm}
\ifdefined\arxiv\else
\newpage
\fi
\begin{proof}
	The privacy guarantee follows as in \citep[Lemma~3.1]{kaplan2022quantiles} except before each split we compute $K-1$ quantiles with $K-1$ times less budget.
	As in the previous proof, we have w.p. $\ge1-\beta$ that 
	\begin{equation}
		\Gap_{q_i}(\*x,o_i)\le\frac2{\varepsilon_i}\log\frac m{\beta\Psi_{\*x}^{(\varepsilon)}}+2\hat\gamma_i=\frac{2k_i^2}{\bar\varepsilon}\log\frac m{\beta\Psi_{\*x}^{(\varepsilon)}}+2(1-\tilde q_i)\Gap_{\underline q_i}(\*x,\hat a_i)+2\tilde q_i\Gap_{\overline q_i}(\*x,\hat b_i)~\forall~i
	\end{equation}
	If for each $k\le\lceil\log_K(m+1)\rceil$ we define $E_k$ to be the maximum error of any quantile of at most depth $k$ in the tree then since both $\underline q_i$ and $\overline q_i$ are at depth at least one less than $q_i$ we have $E_k\le\frac{2A_k}{\bar\varepsilon}\log\frac m{\beta\Psi_{\*x}^{(\varepsilon)}}$, where $A_k=k^p+2A_{k-1}$ and $A_1=1$.
	For the case of $p=2$,  $A_k\le6\cdot2^k$ and $1/\bar\varepsilon=\frac{K-1}\varepsilon\sum_{k=1}^{\lceil\log_K(m+1)\rceil}\frac1{k^2}\le\frac{\pi^2}{6\varepsilon}(K-1)$ so we have that 
	\begin{equation}
		\max_i\Gap_{q_i}(\*x,o_i)=\max_kE_k\le\frac{12}{\bar\varepsilon}2^{\lceil\log_K(m+1)\rceil}\log\frac m{\beta\Psi_{\*x}^{(\varepsilon)}}\le\frac{2\pi^2}{\varepsilon}(K-1)2^{\lceil\log_K(m+1)\rceil}\log\frac m{\beta\Psi_{\*x}^{(\varepsilon)}}
	\end{equation}
	Substituting $K=\lceil\exp(\sqrt{\log2\log(m+1)})\rceil$ and simplifying yields the result.
	For $p>1$, $A_k\le2^{k-2}\left(2+\Phi\left(\frac12,-p,2\right)\right)$, where $\Phi$ is the Lerch transcendent, and $1/\bar\varepsilon\le\frac{K-1}\varepsilon\zeta(p)$, where $\zeta$ is the Riemann zeta function.
	Therefore
	\ifdefined\arxiv
	\begin{align}
	\begin{split}
		\max_i\Gap_{q_i}(\*x,o_i)
		=\max_kE_k
		&\le\frac{2^{\lceil\log_K(m+1)\rceil}}{2\bar\varepsilon}\left(2+\Phi\left(\frac12,-p,2\right)\right)\log\frac m{\beta\Psi_{\*x}^{(\varepsilon)}}\\
		&\le\frac{c_p}{\varepsilon}(K-1)2^{\lceil\log_K(m+1)\rceil}\log\frac m{\beta\Psi_{\*x}^{(\varepsilon)}}
	\end{split}
	\end{align}
	\else
	\begin{equation}
		\max_i\Gap_{q_i}(\*x,o_i)
		=\max_kE_k
		\le\frac{2^{\lceil\log_K(m+1)\rceil}}{2\bar\varepsilon}\left(2+\Phi\left(\frac12,-p,2\right)\right)\log\frac m{\beta\Psi_{\*x}^{(\varepsilon)}}
		\le\frac{c_p}{\varepsilon}(K-1)2^{\lceil\log_K(m+1)\rceil}\log\frac m{\beta\Psi_{\*x}^{(\varepsilon)}}
	\end{equation}
	\fi
	for $c_p=\left(1+\Phi\left(\frac12,-p,2\right)/2\right)\zeta(p)$.
\end{proof}

Similarly to Theorem~\ref{thm:binary}, the proof establishes a recurrence relationship between the maximum errors at each depth.
Note that in addition to the $K$-ary tree this bound uses depth-dependent budgeting to remove a $\BigO(\log_2m)$-factor;
the constant depending upon the parameter $p>1$ of the latter has a minimum of roughly $8.42$ at $p\approx1.6$.
As discussed before, the new dependence $\tilde\BigO\left(\exp\left(2\sqrt{\log(2)\log(m+1)}\right)\right)$ on $m$ is sub-polynomial, i.e $o(m^\alpha)~\forall~\alpha>0$.
While it is also super-polylogarithmic, its shape for any practical value of $m$ is roughly $\BigO(\log_2^2m)$, making the result of interest as a justification for the negative log-inner-product performance metric.

\subsection{Experimental details}

For the experiments in Section~\ref{sec:algorithm}, specifically Figures~\ref{fig:mixplots}, we evaluate three variants of the algorithm on data drawn from a standard Gaussian distribution and from the Adult ``age" dataset~\citep{kohavi1996adult}. 
In both cases we use 1000 samples and run each experiment 40 times, reporting the average performance.
As we do for all datasets, we use reasonable guesses of mean, scale, and bounds on each dataset to set priors.
As in this section we report the Uniform, we need to specify its range; for Gaussian we use $[-10,10]$, while for ``age" we use $[10,120]$.

The original AQ algorithm of~\citet{kaplan2022quantiles} is now fully specified.
We test two variants of our $K$-ary modification:
one with edge-based adaptation, and the other using the original conditional adaptation.
For both cases we set $K$ as a function of $m$ according to the formula in Theorem~\ref{thm:kary-main}, and we set the power $p$ of the depth-dependent budget discounting to 1.5, which is close to the theoretically optimal value of around 1.6 (c.f. Thm~\ref{thm:kary}).

\ifdefined\arxiv\else
\newpage
\fi
\section{Section~\ref{sec:advantages} details}\label{app:advantages}

\subsection{Robustness-consistency tradeoffs}

While prediction-dependent guarantees work well if the prediction is accurate, without safeguards they may perform catastrophically poorly if the prediction is incorrect.
Quantiles provide a prime demonstration of the importance of robustness, as using priors allows for approaches that may assign very little probability to the interval containing the quantile.
For example, if one is confident that it has a specific value $x\in(a,b)$ one can specify a more concentrated prior, e.g. the Laplace distribution around $x$.
Alternatively, if one believes the data is drawn i.i.d. from some a known distribution then $\mu$ can be constructed via its CDF using order statistics~\citep[Equation~2.1.5]{david2003order}.
These reasonable approaches can result in distributions with exponential or high-order-polynomial tails, using which directly may work poorly if the prediction is incorrect.\looseness-1

Luckily, for our negative log-inner-product error metric it is straightforward to show a parameterized robustness-consistency tradeoff by simply mixing the prediction prior $\mu$ with a robust prior $\rho$:\looseness-1
\begin{Cor}\label{cor:quantile}
	For any prior $\mu:\R\mapsto\R_{\ge0}$, robust prior $\rho:\R\mapsto\R_{\ge0}$, and robustness parameter $\lambda\in[0,1]$, releasing $o\in\R$ w.p. $\propto\exp(-\varepsilon\Gap_q(\*x,o)/2)\mu^{(\lambda)}(o)$ for $\mu^{(\lambda)}=(1-\lambda)\mu+\lambda\rho$ is $\left(\frac2\varepsilon\log\frac{1/\beta}{\lambda\Psi_{\*x}^{(q,\varepsilon)}(\rho)}\right)$-robust and $\left(\frac2\varepsilon\log\frac{1/\beta}{1-\lambda}\right)$-consistent w.p. $\ge1-\beta$.
\end{Cor}
\begin{proof}
	Apply Lemma~\ref{lem:quantile} and linearity of $\Psi_{\*x}^{(q,\varepsilon)}(\mu^{(\lambda)})=(1-\lambda)\Psi_{\*x}^{(q,\varepsilon)}(\mu)+\lambda\Psi_{\*x}^{(q,\varepsilon)}(\rho)$.
\end{proof}
Thus if the interval is finite and we set $\rho$ to be the uniform prior, using $\mu^{(\lambda)}$ in the algorithm will have a high probability guarantee at most $\frac2\varepsilon\log\frac1\lambda$-worse than the prediction-free guarantee of~\citet[Lemma~A.1]{kaplan2022quantiles}, no matter how poor $\mu$ is for the data, while also guaranteeing w.p. $\ge1-\beta$ that the error will be at most $\frac2\varepsilon\log\frac{1/\beta}{1-\lambda}$ if $\mu$ is perfect.
A similar result holds for the case of an infinite interval if we instead use a Cauchy prior.
Corollary~\ref{cor:quantile} demonstrates the usefulness of the algorithms with predictions framework for not only quantifying improvement in utility using external information but also for making the resulting DP algorithms robust to prediction noise.\looseness-1

The above argument for single-quantiles is straightforward to extend to the negative log of the harmonic means of the inner products.
In-fact for the binary case with uniform quantiles we can trade-off between $\PolyLog(m)$-guarantees similar to those of~\citet{kaplan2022quantiles} and our prediction-dependent bounds:\looseness-1
\begin{Cor}\label{cor:multiple}
	Consider priors $\mu_1,\dots,\mu_m:\R\mapsto\R_{\ge0}$, Cauchy prior $\rho:\R\mapsto\R_{\ge0}$ with location $\frac{a+b}2$ and scale $\frac{b-a}2$, and robustness parameter $\lambda\in[0,1]$.
	Then running Algorithm~\ref{alg:quantiles} on quantiles that are uniform negative powers of two with $K=2$, edge-based prior adaptation, $\varepsilon_i=\bar\varepsilon=\varepsilon/\lceil\log_2m\rceil~\forall~i$, and priors $\mu_i^{(\lambda)}=\lambda\rho+(1-\lambda)\mu_i~\forall~i$ is $\left(\frac2\varepsilon\lceil\log_2m\rceil^2\log\left(\pi m\frac{b-a+\frac{4R^2}{b-a}}{2\lambda\beta\psi_{\*x}}\right)\right)$-robust and $\left(\frac2\varepsilon\phi^{\log_2m}\lceil\log_2m\rceil\log\frac{m/\beta}{1-\lambda}\right)$-consistent w.p. $\ge1-\beta$.
\end{Cor}
\begin{proof}
	Apply Lemma~\ref{lem:empirical-quantiles}, Theorem~\ref{thm:binary}, and the linearity of inner products making up $\hat\Psi_{\*x}^{(\varepsilon)}$ and $\Psi_{\*x}^{(\varepsilon)}$.\looseness-1
\end{proof}

\subsection{Learning predictions, privately}

Past work, e.g. the public-private framework~\cite{liu2021leveraging,bassily2022private,bie2022private}, has often focused on domain adaptation-type learning where we adapt a public source to private target.
We avoid assuming access to large quantities of i.i.d. public data and instead assume numerous tasks that can have sensitive data and may be adversarially generated.
As discussed before, this is the online setting where we see loss functions defined by a sequence of datasets $\*x_1,\dots,\*x_T$ and aim to compete with best fixed prediction in-hindsight.
Note such a guarantee can also be converted into excess risk bounds (c.f. Appendix~\ref{sec:o2b}).

\subsubsection{Non-Euclidean DP-FTRL}

Because the optimization domain is not well-described by the $\ell_2$-ball, we are able to obtain significant savings in dependence on the dimension and in some cases even in the number of instances $T$ by extending the DP-FTRL algorithm of \cite{kairouz2021practical} to use non-Euclidean regularizers, as in Algorithm~\ref{alg:dpftrl}.
For this we prove the following regret guarantee:
\ifdefined\arxiv\else
\newpage
\fi
\begin{Thm}\label{thm:dpftrl}
	Let $\theta_1,\dots,\theta_T$ be the outputs of Algorithm~\ref{alg:dpftrl} using a regularizer $\phi:\Theta\mapsto\R$ that is strongly-convex w.r.t. $\|\cdot\|$.
	Suppose $\forall~t\in[T]$ that $\ell_{\*x_t}(\cdot)$ is $L$-Lipschitz w.r.t. $\|\cdot\|$ and its gradient has $\ell_2$-sensitivity $\Delta_2$.
	Then w.p. $\ge1-\beta'$ we have $\forall~\theta^\ast\in\Theta$ that\looseness-1
	\begin{equation}
	\sum_{t=1}^T\ell(\theta_t;\*x_t)-\ell(\theta^\ast;\*x_t)
	\le\frac{\phi(\theta^\ast)-\phi(\theta_1)}\eta+\eta L\left(L+\left(G+C\sqrt{2\log\frac T{\beta'}}\right)\sigma\Delta_2\sqrt{\lceil\log_2T\rceil}\right)T
	\end{equation}
	where $G=\E_{\*z\sim\N(\*0_p,\*I_p)}\sup_{\|\*y\|\le1}\langle\*z,\*y\rangle=\E_{\*z\sim\N(\*0_p,1)}\|\*z\|_\ast$ is the Gaussian width of the unit $\|\cdot\|$-ball and $C$ is the Lipschitz constant of $\|\cdot\|_\ast$ w.r.t. $\|\cdot\|_2$.
	Furthermore, for any $\varepsilon'\le2\log\frac1{\delta'}$, setting $\sigma=\frac1{\varepsilon'}\sqrt{2\lceil\log_2T\rceil\log\frac1{\delta'}}$ makes the algorithm $(\varepsilon',\delta')$-DP.
\end{Thm}
\begin{proof}
	The privacy guarantee follows from past results for tree aggregation~\cite{smith2013optimal,kairouz2021practical}.
	For all $t\in[T]$ we use the shorthand $\nabla_t=\nabla_\theta\ell_{\*x_t}(\theta_t)$;
	we can then define $\tilde\theta_t=\argmin_{\theta\in\Theta}\phi(\theta)+\eta\sum_{s=1}^t\langle\nabla_s,\theta\rangle$ and $\*b_t=\*g_t-\sum_{s=1}^t\nabla_s$.
	Then 
	\begin{align}
	\begin{split}
	\sum_{t=1}^T\ell_{\*x_t}(\theta_t)-\ell_{\*x_t}(\theta^\ast)
	\le\sum_{t=1}^T\langle\nabla_t,\theta_t-\theta^\ast\rangle
	&=\sum_{t=1}^T\langle\nabla_t,\tilde\theta_t-\theta^\ast\rangle+\sum_{t=1}^T\langle\nabla_t,\theta_t-\tilde\theta_t\rangle\\
	&\le\frac{\phi(\theta^\ast)-\phi(\theta_1)}\eta+\eta\sum_{t=1}^T\|\nabla_t\|_\ast^2+\sum_{t=1}^T\|\nabla_t\|_\ast\|\tilde\theta_t-\theta_t\|\\
	&\le\frac{\phi(\theta^\ast)-\phi(\theta_1)}\eta+\eta L\left(LT+\sum_{t=1}^T\|\*b_t\|_\ast\right)
	\end{split}
	\end{align}
	where the first inequality follows from the standard linear approximation in online convex optimization \citep{zinkevich2003oco}, the second by the regret guarantee for online mirror descent \citep[Theorem~2.15]{shalev-shwartz2011oco}, and the last by applying \citet[Lemma~7]{mcmahan2017survey} with $\phi_1(\cdot)=\phi(\cdot)+\eta\sum_{s=1}^t\langle\nabla_s,\cdot\rangle$, $\psi(\cdot)=\eta\langle\*b_t,\cdot\rangle$, and $\phi_2(\cdot)=\phi(\cdot)+\eta\langle\*g_t,\cdot\rangle$, yielding $\|\tilde\theta_t-\theta_t\|\le\eta\|\*b_t\|_\ast~\forall~t\in[T]$.
	The final guarantee follows by observing that the tree aggregation protocol adds noise $\*b_t\sim\N(\*0_p,\sigma^2\Delta_2^2\lceil\log_2t\rceil)$ to each prefix sum and applying the Gaussian concentration of Lipschitz functions \citep[Theorem~5.6]{boucheron2012concentration}.
\end{proof}

The above proof of this result follows that of the Euclidean case, which can be recovered by setting $G=\BigO(\sqrt d)$, $C=1$, and $\Delta_2=\BigO(L)$.\footnote{As of this writing, the most recent arXiv version of \citet[Theorem~C.1]{kairouz2021practical} has a typo leading to missing a Lipschitz constant in the bound, confirmed via correspondence with the authors.}
In addition to the Lipschitz constants $L$, a key term that can lead to improvement is the Gaussian width $G$ of the unit $\|\cdot\|$-ball, which for the Euclidean case is $\BigO(\sqrt d)$ but e.g. for $\|\cdot\|=\|\cdot\|_1$ is $\BigO(\sqrt{\log d})$.
Note that a related dependence on the Laplace width of $\Theta$ appears in \citet[Theorem~3.1]{agarwal2017price}, although their guarantee only holds for linear losses and is not obviously extendable.
Thus Theorem~\ref{thm:dpftrl} may be of independent interest for DP online learning.

\begin{algorithm}[!t]
	\DontPrintSemicolon
	\KwIn{Datasets $\*x_1,\dots,\*x_T$ arriving in a stream in arbitrary order, domain $\Theta\subset\R^p$, step-size $\eta>0$, noise scale $\sigma>0$, $\ell_2$-sensitivity $\Delta_2>0$, regularizer $\phi:\Theta\mapsto\R$}
	$\*g_1\gets\*0_p$\\
	$\T\gets$\texttt{InitializeTree}($T,\sigma^2,\Delta_2$)\tcp*{start tree aggregation}
	\For{$t=1,\dots,T$}{
		$\theta_t\gets\argmin_{\theta\in\Theta}\phi(\theta)+\eta\langle\*g_t,\theta\rangle$\\
		suffer $\ell_{\*x_t}(\theta_t)$\\
		$\T\gets$\texttt{AddToTree}($\T,t,\nabla_\theta\ell_{\*x_t}(\theta_t)$)\tcp*{add gradient to tree}
		$\*g_{t+1}\gets$\texttt{GetSum}($\T,t$)\tcp*{estimate $\sum_{s=1}^t\nabla_\theta\ell_{\*x_s}(\theta_s)$}
	}
	\caption{\label{alg:dpftrl}Non-Euclidean DP-FTRL. For the \texttt{InitializeTree}, \texttt{AddToTree}, and \texttt{GetSum} subroutines see \citet[Section~B.1]{kairouz2021practical}.}
\end{algorithm}

\subsubsection{Learning priors for one or more quantiles}

We now turn to learning priors $\mu_t=\begin{pmatrix}\mu_{t[1]},\cdots,\mu_{t[m]}\end{pmatrix}$ to privately estimate $m$ quantiles $q_1,\dots,q_m$ on each of a sequence of $T$ datasets $\*x_t$.
We will aim to set $\mu_1,\dots,\mu_T$ s.t. if at each time $t$ we run Algorithm~\ref{alg:quantiles} with privacy $\varepsilon>0$ then the guarantees given by Lemmas~\ref{thm:binary} and~\ref{thm:kary} will be asymptotically at least as good as those of the best set of measures in $\F^m$, where $\F$ is some class of measures on the finite interval $(a,b)$. 
The latter we will assume to be known and bounded.
Note that in this section almost all single-quantile results follow from setting $m=1$, so we study it jointly with learning for multiple quantiles.\looseness-1

Ignoring constants, the loss functions implied by our prediction-dependent upper bounds for multiple-quantiles are the following negative log-harmonic sums of prior-EM inner-products: 
\begin{equation}
U_{\*x_t}^{(\varepsilon)}(\mu)
=\log\sum_{i=1}^m\frac1{\Psi_{\*x_t}^{(q_i,\varepsilon_i)}(\mu_{[i]})}
=\log\sum_{i=1}^m\frac1{\int_a^b\exp(-\varepsilon_i\Gap_{q_i}(\*x_t,o)/2)\mu_{[i]}(o)do}
\end{equation}
We focus on minimizing regret $\max_{\mu\in\F^m}\sum_{t=1}^TU_{\*x_t}^{(\varepsilon)}(\mu_t)-U_{\*x_t}^{(\varepsilon)}(\mu)$ over these losses for priors $\mu_{[i]}$ in a class $\F_{V,d}$ of probability measures that are piecewise $V$-Lipschitz over each of $d$ intervals uniformly partitioning $[a,b)$.
This is chosen because it covers the class $\F_{V,1}$ of $V$-Lipschitz measures and the class of $\F_{0,d}$ of discrete measures that are constant on each of the $d$ intervals.
The latter can be parameterized by $\*W\in\triangle_d^m$, so that the losses have the form $U_{\*x_t}^{(\varepsilon)}(\mu_{\*W})=\log\sum_{i=1}^m\langle\*s_{t,i},\*W_{[i]}]\rangle^{-1}$ for $\*s_{t,i}\in\R_{\ge0}^d$.
This can be seen by setting $\*s_{t,i[j]}=\frac d{b-a}\int_{a+\frac{b-a}d(j-1)}^{a+\frac{b-a}dj}\exp(-\varepsilon_i\Gap_{q_i}(\*x_t,o)/2)do$ and $\mu_{\*W_{[i]}}(o)=\frac d{b-a}\*W_{[i,j]}$ over the interval $\left[a+\frac{b-a}d(j-1),a+\frac{b-a}dj\right)$.
Finally, for $\lambda\in[0,1]$ we also let $\F^{(\lambda)}=\{(1-\lambda)\mu+\frac\lambda{b-a}:\mu\in\F\}$ denote the class of mixtures of measures $\mu\in\F$ with the uniform measure.\looseness-1

As detailed in Appendix~\ref{sec:overlap}, losses of the form $-\log\langle\*s_t,\cdot\rangle$, i.e. those above when $m=1$, have been studied in (non-private) online learning~\cite{hazan2007logarithmic,balcan2021ltl}.
However, specialized approaches, e.g. those taking advantage exp-concavity, are not obviously implementable via prefix sums of gradients, the standard approach to private online learning~\cite{smith2013optimal,agarwal2017price,kairouz2021practical}.
Still, we can at least use the fact that we are optimizing over a product of simplices to improve the dimension-dependence by applying Non-Euclidean DP-FTRL with entropic regularizer $\phi(\*W)=m\langle\*W,\log\*W\rangle$, which yields an $m$-way exponentiated gradient~(EG) update~\cite{kivinen1997eg}.
To apply its guarantee for the problem of learning priors for quantile estimation, we need to bound the sensitivity of the gradients $\nabla_{\*W}U_{\*x_t}^{(\varepsilon)}(\mu_{\*W})$ to changes in the underlying datasets $\*x_t$.
This is often done via a bound on the gradient norm, which in our case is unbounded near the boundary of the simplex.
We thus restrict to $\gamma$-robust priors for some $\gamma\in(0,1]$ by constraining $\*W\in\triangle_d^m$ to have entries lower bounded by $\gamma/d$---a domain where $\|\nabla_{\*W}U_{\*x_t}^{(\varepsilon)}(\mu_{\*W})\|_1\le d/\gamma$ (c.f. Lemma~\ref{lem:lipschitz})---and bounding the resulting approximation error;
we are not aware of even a non-private approach that avoids this except by taking advantage of exp-concavity~\cite{hazan2007logarithmic}.\looseness-1

We thus have a bound of $2d/\gamma$ on the $\ell_2$-sensitivity.
However, this may be too loose since it allows for changing the entire dataset $\*x_t$, whereas we are only interested in changing one entry.
Indeed, for small $\varepsilon$ we can obtain a tighter bound:\looseness-1
\begin{Lem}\label{lem:refined-sensitivity}
	The $\ell_2$-sensitivity of $\nabla_{\*w}U_{\*x_t}^{(\varepsilon)}(\mu_{\*w})$ is $\frac d\gamma\min\{2,e^{\tilde\varepsilon_m}-1\}$, where $\tilde\varepsilon_m=(1+1_{m>1})\max_i\varepsilon_i$.
\end{Lem}
\begin{proof}[Proof for $m=1$; c.f. Appendix~\ref{sec:refined-sensitivity}]
	Let $\*{\tilde x}_t$ be a neighboring dataset of $\*x_t$ and let $U_{\*{\tilde x}_t}^{(\varepsilon)}(\mu_{\*W})=-\log\langle\*{\tilde s}_t,\*w\rangle$ be the corresponding loss.
	Note that $\max_{o\in[a,b]}|\Gap_q(\*x_t,o)-\Gap_q(\*{\tilde x}_t,o)|\le1$ so\vspace{-1mm}
	\begin{align}\label{eq:l2sensitivity}
	\begin{split}
	\*{\tilde s}_{t[j]}
	=\int_{a+\frac{b-a}d(j-1)}^{a+\frac{b-a}dj}\exp\left(-\frac\varepsilon2\Gap_q(\*{\tilde x}_t,o)\right)do
	\in e^{\pm\frac\varepsilon2}\int_{a+\frac{b-a}d(j-1)}^{a+\frac{b-a}dj}\exp\left(-\frac\varepsilon2\Gap_q(\*x_t,o)\right)do
	=e^{\pm\frac\varepsilon2}\*s_{t[j]}\vspace{-1mm}
	\end{split}
	\end{align}
	Therefore since $m=1$ we denote $\*w=\*W_{[1]}$, $\*s_t=\*s_{t,1}$, and $\*{\tilde s}_t=\*{\tilde s}_{t,1}$ and have\vspace{-1mm}
	\begin{align}
	\begin{split}
	\|\nabla_{\*w}U_{\*x_t}^{(\varepsilon)}(\mu_{\*w})-\nabla_{\*w}U_{\*{\tilde x}_t}^{(\varepsilon)}(\mu_{\*w})\|_2
	=\sqrt{\sum_{j=1}^d\left(\frac{\*s_{t[j]}}{\langle\*s_t,\*w\rangle}-\frac{\*{\tilde s}_{t[j]}}{\langle\*{\tilde s}_t,\*w\rangle}\right)^2}
	&=\sqrt{\sum_{j=1}^d\frac{\*s_{t[j]}^2}{\langle\*s_t,\*w\rangle^2}\left(1-\frac{\*{\tilde s}_{t[j]}\langle\*s_t,\*w\rangle}{\*s_{t[j]}\langle\*{\tilde s}_t,\*w\rangle}\right)^2}\\
	&\le\|\nabla_{\*w}U_{\*x_t}^{(\varepsilon)}(\mu_{\*w})\|_1\max_j|1-\kappa_j|\vspace{-1mm}
	\end{split}
	\end{align}
	where $\kappa_j=\frac{\*{\tilde s}_{t[j]}\langle\*s_t,\*w\rangle}{\*s_{t[j]}\langle\*{\tilde s}_t,\*w\rangle}\in\frac{\*s_{t[j]}\exp(\pm\frac\varepsilon2)\langle\*s_t,\*w\rangle}{\*s_{t[j]}\langle\*s_t,\*w\rangle\exp(\pm\frac\varepsilon2)}\in\exp(\pm\varepsilon)$ by Equation~\ref{eq:l2sensitivity}.
	The result follows by taking the minimum with the bound on the Euclidean norm of the gradient (Lemma~\ref{lem:lipschitz}).
\end{proof}
Since $e^\varepsilon-1\le2\varepsilon$ for $\varepsilon\in(0,1.25]$, for small $\varepsilon$ this allows us to add less noise in DP-FTRL.
With this sensitivity bound, we apply Algorithm~\ref{alg:dpftrl} using the entropic regularizer to obtain the following result:%arxiv\looseness-1
\begin{Thm}\label{thm:quantile-regret}
	For $d\ge2,\gamma\in(0,1/2]$ if we run Algorithm~\ref{alg:dpftrl} on $U_{\*x_t}^{(\varepsilon)}(\mu_{\*W})=\log\sum_{i=1}^m\frac1{\Psi_{\*x_t}^{(q_i,\varepsilon_i)}(\mu_{\*W})}$ over $\gamma$-robust priors with step-size $\eta=\frac{\gamma m}d\sqrt{\frac{\log(d)/T}{1+\left(2\sqrt{\log(md)}+\sqrt{2\log\frac T{\beta'}}\right)\sigma\sqrt{\log\lceil\log_2T\rceil}\min\{1,\tilde\varepsilon_m\}}}$ and regularizer $\phi(\*W)=m\langle\*W,\log\*W\rangle$ then for any $V\ge0$, $\lambda\in[0,1]$, and $\beta'\in(0,1]$ we will have regret
	\ifdefined\arxiv
	\begin{align}
	\begin{split}
		\max_{\mu_{[i]}\in\F_{V,d}^{(\lambda)}}&\sum_{t=1}^TU_{\*x_t}^{(\varepsilon)}(\mu_{\*W_t})-U_{\*x_t}^{(\varepsilon)}(\mu)\\
		&\le\frac{VmT}{\gamma d\bar\psi}(b-a)^3+2\max\{\gamma-\lambda,0\}T\log2\\
		&\quad+\frac{2md}\gamma\sqrt{\left(1+\left(4\sqrt{\log(md)}+2\sqrt{2\log\frac T{\beta'}}\right)\sigma\sqrt{\lceil\log_2T\rceil}\min\{1,\tilde\varepsilon_m\}\right)T\log d}
	\end{split}
	\end{align}
	\else
	\begin{align}
	\begin{split}
	\max_{\mu_{[i]}\in\F_{V,d}^{(\lambda)}}\sum_{t=1}^TU_{\*x_t}^{(\varepsilon)}(\mu_{\*W_t})-U_{\*x_t}^{(\varepsilon)}(\mu)
	&\le\frac{VmT}{\gamma d\bar\psi}(b-a)^3+2\max\{\gamma-\lambda,0\}T\log2\\
	&\quad+\frac{2md}\gamma\sqrt{\left(1+\left(4\sqrt{\log(md)}+2\sqrt{2\log\frac T{\beta'}}\right)\sigma\sqrt{\lceil\log_2T\rceil}\min\{1,\tilde\varepsilon_m\}\right)T\log d}
	\end{split}
	\end{align}
	\fi
	w.p. $\ge1-\beta'$, where $\bar\psi$ is the harmonic mean of $\psi_{\*x_t}=\min_k\*x_{t[k+1]}-\*x_{t[k]}$ and $\tilde\varepsilon_m=(1+1_{m>1})\max_i\varepsilon_i$.
	For any $\varepsilon'\le2\log\frac1{\delta'}$ setting $\sigma=\frac1{\varepsilon'}\sqrt{2\lceil\log_2T\rceil\log\frac1{\delta'}}$ makes this procedure $(\varepsilon',\delta')$-DP.\looseness-1
\end{Thm}
\begin{proof}
	For set of $\gamma$-robust priors $\rho$ s.t.  $\rho_{[i]}=\min\{1-\gamma+\lambda,1\}\mu_{[i]}+\frac{\max\{\gamma-\lambda,0\}}{b-a}$ and $\*W\in\triangle_d^m$ s.t. $\*W_{[i,j]}=\frac{b-a}d\int_{a+\frac{b-a}d(j-1)}^{a+\frac{b-a}dj}\rho_{[i]}(o)do$ we can divide the regret into three components:
	\begin{equation}
	\sum_{t=1}^TU_{\*x_t}^{(\varepsilon)}(\mu_{\*W_t})-U_{\*x_t}^{(\varepsilon)}(\mu)
	=\sum_{t=1}^TU_{\*x_t}^{(\varepsilon)}(\mu_{\*W_t})-U_{\*x_t}^{(\varepsilon)}(\mu_{\*W})+\sum_{t=1}^TU_{\*x_t}^{(\varepsilon)}(\mu_{\*W})-U_{\*x_t}^{(\varepsilon)}(\rho)+\sum_{t=1}^TU_{\*x_t}^{(\varepsilon)}(\rho)-U_{\*x_t}^{(\varepsilon)}(\mu)
	\end{equation}
	The first summation is the regret of DP-FTRL with regularizer $\phi$, which is strongly convex w.r.t. $\|\cdot\|_1$.
	The Gaussian width of its unit ball is $2\sqrt{\log(md)}$, by Lemma~\ref{lem:lipschitz} the losses are $\frac d\gamma$-Lipschitz w.r.t. $\|\cdot\|_1$, and by Lemma~\ref{lem:refined-sensitivity} the $\ell_2$-sensitivity is $\Delta_2=\frac d\gamma\min\{2,e^{\tilde\varepsilon_m}-1\}\le\frac{2d}\gamma\min\{1,\tilde\varepsilon_m\}$, so applying Theorem~\ref{thm:dpftrl} yields the bound $\frac{m^2\log d}\eta+\frac{\eta d^2T}{\gamma^2}\left(1+\left(4\sqrt{\log d}+2\sqrt{2\log\frac T{\beta'}}\right)\sigma\sqrt{\lceil\log_2 T\rceil}\min\{1,\varepsilon\}\right)$.
	The second summation is a sum over the errors due to discretization, where we have
	\begin{align}
	\begin{split}
	\sum_{t=1}^TU_{\*x_t}^{(\varepsilon)}(\mu_{\*W})-U_{\*x_t}^{(\varepsilon)}(\rho)
	&=\sum_{t=1}^T\log\sum_{i=1}^m\langle\*s_{t,i},\*W_{[i]}\rangle^{-1}-\log\sum_{i=1}^m\frac1{\int_a^b\exp(-\varepsilon_i\Gap_{q_i}(\*x_t,o)/2)\rho_{[i]}(o)do}\\
	&\le\sum_{t=1}^T\sum_{i=1}^m\frac{\int_a^b\exp(-\frac{\varepsilon_i}2\Gap_{q_i}(\*x_t,o))\rho_{[i]}(o)do-\langle\*s_{t,i},\*W_{[i]}\rangle}{\langle\*s_{t,i},\*W_{[i]}\rangle}\\
	&\le\sum_{t=1}^T\sum_{i=1}^m\frac{\sum_{j=1}^d\int_{a+\frac{b-a}d(j-1)}^{a+\frac{b-a}dj}\exp(-\frac{\varepsilon_i}2\Gap_{q_i}(\*x_t,o))(\rho_{[i]}(o)-\mu_{\*W_{[i]}}(o))do}{\gamma\psi_{\*x_t}/(b-a)}\\
	&\le\sum_{t=1}^T\sum_{i=1}^m\frac{\sum_{j=1}^d\int_{a+\frac{b-a}d(j-1)}^{a+\frac{b-a}dj}|\rho_{[i]}(o)-\rho_{[i]}(o_{i,j})|do}{\gamma\psi_{\*x_t}/(b-a)}
	\le\frac{VmT}{\gamma d\bar\psi}(b-a)^3
	\end{split}
	\end{align}
	where the first inequality follows by concavity, the second by using the definition of $\*W$ to see that $\langle\*s_{t,i},\*W_{[i]}\rangle=\int_a^b\exp(-\frac{\varepsilon_i}2\Gap_{q_i}(\*x_t,o))\mu_{\*W_{[i]}}(o)do\ge\frac{\gamma\psi_{\*x_t}}{b-a}$, the third by H\"older's inequality and the mean value theorem for some $o_{i,j}\in(a+\frac{b-a}d(j-1),a+\frac{b-a}dj)$, and the fourth by the Lipschitzness of $\rho_{[i]}\in\F_{V,d}^{(\gamma)}$.
	The third summation is a sum over the errors due to $\gamma$-robustness, with the result following by $U_{\*x_t}^{(\varepsilon)}(\rho)-U_{\*x_t}^{(\varepsilon)}(\mu)\le U_{\*x_t}^{(\varepsilon)}(\mu)-\log(1-\max\{\gamma-\lambda,0\})-U_{\*x_t}^{(\varepsilon)}(\mu)\le2\max\{\gamma-\lambda,0\}\log2$.
\end{proof}

\newpage
Note that in the case of $V>0$ or $\lambda=0$ we will need to set $d=\omega_T(1)$ or $\gamma=o_T(1)$ in order to obtain sublinear regret.
Thus for these more difficult classes our extension of DP-FTRL to non-Euclidean regularizers yields improved rates, as in the Euclidean case the first term has an extra $\sqrt[4]d$-factor.
The following provides some specific upper bounds derived from Theorem~\ref{thm:quantile-regret}:
\begin{Cor}\label{cor:quantile-regret}
	For each of the following classes of priors there exist settings of $d$ (where needed) and $\gamma>0$ in Theorem~\ref{thm:quantile-regret} that guarantee obtain the following regret w.p. $\ge1-\beta'$:
	\begin{enumerate}
		\item $\lambda$-robust and discrete $\mu_{[i]}\in\F_{0,d}^{(\lambda)}$:  $\tilde\BigO\left(\frac{dm}\lambda\sqrt{\left(1+\frac{\min\{1,\tilde\varepsilon_m\}}{\varepsilon'}\right)T}\right)$
		\item $\lambda$-robust and $V$-Lipschitz $\mu_{[i]}\in\F_{V,1}^{(\lambda)}$: $\tilde\BigO\left(\frac m\lambda\sqrt{\frac V{\bar\psi}}\sqrt[4]{\left(1+\frac{\min\{1,\tilde\varepsilon_m\}}{\varepsilon'}\right)T^3}\right)$
		\item discrete $\mu_{[i]}\in\F_{0,d}$: $\tilde\BigO\left(\sqrt{dm}\sqrt[4]{\left(1+\frac{\min\{1,\tilde\varepsilon_m\}}{\varepsilon'}\right)T^3}\right)$
		\item $V$-Lipschitz $\mu_{[i]}\in\F_{V,1}$: $\tilde\BigO\left(\sqrt m\sqrt[4]{\frac V{\bar\psi}}\sqrt[8]{\left(1+\frac{\min\{1,\tilde\varepsilon_m\}}{\varepsilon'}\right)T^7}\right)$
	\end{enumerate}
\end{Cor}
Thus competing with $\lambda$-robust priors with discrete PDFs enjoys the fastest regret rate of $\tilde\BigO(\sqrt T)$, while either removing robustness or competing with any $V$-Lipschitz prior has regret $\tilde\BigO(T^{3/4})$, and doing both has regret $\tilde\BigO(T^{7/8})$.
When comparing to Lipschitz priors we also incur a dependence on the inverse of minimum datapoint separation, which may be small. 
A notable aspect of all the bounds is that the regret {\em improves} with small $\varepsilon$ due to the sensitivity analysis in Lemma~\ref{lem:refined-sensitivity}; 
indeed for $\varepsilon=\BigO(\varepsilon')$ the regret bound only has a $\BigO(\log\frac1{\delta'})$-dependence on the privacy guarantee.
Finally, for $\lambda$-robust priors we can also apply the $\log\frac{b-a}{\lambda\psi}$-boundedness of $-\log\Psi_{\*x}^{(q,\varepsilon)}(\mu)$ and standard online-to-batch conversion (e.g.~\citet[Proposition~1]{cesa-bianchi2004online2batch} to obtain the following sample complexity guarantee:\looseness-1
\begin{Cor}
	For any $\alpha>0$ and distribution $\D$ over finite datasets $\*x$ of $\psi$-separated points from $(a,b)$, if we run the algorithm in Theorem~\ref{thm:quantile-regret} on $T=\Omega\left(\frac{\log\frac1{\beta'}}{\alpha^2}\left(\frac{d^2m^2}{\lambda^2}\left(1+\frac{\min\{1,\tilde\varepsilon_m\}}{\varepsilon'}\right)+\log^2\frac1{\lambda\psi}\right)\right)$ i.i.d. samples from $\D$ then w.p. $\ge1-\beta'$ the average $\*{\hat W}=\frac1T\sum_{t=1}^T\*W_t$ of the resulting iterates satisfies $\E_{\*x\sim\D}\log\sum_{i=1}^m\frac1{\Psi_{\*x}^{(q_i,\varepsilon_i)}(\mu_{\*{\hat W}_{[i]}})}\le\min_{\mu_{[i]}\in\F_{0,d}^{(\lambda)}}\E_{\*x\sim\D}\log\sum_{i=1}^m\frac1{\Psi_{\*x}^{(q_i,\varepsilon_i)}(\mu_{[i]})}+\alpha$.
	For $\alpha$-suboptimality w.r.t. $\mu_{[i]}\in\F_{V,1}^{(\lambda)}$ the sample complexity is $\Omega\left(\frac{\log\frac1{\beta'}}{\alpha^2}\left(\frac{V^2m^2}{\lambda^4\psi^2\alpha^2}\left(1+\frac{\min\{1,\tilde\varepsilon_m\}}{\varepsilon'}\right)+\log^2\frac1{\lambda\psi}\right)\right)$.
\end{Cor}

\newpage
\section{Section~\ref{sec:applications} details}\label{app:applications}

\subsection{Location-scale families}

A location-scale model is a distribution parameterized by a location $\nu\in\R$ and scale $\sigma\in\R_{\ge0}$ whose density has the form $\mu_{\nu,\sigma}(x)=\frac1\sigma f\left(\frac{x-\nu}\sigma\right)$ for some centered probability measure $f:\R\mapsto\R_{\ge0}$.

\subsubsection{Impossibility of simultaneous robustness and convexity}

\begin{Thm}\label{thm:impossibility}
	Let $f:\R\mapsto\R_{\ge0}$ be a centered probability measure and for each $\theta\in\Theta$ define $\mu_\theta(x)=f(x-\theta)$.
	\begin{enumerate}
		\item If $f$ is continuous then $U_{\*x}(\mu_\theta)$ is convex in $\theta$ for all sorted dataset $\*x\in\R^n$ if and only if $f$ is log-concave.
		\item There exist constants $a,b>0$ s.t. for any $r>0$, $\psi\in(0,\frac R{2n}]$, $q\ge\frac1n$, and $\theta\in\R$ there exists a sorted dataset $\*x\in(\theta\pm R)^n$ with $\min_{i\in[n-1]}\*x_{[i+1]}-\*x_{[i]}=\psi$ s.t. $U_{\*x}^{(q)}(\mu_\theta)=aR+\log\frac b\psi$.
	\end{enumerate}
\end{Thm}
\begin{proof}
	For the first direction of the first result, consider any $\theta,\theta'\in\R$ and $\lambda\in[0,1]$.
	We have that
	\begin{equation}
	U_{*x}^{(q)}(\mu_{\lambda\theta+(1-\lambda)\theta'})
	-\left(\lambda U_{\*x}^{(q)}(\mu_\theta)-(1-\lambda)\log U_{\*x}^{(q)}(\mu_{\theta'}\right)
	=\log\frac{\Psi_{\*x}^{(q)}(\mu_\theta)^\lambda\Psi_{\*x}^{(q)}(\mu_{\theta'})^{1-\lambda}}{\Psi_{\*x}^{(q)}\mu_{\lambda\theta+(1-\lambda)\theta'})}
	\end{equation}
	so it suffices to show that $\Psi_{\*x}{(q)}(\mu_{\lambda\theta+(1-\lambda)\theta'})\ge\Psi_{\*x}^{(q)}(\mu_\theta)^\lambda\Psi_{\*x}^{(q)}(\mu_{\theta'})^{1-\lambda}$.
	By the log-concavity of $f$ we have
	\begin{equation}
	\mu_{\lambda\theta+(1-\lambda)\theta'}(\lambda x+(1-\lambda)y)
	=f(\lambda(x-\theta)+(1-\lambda)(y-\theta'))
	\ge f(x-\theta)^\lambda f(y-\theta')^{1-\lambda}
	=\mu_\theta(x)^\lambda\mu_{\theta'}(y)^{1-\lambda}
	\end{equation}
	for all $x,y\in\R$.
	Therefore by the Pr\'{e}kopa-Leindler inequality we have that
	\begin{align}
	\begin{split}
	\Psi_{\*x}^{(q)}(\mu_{\lambda\theta+(1-\lambda)\theta'})
	=\int_{\*x_{[\lfloor qn\rfloor]}}^{\*x_{[\lfloor qn\rfloor+1]}}\mu_{\lambda\theta+(1-\lambda)\theta'}(x)dx
	&\ge\left(\int_{\*x_{[\lfloor qn\rfloor]}}^{\*x_{[\lfloor qn\rfloor+1]}}\mu_\theta(x)dx\right)^\lambda\left(\int_{\*x_{[\lfloor qn\rfloor]}}^{\*x_{[\lfloor qn\rfloor+1]}}\mu_{\theta'}(x)dx\right)^{1-\lambda}\\
	&=\Psi_{\*x}{(q)}(\mu_\theta)^\lambda\Psi_{\*x}^{(q)}(\mu_{\theta'})^{1-\lambda}
	\end{split}
	\end{align}
	For the second direction, by assumption $\exists~a<c,b>c$ s.t. $\sqrt{f(x)f(y)}>f(\frac{x+y}2)~\forall~x,y\in[a,b]$, i.e. $f$ is strictly log-convex on $[a,b]$.
	Let $\*x\in\R^n$ be any dataset s.t. $\*x_{[\lfloor qn\rfloor+1]}-\*x_{[\lfloor qn\rfloor]}\le\frac{b-a}2$ and set $\theta=\*x_{[\lfloor qn\rfloor]}-a$, $\theta'=\*x_{[\lfloor qn\rfloor]}-\frac{a+b}2$.
	Then we have
	\begin{align}
	\begin{split}
	\sqrt{\int_{\*x_{[\lfloor qn\rfloor]}}^{\*x_{[\lfloor qn\rfloor+1]}}\mu_\theta(x)dx\int_{\*x_{[\lfloor qn\rfloor]}}^{\*x_{[\lfloor qn\rfloor+1]}}\mu_{\theta'}(x)dx}
	&=\sqrt{\int_{\*x_{[\lfloor qn\rfloor]}}^{\*x_{[\lfloor qn\rfloor+1]}}\sqrt{\mu_\theta(x)}^2dx\int_{\*x_{[\lfloor qn\rfloor]}}^{\*x_{[\lfloor qn\rfloor+1]}}\sqrt{\mu_{\theta'}(x)}^2dx}\\
	&\ge\int_{\*x_{[\lfloor qn\rfloor]}}^{\*x_{[\lfloor qn\rfloor+1]}}\sqrt{\mu_\theta(x)\mu_{\theta'}(x)}dx\\
	&=\int_{\*x_{[\lfloor qn\rfloor]}}^{\*x_{[\lfloor qn\rfloor+1]}}\sqrt{f(x-\theta)f(x-\theta')}dx\\
	&>\int_{\*x_{[\lfloor qn\rfloor]}}^{\*x_{[\lfloor qn\rfloor+1]}}f\left(x-\frac{\theta+\theta'}2\right)dx
	=\int_{\*x_{[\lfloor qn\rfloor]}}^{\*x_{[\lfloor qn\rfloor+1]}}\mu_\frac{\theta+\theta'}2(x)dx
	\end{split}
	\end{align}
	where the first inequality is H\"older's and the second is due to the strict log-convexity of $f$ on $[a,b]$.
	Taking the logarithm of both sides followed by their negatives completes the proof.
	
	Finally, for the second result, since $f$ is centered and log-concave, by \citet[Lemma~1]{cule2010theoretical} there exist constants $C,c>0$ s.t. $\mu_\theta(x)\le C\exp(-c|x-\theta|)~\forall~\theta\in\R$.
	Let
	\ifdefined\arxiv
	\begin{equation}
		\*x=\begin{pmatrix}\theta+R-n\psi&\theta+R-(n-1)\psi&\cdots&\theta+R-2\psi&\theta+R-\psi\end{pmatrix}
	\end{equation}
	\else
	 $\*x=\begin{pmatrix}\theta+R-n\psi&\theta+R-(n-1)\psi&\cdots&\theta+R-2\psi&\theta+R-\psi\end{pmatrix}$, 
	 \fi
	 so that $|\*x_{[\lfloor qn\rfloor]}-\theta|\ge|\*x_{[1]}-\theta|=R-n\psi\ge\frac R2$.
	Then
	\begin{equation}
	\Psi_{\*x}^{(q)}(\mu_\theta)
	=\int_{\*x_{[\lfloor qn\rfloor]}}^{\*x_{[\lfloor qn\rfloor+1]}}\mu_\theta(x)dx
	\le C\psi\exp(-c|\*x_{[\lfloor qn\rfloor]}-\theta|)
	\le C\psi\exp(-cR/2)
	\end{equation}
	so $U_{\*x}^{(q)}(\mu)=-\log\Psi_{\*x}^{(q)}(\mu_\theta)\ge\log\frac1{C\psi}+\frac{cR}2$.
\end{proof}

Variants of the first result have been shown in the censored regression literature \cite{burridge1981note,pratt1981concavity}.
In fact, \citet{burridge1981note} shows convexity of $U_{\*x}^{(q)}(\mu_{\frac{\langle\*v,\*f\rangle}\phi,\frac1\phi})$ w.r.t. $(\*v,\phi)\in\R^d\times\R_{>0}$, i.e. simultaneous learning of a feature map and inverse scale. 
Convexity of $U_{\*x}
=-\log\Psi_{\*x}
=\log\sum_{i=1}^m\frac1{\Psi_{\*x}^{(q_i)}}
=\log\sum_{i=1}^m\exp(-\log\Psi_{\*x}^{(q_i)})$ follows because $\log\sum_{i=1}^me^{x_i}$ is convex and non-decreasing in each argument.
Note that for the converse direction, the dataset $\*x$ is not a degenerate case;
in-fact if $f$ is strictly log-convex over an interval $[a,b]$ then any dataset whose optimal interval has length smaller than $\frac{b-a}2$ will yield a non-convex $U_{\*x}^{(q)}(\mu_\theta)$.

\subsubsection{The case of the Laplacian}

For the Laplace prior with $a=\*x_{[\lfloor qn\rfloor]}$ and $b=\*x_{[\lfloor qn\rfloor+1]}$ we have
\ifdefined\arxiv
\begin{align}
\begin{split}
-\log&\Psi_{\*x}^{(q)}(\mu_{\frac\theta\phi,\frac1\phi})\\
&=\log2\\
&\quad-\log\left(\sign\left(b-\frac\theta\phi\right)\left(1-\exp\left(-\left|b-\frac\theta\phi\right|\phi\right)\right)-\sign\left(a-\frac\theta\phi\right)\left(1-\exp\left(-\left|a-\frac\theta\phi\right|\phi\right)\right)\right)
\end{split}
\end{align}
\else
\begin{align}
\begin{split}
-\log&\Psi_{\*x}^{(q)}(\mu_{\frac\theta\phi,\frac1\phi})\\
&=\log2-\log\left(\sign\left(b-\frac\theta\phi\right)\left(1-\exp\left(-\left|b-\frac\theta\phi\right|\phi\right)\right)-\sign\left(a-\frac\theta\phi\right)\left(1-\exp\left(-\left|a-\frac\theta\phi\right|\phi\right)\right)\right)
\end{split}
\end{align}
\fi
For $\theta<a\phi$ this simplifies to
\ifdefined\arxiv
\begin{align}
\begin{split}
\log2-\log\left(e^{\theta-a\phi}-e^{\theta-b\phi}\right)
&=\log2-\log\left((e^{\frac{b-a}2\phi}-e^{\frac{a-b}2\phi})e^{\theta-\frac{a+b}2\phi}\right)\\
&=\left|\theta-\frac{a+b}2\phi\right|-\log\left(\sinh\left(\frac{b-a}2\phi\right)\right)
\end{split}
\end{align}
\else
\begin{equation}
\log2-\log\left(e^{\theta-a\phi}-e^{\theta-b\phi}\right)
=\log2-\log\left((e^{\frac{b-a}2\phi}-e^{\frac{a-b}2\phi})e^{\theta-\frac{a+b}2\phi}\right)
=\left|\theta-\frac{a+b}2\phi\right|-\log\left(\sinh\left(\frac{b-a}2\phi\right)\right)
\end{equation}
\fi
and similarly for $\theta>b\phi$ it becomes
\ifdefined\arxiv
\begin{align}
\begin{split}
\log2-\log\left(e^{b\phi-\theta}-e^{a\phi-\theta}\right)
&=\log2-\log\left((e^{\frac{b-a}2\phi}-e^{\frac{a-b}2\phi})e^{\frac{a+b}2\phi-\theta}\right)\\
&=\left|\frac{a+b}2\phi-\theta\right|-\log\left(\sinh\left(\frac{b-a}2\phi\right)\right)
\end{split}
\end{align}
\else
\begin{equation}
\log2-\log\left(e^{b\phi-\theta}-e^{a\phi-\theta}\right)
=\log2-\log\left((e^{\frac{b-a}2\phi}-e^{\frac{a-b}2\phi})e^{\frac{a+b}2\phi-\theta}\right)
=\left|\frac{a+b}2\phi-\theta\right|-\log\left(\sinh\left(\frac{b-a}2\phi\right)\right)
\end{equation}
\fi
On the other hand for $\theta\in[a\phi,b\phi]$ it is
\begin{align}
\begin{split}
\log2-\log\left(2-e^{-|b\phi-\theta|}-e^{-|a\phi-\theta|}\right)
&=\log2-\log\left(2-e^{\theta-b\phi}-e^{a\phi-\theta}\right)\\
&=\log2-\log\left(e^{-\frac{b-a}2\phi}\left(2e^{\frac{b-a}2\phi}-e^{\theta-\frac{a+b}2\phi}-e^{\frac{a+b}2\phi-\theta}\right)\right)\\
&=\frac{b-a}2\phi+\log2-\log\left(2e^{\frac{b-a}2\phi}-e^{\theta-\frac{a+b}2\phi}-e^{\frac{a+b}2\phi-\theta}\right)\\
&=\frac{b-a}2\phi-\log\left(e^{\frac{b-a}2\phi}-\cosh\left(\theta-\frac{a+b}2\phi\right)\right)
\end{split}
\end{align}
Thus we have
\begin{equation}\label{eq:closed-form}
	U_{\*x}^{(q)}(\mu_{\frac\theta\phi,\frac1\phi})=
	\begin{cases}
		\frac{b-a}2\phi-\log\left(\exp\left(\frac{b-a}2\phi\right)-\cosh\left(\theta-\frac{a+b}2\phi\right)\right) & \textrm{if}\quad\theta\in[a\phi,b\phi]\\
		\left|\theta-\frac{a+b}2\phi\right|-\log\left(\sinh\left(\frac{b-a}2\phi\right)\right) & \textrm{else}\\
	\end{cases}
\end{equation}

Suppose $\*x\in[\pm B]^n$ and has the optimal interval has separation $\psi>0$,  $\frac\theta\phi\in[\pm B]$, and $\frac1\phi\in[\sigma_{\min},\sigma_{\max}]$.
Then $\phi\in[1/\sigma_{\max},1/\sigma_{\min}]$ and $\theta\in[\pm B/\sigma_{\min}]$, and so
\begin{equation}\label{eq:laplace-bound}
U_{\*x}^{(q)}(\mu_{\frac\theta\phi,\frac1\phi})
\le\frac{2B}{\sigma_{\min}}+\log\frac{2\sigma_{\max}}{\psi}	
\end{equation}

For $\theta\not\in[a\phi,b\phi]$, the derivative w.r.t. $\theta$ always has magnitude 1.
Within the interval, the derivative w.r.t. $\theta$ is $-\frac{\sinh(\frac{a+b}2\phi-\theta)}{\exp(\frac{b-a}2\phi)-\cosh(\theta-\frac{a+b}2\phi)}$, which attains its extrema at the endpoints $a\phi$ and $b\phi$, where its magnitude is also 1.
Outside the interval, the derivative w.r.t. $\phi$ has magnitude 
\ifdefined\arxiv
\begin{align}
\begin{split}
\left|\frac{a+b}2\sign\left(\frac{a+b}2\phi-\theta\right)-\frac{b-a}2\coth\left(\frac{b-a}2\phi\right)\right|
&\le\frac{|a+b|}2+\frac{b-a}2\coth\left(\frac{b-a}2\phi\right)\\
&\le\frac{|a+b|}2+\frac{b-a}2\left(\frac{2/\phi}{(b-a)}+1\right)\\
&=\frac{|a+b|}2+\frac{b-a}2+\frac1\phi
\end{split}
\end{align}
\else
\begin{align}
\begin{split}
\left|\frac{a+b}2\sign\left(\frac{a+b}2\phi-\theta\right)-\frac{b-a}2\coth\left(\frac{b-a}2\phi\right)\right|
&\le\frac{|a+b|}2+\frac{b-a}2\coth\left(\frac{b-a}2\phi\right)\\
&\le\frac{|a+b|}2+\frac{b-a}2\left(\frac{2/\phi}{(b-a)}+1\right)
=\frac{|a+b|}2+\frac{b-a}2+\frac1\phi
\end{split}
\end{align}
\fi
while inside the interval the derivative w.r.t. $\phi$ is $\frac{b-a}2-\frac{(b-a)\exp(\frac{b-a}2\phi)-(a+b)\sinh(\frac{a+b}2\phi-\theta)}{2\left(\exp(\frac{b-a}2\phi)-\cosh(\frac{a+b}2\phi-\theta)\right)}
$, which again attains its extrema at the endpoints $a\phi$ and $b\phi$, yielding magnitudes 
\ifdefined\arxiv
\begin{align}
\begin{split}
\frac{b-a}2+\frac{b-a}2\left(\coth\left(\frac{b-a}2\phi\right)+1\right)+\frac{|a+b|}2
&\le\frac{b-a}2\left(\frac{2/\phi}{(b-a)}+3\right)+\frac{|a+b|}2\\
&\le\frac1\phi+\frac32(b-a)+\frac{|a+b|}2
\end{split}
\end{align}
\else
\begin{equation}
	\frac{b-a}2+\frac{b-a}2\left(\coth\left(\frac{b-a}2\phi\right)+1\right)+\frac{|a+b|}2
	\le\frac{b-a}2\left(\frac{2/\phi}{(b-a)}+3\right)+\frac{|a+b|}2
	\le\frac1\phi+\frac32(b-a)+\frac{|a+b|}2
\end{equation}
\fi
Thus we have
\begin{equation}
|\partial_\theta U_{\*x}^{(q)}(\mu_{\frac\theta\phi,\frac1\phi})|
\le1
\qquad\textrm{and}\qquad
|\partial_\phi U_{\*x}^{(q)}(\mu_{\frac\theta\phi,\frac1\phi})|
\le4B+\sigma_{\max}
\end{equation}
\ifdefined\arxiv\newpage\else\fi

\subsection{Public-private release}

\subsubsection{Guarantees}

\begin{Thm}\label{thm:pubpri}
	Suppose for $N\ge n$ we have a private dataset $\*x\sim\D^n$ and a public dataset $\*x'\sim{\D'}^N$, both drawn from $\kappa$-bounded distributions over $[\pm B]$.
	Use i.i.d. draws from the public dataset to construct $T=\lfloor N/n\rfloor$ datasets $\*x_t'\sim{\D'}^n$ and run online gradient descent on the resulting losses $\ell_{\*x_t'}(\theta,\psi)=\LSE_i(\ell_{\*x_t}^{(q_i)}(\theta_{[i]},\phi_{[i]}))$ over the parameter space $\theta\in[\pm B/\sigma_{\min}]^m$ starting at $\theta=\*0_m$ and $\phi\in[1/\sigma_{\max},1/\sigma_{\min}]^m$ starting at the midpoint, with stepsize $B\sqrt{\frac mT}$ for $\theta$ and $\frac{\sigma_{\max}-\sigma_{\min}}{4B+\sigma_{\max}}\sqrt{\frac mT}$ for $\phi$, obtaining iterates $(\theta_1,\phi_1),\dots,(\theta_T,\phi_T)$.
	Return the priors $\mu_i=\mu_{\frac{\bar\theta_{[i]}}{\bar\phi_{[i]}},\frac1{\bar\phi_{[i]}}}$ for $\bar\theta=\frac1T\sum_{t=1}^T\theta_t$ and $\bar\phi=\frac1T\sum_{t=1}^T\phi_t$ the average of these iterates.
	Then $\mu'=\begin{pmatrix}\mu_1&\cdots&\mu_m\end{pmatrix}$ satisfies 
	\ifdefined\arxiv
	\begin{align}
	\begin{split}
	&\E_{\*x\sim\D^n}U_{\*x}(\mu')\\
	&\quad\le\min_{\mu\in\Lap_{B,\sigma_{\min},\sigma_{\max}}^m}\E_{\*x\sim\D^n}U_{\*x}(\mu)+2\left(\frac{2B}{\sigma_{\min}}+\log\frac{4\kappa m(n+1)N\sigma_{\max}}{\beta'}\right)\TV_q(\D,\D')\\
	&\qquad+(B+4B\sigma_{\max}+\sigma_{\max}^2)\sqrt{\frac{m(n+1)}N}+2\left(\frac{4B}{\sigma_{\min}}+\log\frac{4\kappa m(n+1)N\sigma_{\max}}{\beta'}\right)\sqrt{\frac{2(n+1)}N\log\frac4{\beta'}}\\
	&\qquad+\frac{(n+1)\beta'}N\left(3+\frac{4B}{\sigma_{\min}}+4\log\frac{2\kappa(n+1)N\sqrt{2m\sigma_{\max}}}{\beta'}\right)
	\end{split}
	\end{align}
	\else
	\begin{align}
	\begin{split}
	\E_{\*x\sim\D^n}U_{\*x}(\mu')
	&\le\min_{\mu\in\Lap_{B,\sigma_{\min},\sigma_{\max}}^m}\E_{\*x\sim\D^n}U_{\*x}(\mu)+2\left(\frac{2B}{\sigma_{\min}}+\log\frac{4\kappa m(n+1)N\sigma_{\max}}{\beta'}\right)\TV_q(\D,\D')\\
	&\quad+(B+4B\sigma_{\max}+\sigma_{\max}^2)\sqrt{\frac{m(n+1)}N}+2\left(\frac{4B}{\sigma_{\min}}+\log\frac{4\kappa m(n+1)N\sigma_{\max}}{\beta'}\right)\sqrt{\frac{2(n+1)}N\log\frac4{\beta'}}\\
	&\quad+\frac{(n+1)\beta'}N\left(3+\frac{4B}{\sigma_{\min}}+4\log\frac{2\kappa(n+1)N\sqrt{2m\sigma_{\max}}}{\beta'}\right)
	\end{split}
	\end{align}
	\fi
	where $\Lap_{B,\sigma_{\min},\sigma_{\max}}$ is the set of Laplace priors with locations in $[\pm B]$ and scales in $[\sigma_{\min},\sigma_{\max}]$.
\end{Thm}

\begin{proof}
	Define ${\D_\psi'}^n$ to be the conditional distribution over $\*z\sim{\D'}^n$ s.t. $\psi_{\*z}\ge\psi$, with associated density $\rho_\psi'(\*z)=\frac{\rho'(\*z)1_{\psi_{\*z}\ge\psi}}{1-p_\psi'}$, where $p_\psi'=\int_{\psi_{\*z}<\psi}\rho'(\*z)\le\kappa n^2\psi$.
	Then we have for any $\mu^\ast\in\Lap_{B,\sigma_{\min},\sigma_{\max}}^m$ that
	\ifdefined\arxiv
	\begin{align}
	\begin{split}
	&\E_{\*z\sim\D^n}U_{\*x}(\mu')\\
	&\quad=\E_{\*z\sim\D^n}U_{\*z}(\mu')-\E_{\*z\sim{\D'}^n}U_{\*z}(\mu')+\E_{\*z\sim{\D'}^n}U_{\*z}(\mu')-\E_{\*z\sim{\D_\psi'}^n}U_{\*z}(\mu')+\E_{\*z\sim{\D_\psi'}^n}U_{\*z}(\mu')\\
	&\quad\le\int U_{\*z}(\mu')(\rho(\*z)-\rho'(\*z))+\int U_{\*z}(\mu')(\rho'(\*z)-\rho_\psi'(\*x))+\E_{\*z\sim{\D_\psi'}^n}U_{\*z}(\mu^\ast)+\Err_\psi\\
	&\quad\le\E_{\*z\sim\D^n}U_{\*x}(\mu^\ast)+\int(U_{\*z}(\mu')+U_{\*x}(\mu^\ast))|\rho(\*z)-\rho'(\*x)|+\int(U_{\*z}(\mu')+U_{\*z}(\mu^\ast))|\rho'(\*z)-\rho_\psi'(\*z)|+\Err_\psi
	\end{split}
	\end{align}
	\else
	\begin{align}
	\begin{split}
		\E_{\*z\sim\D^n}U_{\*x}(\mu')
		&=\E_{\*z\sim\D^n}U_{\*z}(\mu')-\E_{\*z\sim{\D'}^n}U_{\*z}(\mu')+\E_{\*z\sim{\D'}^n}U_{\*z}(\mu')-\E_{\*z\sim{\D_\psi'}^n}U_{\*z}(\mu')+\E_{\*z\sim{\D_\psi'}^n}U_{\*z}(\mu')\\
		&\le\int U_{\*z}(\mu')(\rho(\*z)-\rho'(\*z))+\int U_{\*z}(\mu')(\rho'(\*z)-\rho_\psi'(\*x))+\E_{\*z\sim{\D_\psi'}^n}U_{\*z}(\mu^\ast)+\Err_\psi\\
		&\le\E_{\*z\sim\D^n}U_{\*x}(\mu^\ast)+\int(U_{\*z}(\mu')+U_{\*x}(\mu^\ast))|\rho(\*z)-\rho'(\*x)|+\int(U_{\*z}(\mu')+U_{\*z}(\mu^\ast))|\rho'(\*z)-\rho_\psi'(\*z)|+\Err_\psi
	\end{split}
	\end{align}
	\fi
	where $\Err_\psi$ is the error of running online gradient descent with the specified step-sizes on samples $\*z_t'\sim{\D_\psi'}^n$ for $t=1,\dots,T$.
	Now if $\*z$ has entries drawn i.i.d. from a $\kappa$-bounded distribution $\D^n$ (or ${\D'}^n$), then we have that 
	\begin{equation}
		\int_0^\psi\rho_{\psi_{\*z}}(y)dy=\Pr(\psi_{\*z}
		\le\psi:\*z\sim\D^n)\le n(n-1)\max_{z\in\R}\Pr(|z-z'|\le\psi:z'\sim\D)
		\le\kappa n^2\psi
	\end{equation}
	where $\rho_{\psi_{\*z}}$ is the density of $\psi_{\*z}$ for $\*z\sim\D^n$ (not to be confused with the conditional density $\rho_\psi$ over $\*z$);
	the same holds for the analog $\rho_{\psi_{\*z}}'$ for ${\D'}^n$.
	Since this holds for all $\psi\ge0$ and $\log\frac1y$ is monotonically decreasing on $y>0$, this means the worst-case measure that $\rho_{\psi_{\*z}}$ can be is constant over $[0,\psi]$ and thus $\int_0^\psi\rho_{\psi_{\*z}}(y)\log\frac1ydy\le\kappa n^2\int_0^\psi\log\frac1ydy=\kappa n^2\psi(1+\log\frac1\psi)$, and similarly for $\rho_{\psi_{\*z}}'$.
	We then bound the first integral, noting that $U_{\*z}=\LSE_i(U_{\*z}^{(q_i)})\le\max_iU_{\*z}^{(q_i)}+\log m\le\frac{2B}{\sigma_{\min}}+\log\frac{2m\sigma_{\max}}{\psi_{\*z}}$ and that the r.v. $\psi_{\*z}$ depends only on the joint distribution over the order statistics of $\D^n$ and ${\D'}^n$:
	\ifdefined\arxiv
	\begin{align}
	\begin{split}
	\int(U_{\*z}(\mu')+U_{\*z}(\mu^\ast))&|\rho(\*z)-\rho'(\*z)|\\
	&\le\int\left(\frac{2B}{\sigma_{\min}}+\log\frac{2m\sigma_{\max}}{\psi_{\*z}}\right)|\rho(\*z)-\rho'(\*z)|\\
	&\le2\left(\frac{2B}{\sigma_{\min}}+\log\frac{2m\sigma_{\max}}\psi\right)\TV_q(\D,\D')+\int_{\psi_{\*z}<\psi}|\rho(\*z)-\rho'(\*z)|\log\frac1{\psi_{\*z}}\\
	&\le2\left(\frac{2B}{\sigma_{\min}}+\log\frac{2m\sigma_{\max}}\psi\right)\TV_q(\D,\D')+\int_0^\psi(\rho_{\psi_{\*z}}(y)+\rho_{\psi_{\*z}}'(y))\log\frac1ydy\\
	&\le2\left(\frac{2B}{\sigma_{\min}}+\log\frac{2m\sigma_{\max}}\psi\right)\TV_q(\D,\D')+2\kappa n^2\psi\left(1+\log\frac1\psi\right)\vspace{-1mm}
	\end{split}
	\end{align}
	\else
	\begin{align}
	\begin{split}
		\int(U_{\*z}(\mu')+U_{\*z}(\mu^\ast))|\rho(\*z)-\rho'(\*z)|
		&\le\int\left(\frac{2B}{\sigma_{\min}}+\log\frac{2m\sigma_{\max}}{\psi_{\*z}}\right)|\rho(\*z)-\rho'(\*z)|\\
		&\le2\left(\frac{2B}{\sigma_{\min}}+\log\frac{2m\sigma_{\max}}\psi\right)\TV_q(\D,\D')+\int_{\psi_{\*z}<\psi}|\rho(\*z)-\rho'(\*z)|\log\frac1{\psi_{\*z}}\\
		&\le2\left(\frac{2B}{\sigma_{\min}}+\log\frac{2m\sigma_{\max}}\psi\right)\TV_q(\D,\D')+\int_0^\psi(\rho_{\psi_{\*z}}(y)+\rho_{\psi_{\*z}}'(y))\log\frac1ydy\\
		&\le2\left(\frac{2B}{\sigma_{\min}}+\log\frac{2m\sigma_{\max}}\psi\right)\TV_q(\D,\D')+2\kappa n^2\psi\left(1+\log\frac1\psi\right)\vspace{-1mm}
	\end{split}
	\end{align}
	\fi
	For the second integral we have for $p_\psi'=\int_{\psi_{\*z}<\psi}\rho'(\*z)\le\kappa n^2\psi$ that\vspace{-1mm}
	\begin{align}
	\begin{split}
		\int(U_{\*z}(\mu')&+U_{\*z}(\mu^\ast))|\rho'(\*z)-\rho_\psi'(\*z)|\\
		&=\int_{\psi_{\*z}\ge\psi}(U_{\*x}(\mu')+U_{\*z}(\mu^\ast))\left|\rho'(\*z)-\frac{\rho'(\*z)}{1-p_\psi'}\right|+\int_{\psi_{\*z}<\psi}(U_{\*z}(\mu')+U_{\*z}(\mu^\ast))\rho'(\*z)\\
		&=\frac{2p_\psi'}{1-p_\psi'}\int_{\psi_{\*z}\ge\psi}\left(\frac{2B}{\sigma_{\min}}+\log\frac{2m\sigma_{\max}}\psi\right)\rho'(\*z)+\int_{\psi_{\*z}<\psi}\left(\frac{2B}{\sigma_{\min}}+\log\frac{2m\sigma_{\max}}{\psi_{\*z}}\right)\rho'(\*z)\\
		&=2p_\psi'\left(\frac{4B}{\sigma_{\min}}+\log\frac{4m^2\sigma_{\max}^2}\psi\right)+\int_{\psi_{\*z}<\psi}\rho'(\*z)\log\frac1{\psi_{\*z}}\\
		&\le2\kappa n^2\psi\left(\frac{4B}{\sigma_{\min}}+\log\frac{4m^2\sigma_{\max}^2}\psi\right)+\kappa n^2\psi\left(1+\log\frac1\psi\right)\vspace{-1mm}
	\end{split}
	\end{align}
	Finally, we bound $\Err_\psi$.
	By $\kappa$-boundedness of $\D'$, the probability that $\exists~t\in[T]$ s.t. $\psi_{\*z_t'}<\psi~\forall~t\in[T]$ is at most $\kappa n^2T\psi$, so if we set $\psi=\frac{\beta'}{2\kappa n^2T}$ then w.p. $\ge1-\beta'/2$ the sampling $\*z_t'$ from $\*x'$ as specified is equivalent to rejection sampling from ${\D_\psi'}^n$, on which the functions $U_{\*z}$ are bounded by $\frac{2B}{\sigma_{\min}}+\log\frac{2m\sigma_{\max}}\psi$.
	Therefore with probability $\ge 1-\beta'/2$ by \citet[Theorem~2.21]{shalev-shwartz2011oco} and Theorem~\ref{thm:o2b} we have that w.p. $1-\beta'/2$\vspace{-1mm}
	\begin{align}
	\begin{split}
	\Err_\psi
	&\le(B+(\sigma_{\max}-\sigma_{\min})(4B+\sigma_{\max}))\sqrt{\frac mT}+2\left(\frac{4B}{\sigma_{\min}}+\log\frac{2m\sigma_{\max}}\psi\right)\sqrt{\frac2T\log\frac4{\beta'}}\\
	&=(B+4B\sigma_{\max}+\sigma_{\max}^2)\sqrt{\frac{m(n+1)}N}+2\left(\frac{4B}{\sigma_{\min}}+\log\frac{2m\sigma_{\max}}\psi\right)\sqrt{\frac{2(n+1)}N\log\frac4{\beta'}}\vspace{-1mm}
	\end{split}
	\end{align}
	Combining terms and substituting the selected value for $\psi$ yields the result.
\end{proof}
\ifdefined\arxiv\newpage\else\fi

\subsubsection{Experimental details}\label{sec:pubpri-details}

For our public-private experiments we evaluate several methods on the Adult (``age" and ``hours" categories) and Goodreads (``rating" and ``page count" categories).
For the former we use the train set as the public data, while for the latter we use the ``History" genre as the public data and the ``Poetry" genre as the private data~\cite{wan2018goodreads}.
The public data are used to fit Laplace location and scale parameters using the COCOB optimizer run until progress stops.
We use the implementation here: \url{https://github.com/anandsaha/nips.cocob.pytorch}.
All evaluations are averages of forty trials.\looseness-1

We use the following reasonable guesses for locations $\nu$, scales $\sigma$, and quantile ranges $[a,b]$ for these distributions:
\ifdefined\arxiv
\begin{itemize}[itemsep=1pt]
\else
\begin{itemize}[leftmargin=*,topsep=-2pt,noitemsep]\setlength\itemsep{1pt}
\fi
	\item age: $\nu=40$, $\sigma=5$, $a=10$, $b=120$
	\item hours: $\nu=40$, $\sigma=2$, $a=0$, $b=168$
	\item rating: $\nu=2.5$, $\sigma=0.5$, $a=0$, $b=5$
	\item page count: $\nu=200$, $\sigma=25$, $a=0$, $b=\frac{1000}{1-q}$
\end{itemize}
Note that, here and elsewhere, using $q$-dependent range for $b$ only helps the Uniform prior, which is the baseline.
The scales $\sigma$ are used to set the scale parameter of the Cauchy distribution for \texttt{public quantiles}---its location is fixed by the public quantiles.
Meanwhile the locations $\nu$ are used to set to {\em scale} parameter of the half-Cauchy prior used to mix with \texttt{PubFit} for robustness (using coefficient 0.1 on the robust prior).
We choose this prior because the data are all nonnegative.

\subsection{Sequential release}

\subsubsection{Guarantees}

\begin{Thm}\label{thm:sequential}
	Consider a sequence of datasets $\*x_t\in[\pm R]^{n_t}$ and associated feature vectors $\*f_t\in[\pm F]^d$.
	Suppose we set the component priors $\mu_{t,i}$ of $\mu_t$ as the Laplace distributions $\mu_{t,i}=\mu_{\frac{\langle\*v_{t,i}\*f_t\rangle}{\phi_{t,i}},\frac1{\phi_{t,i}}}$, where $\*v_{t,i}\in[\pm B/\sigma_{\min}]^d$ and $\phi_i\in[1/\sigma_{\max},1/\sigma_{\min}]$ are determined by separate runs of DP-FTRL with budgets $(\varepsilon'/2,\delta'/2)$ and step-sizes $\eta_1=\frac B{F\sigma_{\min}}\sqrt{\frac{2m\varepsilon_1'}{\lceil\log_2(T+1)\rceil T\left(1+\sqrt{2md\log\frac T{\beta'}\log\frac1{\delta'}}\right)}}$, and $\eta_2=\frac{1/\sigma_{\min}}{B+\sigma_{\max}}\sqrt{\frac{m\varepsilon_2'}{2\lceil\log_2(T+1)\rceil T\left(1+\sqrt{2m\log\frac T{\beta'}\log\frac1{\delta'}}\right)}}$.
	Then we have regret
	\ifdefined\arxiv
	\begin{align}
	\begin{split}
		\max_{\begin{smallmatrix}\*w_i\in[\pm B]^d\\\sigma_i\in[\sigma_{\min},\sigma_{\max}]\end{smallmatrix}}&\sum_{t=1}^TU_{\*x_t}(\mu_t)-U_{\*x_t}(\mu_{\langle\*w_i,\*f_t\rangle,\sigma_i})\\
		&\le\frac{B(F+1)+\sigma_{\max}}{\sigma_{\min}}\sqrt{md\lceil\log_2(T+1)\rceil T\left(4+\frac8{\varepsilon'}\sqrt{2md\log\frac T{\beta'}\log\frac2{\delta'}}\right)}
	\end{split}
	\end{align}
	\else
	\begin{equation}
		\max_{\begin{smallmatrix}\*w_i\in[\pm B]^d\\\sigma_i\in[\sigma_{\min},\sigma_{\max}]\end{smallmatrix}}\sum_{t=1}^TU_{\*x_t}(\mu_t)-U_{\*x_t}(\mu_{\langle\*w_i,\*f_t\rangle,\sigma_i})
		\le\frac{B(F+1)+\sigma_{\max}}{\sigma_{\min}}\sqrt{md\lceil\log_2(T+1)\rceil T\left(4+\frac8{\varepsilon'}\sqrt{2md\log\frac T{\beta'}\log\frac2{\delta'}}\right)}
	\end{equation}
	\fi
	For sufficiently small $\varepsilon'$ (including $\varepsilon'\le 1$) we can instead simplify the regret to
	\begin{equation}
		\frac4{\sigma_{\min}}\left(BFd^\frac34+B+\sigma_{\max}\right)\sqrt{\frac{m\lceil\log_2(T+1)\rceil T}{\varepsilon'}\sqrt{2m\log\frac T{\beta'}\log\frac2{\delta'}}}
	\end{equation}
\end{Thm}
\begin{proof}
	Note that
	\begin{equation}
		\sum_{j=1}^m\|\nabla_{\*v_j}\LSE_i(\ell_{\*x_t,\*f_t}^{(q_i)})\|_2^2
		\le\|\*f_t\|_2^2\sum_{j=1}^m\left(\frac{\exp(\ell_{\*x_t,\*f_t}^{(q_j)})}{\sum_{i=1}^m\exp(\ell_{\*x_t,\*f_t}^{(q_i)})}\right)^2
		\le F^2d
	\end{equation}
	and 
	\begin{equation}
		\sum_{j=1}^m(\partial_{\phi_j}\LSE_i(\ell_{\*x_t,\*f_t}^{(q_i)}))^2
		\le(4B+\sigma_{\max})^2\sum_{j=1}^m\left(\frac{\exp(\ell_{\*x_t,\*f_t}^{(q_j)})}{\sum_{i=1}^m\exp(\ell_{\*x_t,\*f_t}^{(q_i)})}\right)^2
		\le(4B+\sigma_{\max})^2
	\end{equation}
	and so applying Theorem~\ref{thm:dpftrl} twice with the assumed budgets and step-sizes yields
	\ifdefined\arxiv
	\begin{align}
	\begin{split}
	&\max_{\begin{smallmatrix}\*w_i\in[\pm B]^d\\\sigma_i\in[\sigma_{\min},\sigma_{\max}]\end{smallmatrix}}\sum_{t=1}^TU_{\*x_t}(\mu_t)-U_{\*x_t}(\mu_{\langle\*w_i,\*f_t\rangle,\sigma_i})\\
	&\quad=\max_{\begin{smallmatrix}\*v_i\in[\pm\frac B{\sigma_{\min}}]^d\\\phi_i\in[\frac1\sigma_{\max},\frac1\sigma_{\min}]\end{smallmatrix}}\sum_{t=1}^T\LSE_i(\ell_{\*x_t,\*f_t}^{(q_i)}(\*v_{t,i},\phi_{t,i}))-\LSE_i(\ell_{\*x_t,\*f_t}^{(q_i)}(\*v_i,\phi_i))\\
	&\quad\le\sum_{i=1}^m\frac{\|\*v_{1,i}-\*v_i\|_2^2}{2\eta_1}+\eta_1\lceil\log_2(T+1)\rceil T\left(1+\frac2{\varepsilon'}\sqrt{2md\log\frac T{\beta'}\log\frac2{\delta'}}\right)\sum_{j=1}^m\|\nabla_{\*v_j}\LSE_i(\ell_{\*x_t,\*f_t}^{(q_i)})\|_2^2\\
	&\quad\qquad+\sum_{i=1}^m\frac{(\phi_{1,i}-\phi_i)^2}{2\eta_2}+\eta_2\lceil\log_2(T+1)\rceil T\left(1+\frac2{\varepsilon'}\sqrt{2m\log\frac T{\beta'}\log\frac2{\delta'}}\right)\sum_{j=1}^m(\partial_{\phi_j}\LSE_i(\ell_{\*x_t,\*f_t}^{(q_i)}))^2\\
	&\quad\le\frac{2B^2md}{\eta_1\sigma_{\min}^2}+\eta_1\lceil\log_2(T+1)\rceil TF^2d\left(1+\frac2{\varepsilon'}\sqrt{2md\log\frac T{\beta'}\log\frac2{\delta'}}\right)\\
	&\quad\qquad+\frac m{2\eta_2\sigma_{\min}^2}+\eta_2\lceil\log_2(T+1)\rceil T(B+\sigma_{\max})^2\left(1+\frac2{\varepsilon'}\sqrt{2m\log\frac T{\beta'}\log\frac2{\delta'}}\right)\\
	&\quad\le\frac{2BF}{\sigma_{\min}}\sqrt{2md\lceil\log_2(T+1)\rceil T\left(1+\frac2{\varepsilon'}\sqrt{2md\log\frac T{\beta'}\log\frac2{\delta'}}\right)}\\
	&\quad\qquad+\frac2{\sigma_{\min}}(B+\sigma_{\max})\sqrt{2m\lceil\log_2(T+1)\rceil T\left(1+\frac2{\varepsilon'}\sqrt{2m\log\frac T{\beta'}\log\frac2{\delta'}}\right)}\\
	&\quad\le\frac2{\sigma_{\min}}\left(B(F+1)+\sigma_{\max}\right)\sqrt{md\lceil\log_2(T+1)\rceil T\left(1+\frac2{\varepsilon'}\sqrt{2md\log\frac T{\beta'}\log\frac2{\delta'}}\right)}
	\end{split}
	\end{align}
	\else
	\begin{align}
	\begin{split}
		\max_{\begin{smallmatrix}\*w_i\in[\pm B]^d\\\sigma_i\in[\sigma_{\min},\sigma_{\max}]\end{smallmatrix}}&\sum_{t=1}^TU_{\*x_t}(\mu_t)-U_{\*x_t}(\mu_{\langle\*w_i,\*f_t\rangle,\sigma_i})
		=\max_{\begin{smallmatrix}\*v_i\in[\pm\frac B{\sigma_{\min}}]^d\\\phi_i\in[\frac1\sigma_{\max},\frac1\sigma_{\min}]\end{smallmatrix}}\sum_{t=1}^T\LSE_i(\ell_{\*x_t,\*f_t}^{(q_i)}(\*v_{t,i},\phi_{t,i}))-\LSE_i(\ell_{\*x_t,\*f_t}^{(q_i)}(\*v_i,\phi_i))\\
		&\le\sum_{i=1}^m\frac{\|\*v_{1,i}-\*v_i\|_2^2}{2\eta_1}+\eta_1\lceil\log_2(T+1)\rceil T\left(1+\frac2{\varepsilon'}\sqrt{2md\log\frac T{\beta'}\log\frac2{\delta'}}\right)\sum_{j=1}^m\|\nabla_{\*v_j}\LSE_i(\ell_{\*x_t,\*f_t}^{(q_i)})\|_2^2\\
		&\qquad+\sum_{i=1}^m\frac{(\phi_{1,i}-\phi_i)^2}{2\eta_2}+\eta_2\lceil\log_2(T+1)\rceil T\left(1+\frac2{\varepsilon'}\sqrt{2m\log\frac T{\beta'}\log\frac2{\delta'}}\right)\sum_{j=1}^m(\partial_{\phi_j}\LSE_i(\ell_{\*x_t,\*f_t}^{(q_i)}))^2\\
		&\le\frac{2B^2md}{\eta_1\sigma_{\min}^2}+\eta_1\lceil\log_2(T+1)\rceil TF^2d\left(1+\frac2{\varepsilon'}\sqrt{2md\log\frac T{\beta'}\log\frac2{\delta'}}\right)\\
		&\qquad+\frac m{2\eta_2\sigma_{\min}^2}+\eta_2\lceil\log_2(T+1)\rceil T(B+\sigma_{\max})^2\left(1+\frac2{\varepsilon'}\sqrt{2m\log\frac T{\beta'}\log\frac2{\delta'}}\right)\\
		&\le\frac{2BF}{\sigma_{\min}}\sqrt{2md\lceil\log_2(T+1)\rceil T\left(1+\frac2{\varepsilon'}\sqrt{2md\log\frac T{\beta'}\log\frac2{\delta'}}\right)}\\
		&\qquad+\frac2{\sigma_{\min}}(B+\sigma_{\max})\sqrt{2m\lceil\log_2(T+1)\rceil T\left(1+\frac2{\varepsilon'}\sqrt{2m\log\frac T{\beta'}\log\frac2{\delta'}}\right)}\\
		&\le\frac2{\sigma_{\min}}\left(B(F+1)+\sigma_{\max}\right)\sqrt{md\lceil\log_2(T+1)\rceil T\left(1+\frac2{\varepsilon'}\sqrt{2md\log\frac T{\beta'}\log\frac2{\delta'}}\right)}
	\end{split}
	\end{align}
	\fi
\end{proof}
\ifdefined\arxiv\newpage\else\fi

\subsubsection{Experimental details}\label{sec:online-details}

For sequential release we consider the following tasks:
\begin{itemize}
	\item Synthetic is a stationary dataset generation scheme in which we randomly sample a one standard Gaussian vector $\*a$ for each feature dimension (we use ten) and another $\*b$ of size $m+2$, which we sort.
	On each day $t$ of $T$ we sample the public feature vector $\*f_t$, also from a standard normal, and the ``ground truth" quantiles $q_i$ on that day are then set by $\langle\*a,\*f_t\rangle+\*b_{[i+1]}$.
	We generate the actual data by sampling from the uniform distributions on $[\langle\*a,\*f_t\rangle+\*b_{[i]},\langle\*a,\*f_t\rangle+\*b_{[i+1]}]$.
	The number of points we sample is determined by $\lfloor 100/(m+1)\rfloor$ plus different Poisson-distributed random variable for each;
	in the ``noiseless" setting used in Figure~\ref{fig:timeplots} (left) the Poisson's scale is zero, so the ``ground truth" quantiles are correct for the dataset, while for Figure~\ref{fig:online-logepsplots} (left) we use a Poisson with scale five.
	For the noiseless setting we use 100K timesteps, while for the noisy setting we use 2500.
	\item CitiBike consists of data downloaded from here: \url{https://s3.amazonaws.com/tripdata/index.html},
	We take the period from September 2015 through November 2022, which is roughly 2500 days, although days with less than ten trips---seemingly data errors---are ignored.
	For each day we include a feature vector containing seven dimensions for the day of the week, one dimension for a sinusoidal encoding of the day of the year, and six weather features from the Central Park station downloaded from here \url{https://www.ncei.noaa.gov/cdo-web/}, specifically average wind speed, precipitation, snowfall, snow depth, maximum temperature, and minimum temperature.
	These are scaled to lie within similar ranges.
	\item BBC consists of Reddit's \texttt{worldnews} corpus downloaded from here: \url{https://zissou.infosci.cornell.edu/convokit/datasets/subreddit-corpus/corpus-zipped/}.
	We find all conversations corresponding to a post of a BBC article, specified by the domain \url{bbc.co.uk}, and collect those with at least ten comments.
	We compute the Flesch readability score of each comment using the package here \url{https://github.com/textstat/textstat}.
	The datasets for computing quantiles are then the collection of scores for each headline;
	the size is roughly 10K, corresponding to articles between 2008 and 2018.
	As features we combine a seven-dimensional day-of-the-week encoding, sinusoidal features for the day of the year and the time of day of the post, information about the post itself (whether it is gilded, its own Flesch score, and the number of tokens), and finally a 25-dimensional embedding of the title, set using a normalized sum of GloVe embeddings~\citep{pennington2014glove} of the tokens, excluding English stop-words via NLTK~\citep{loper2002nltk}.
\end{itemize}

We again use reasonable guesses of data information to set the static priors, and to initialized the learning schemes.
\ifdefined\arxiv
\begin{itemize}[itemsep=1pt]
\else
\begin{itemize}[leftmargin=*,topsep=-2pt,noitemsep]\setlength\itemsep{1pt}
\fi
	\item Synthetic: $\nu=0$, $\sigma=1$, $a=-100$, $b=100$
	\item CitiBike: $\nu=10$, $\sigma=1$, $a=0$, $b=50/(1-q)$
	\item BBC: $\nu=50$, $\sigma=10$, $a=-100-100/(1-q)$, $b=100+100q$
\end{itemize}
We use $a$ and $b$ for the static Uniform distributions, $\nu$ and $\sigma$ for the static Cauchy distributions, in the case of nonnegative data (CitiBike) we use $\nu$ for the {\em scale} of the half-Cauchy distribution, and for the learning schemes we initialize their Laplace priors to be centered at $\nu$ with scale $\sigma$.
We again use the COCOB optimizer for non-private and proxy learning, and for robustness we mix with the Cauchy (or half-Cauchy for nonnegative data) with coefficient 0.1 on the robust prior.
For the \texttt{PubPrev} method, we set its scale using $\sigma$.
For DP-FTRL, we heavily tune it to show the possibility of learning on the synthetic task;
the implementation is adapted from the one here: \url{https://github.com/google-research/DP-FTRL}.
All results are reported as averages over forty trials.\looseness-1
%!TEX root = main.tex

\newpage
\section{Additional proofs for multiple quantile release}

\begin{Lem}\label{lem:error}
    In Algorithm~\ref{alg:quantiles}, for any $i\in[m]$ we have
    \begin{enumerate}
        \item $\Gap_{\tilde q_i}(\*{\hat x}_i,o)
        \le\Gap_{q_i}(\*x,o)+\hat\gamma_i~\forall~o\in\R$
        \item $\Gap_{q_i}(\*x,o)
        \le\Gap_{\tilde q_i}(\*{\hat x}_i,o)+\hat\gamma_i~\forall~o\in[\hat a_i,\hat b_i]$
    \end{enumerate}
    where $\hat\gamma_i=(1-\tilde q_i)\Gap_{\underline q_i}(\*x,\hat a_i)+\tilde q_i\Gap_{\overline q_i}(\*x,\hat b_i)$.
\end{Lem}
\begin{proof}
    For $o\in[\hat a_i,\hat b_i]$ we apply the triangle inequality twice to get
    \begin{align}
        \begin{split}
            \Gap_{\tilde q_i}(\*{\hat x}_i,o)
            &=|\max_{\*{\hat x}_{[j]}<o}j-\lfloor\tilde q_i\hat n_i\rfloor|\\
            &=|\max_{\*{\hat x}_{[j]}<o}j+\max_{\*x_{[j]}<\hat a_i}j-\lfloor q_in\rfloor+\lfloor q_in\rfloor-\max_{\*x_{[j]}<\hat a_i}j-\lfloor\tilde q_i\hat n_i\rfloor|\\
            &\le\Gap_{q_i}(\*x,o)+\left|\lfloor\tilde q_i(\lfloor\overline q_in\rfloor-\lfloor\underline q_in\rfloor)\rfloor+\lfloor\underline q_in\rfloor-\max_{\*x_{[j]}<\hat a_i}j-\lfloor\tilde q_i(\max_{\*x_{[j]}<\hat b_i}j-\max_{\*x_{[j]}<\hat a_i}j)\rfloor\right|\\
            &\le\Gap_{q_i}(\*x,o)+(1-\tilde q_i)\Gap_{\underline q_i}(\*x,\hat a_i)+\tilde q_i\Gap_{\overline q_i}(\*x,\hat b_i)
        \end{split}
    \end{align}
    and again to get
    \begin{align}
        \begin{split}
            \Gap_{q_i}(\*x,o)
            &=|\max_{\*x_{[j]}<o}j-\lfloor q_in\rfloor|\\
            &=|\max_{\*{\hat x}_{[j]}<o}j+\max_{\*x_{[j]}<\hat a_i}j-\lfloor\tilde q_i\hat n_i\rfloor+\lfloor\tilde q_i\hat n_i\rfloor-\lfloor q_in\rfloor|\\
            &\le\Gap_{\tilde q_i}(\*{\hat x}_i,o)+\left|\max_{\*x_{[j]}<\hat a_i}j-\lfloor\tilde q_i(\max_{\*x_{[j]}<\hat b_i}j+\max_{\*x_{[j]}<\hat a_i}j)\rfloor-\lfloor\tilde q_i(\lfloor\overline q_in\rfloor-\lfloor\underline q_in\rfloor)\rfloor-\lfloor\underline q_in\rfloor\right|\\
            &\le\Gap_{\tilde q_i}(\*{\hat x}_i,o)+(1-\tilde q_i)\Gap_{\underline q_i}(\*x,\hat a_i)+\tilde q_i\Gap_{\overline q_i}(\*x,\hat b_i)
        \end{split}
    \end{align}
    For $o<\hat a_i$ we use the fact that $\max_{\*x_{[j]}<o}j\le\max_{\*x_{[j]}<\hat a_i}j$ and the triangle inequality to get
    \begin{align}
        \begin{split}
            \Gap_{\tilde q_i}(\*{\hat x}_i,o)
            &=\lfloor\tilde q_i\hat n_i\rfloor\\
            &=\lfloor\tilde q_i(\max_{\*x_{[j]}<\hat b_i}j-\max_{\*x_{[j]}<\hat a_i}j)\rfloor\\
            &\le\lfloor\tilde q_i\max_{\*x_{[j]}<\hat b_i}j\rfloor+\lfloor(1-\tilde q_i)\max_{\*x_{[j]}<\hat a_i}j\rfloor-\max_{\*x_{[j]}<o}j\\
            &=\lfloor\tilde q_i\max_{\*x_{[j]}<\hat b_i}j\rfloor+\lfloor(1-\tilde q_i)\max_{\*x_{[j]}<\hat a_i}j\rfloor-\max_{\*x_{[j]}<o}j+\lfloor q_in\rfloor|-\lfloor\tilde q_i(\lfloor\overline q_in\rfloor-\lfloor\underline q_in\rfloor)\rfloor-\lfloor\underline q_in\rfloor\\
            &\le\Gap_{q_i}(\*x,o)+(1-\tilde q_i)\Gap_{\underline q_i}(\*x,\hat a_i)+\tilde q_i\Gap_{\overline q_i}(\*x,\hat b_i)
        \end{split}
    \end{align}
    For $o>\hat b_i$ we use the fact that $\max_{\*x_{[j]}<\hat b_i}j\le\max_{\*x_{[j]}<o}j$ and the triangle inequality to get
    \begin{align}
        \begin{split}
            \Gap_{\tilde q_i}(\*{\hat x}_i,o)
            &=\lfloor(1-\tilde q_i)\hat n_i\rfloor\\
            &=\lfloor(1-\tilde q_i)(\max_{\*x_{[j]}<\hat b_i}j-\max_{\*x_{[j]}<\hat a_i}j)\rfloor\\
            &\le\max_{\*x_{[j]}<o}j-\lfloor\tilde q_i\max_{\*x_{[j]}<\hat b_i}j-\lfloor(1-\tilde q_i)\max_{\*x_{[j]}<\hat a_i}j\\
            &=\max_{\*x_{[j]}<o}j-\lfloor\tilde q_i\max_{\*x_{[j]}<\hat b_i}j-\lfloor(1-\tilde q_i)\max_{\*x_{[j]}<\hat a_i}j-\lfloor q_in\rfloor+\lfloor\tilde q_i(\lfloor\overline q_in\rfloor-\lfloor\underline q_in\rfloor)\rfloor+\lfloor\underline q_in\rfloor\\
            &\le\Gap_{q_i}(\*x,o)+(1-\tilde q_i)\Gap_{\underline q_i}(\*x,\hat a_i)+\tilde q_i\Gap_{\overline q_i}(\*x,\hat b_i)
        \end{split}
    \end{align}
\end{proof}

\begin{Lem}\label{lem:empirical}
    For any $\gamma>0$ the estimate $o_i$ of the quantile $q_i$ by Algorithm~\ref{alg:quantiles} satisfies 
    \begin{equation}
        Pr\{\Gap_{q_i}(\*x,o_i)\ge\gamma\}\le\frac{\exp\left(\varepsilon_i(\hat\gamma_i-\gamma)/2\right)}{\Psi_{\*{\hat x}_i}^{(\tilde q_i,\varepsilon_i)}(\hat\mu_i)}
    \end{equation}
\end{Lem}
\begin{proof}
    We use $k_i$ to denote the interval $\hat I_k^{(j)}$ sampled at index $i$ in the algorithm and note that $o_i$ corresponds to the released number $o$ at that index.
    Since $o_i\in[\hat a_i,\hat b_i]$, applying Lemma~\ref{lem:error} yields
    \begin{align}
        \begin{split}
            \Pr\{\Gap_{q_i}(\*x,o_i)\ge\gamma\}
            &=\sum_{j=0}^{\hat n_i}\Pr\{k_i=j\}1_{\Gap_{q_i}(\*x,\hat I_j^{(i)})\ge\gamma}\\
            &=\sum_{j=0}^{n_i}\frac{\exp(-\varepsilon\Gap_{\tilde q_i}(\*{\hat x}_i,\hat I_j^{(i)})/2)\hat\mu_i(\hat I_j^{(i)})1_{\Gap_{q_i}(\*x,\hat I_j^{(i)})\ge\gamma}}{\sum_{l=0}^{\hat n_i}\exp(\varepsilon u_{\tilde q_i}(\*{\hat x}_i,\hat I_l^{(i)})/2)\hat\mu_i(\hat I_l)}\\
            &\le\frac{\exp(\varepsilon\hat\gamma_i/2)}{\Psi_{\*{\hat x}_i}^{(\tilde q_i,\varepsilon_i)}(\hat\mu_i)}\sum_{j=0}^{n_i}\exp(-\varepsilon\Gap_{q_i}(\*x,\hat I_j^{(i)})/2)\hat\mu_i(\hat I_j^{(i)})1_{\Gap_{q_i}(\*x,\hat I_j^{(i)})\ge\gamma}\\
            &\le\frac{\exp(\varepsilon(\hat\gamma_i-\gamma)/2)}{\Psi_{\*{\hat x}_i}^{(\tilde q_i,\varepsilon_i)}(\hat\mu_i)}
        \end{split}
    \end{align}
\end{proof}

\begin{Lem}\label{lem:prior}
    For any $\gamma>0$ the estimate $o_i$ of the quantile $q_i$ by Algorithm~\ref{alg:quantiles} with edge-based prior adaptation satisfies
    \begin{equation}
        \Pr\{\Gap_{q_i}(\*x,o_i)\ge\gamma\}
        \le\frac{\exp(\varepsilon(\hat\gamma_i-\gamma/2))}{\Psi_{\*x}^{(q_i,\varepsilon_i)}(\mu_i)}
    \end{equation}
\end{Lem}
\begin{proof}
    Applying Lemma~\ref{lem:error} yields the following lower bound on $\Psi_{\tilde q_i}^{(\varepsilon_i)}(\*{\hat x}_i,\hat\mu_i)$:
    \begin{align}
        \begin{split}
            \sum_{l=0}^{\hat n_i}\exp(\varepsilon u_{\tilde q_i}(\*{\hat x}_i,\hat I_l^{(i)})/2)\hat\mu_i(\hat I_l^{(i)})
            &=\exp(\varepsilon u_{\tilde q_i}(\*{\hat x}_i,\hat I_0^{(i)})/2)\mu_i((-\infty,\hat a_i])+\exp(\varepsilon u_{\tilde q_i}(\*{\hat x}_i,\hat I_{\hat n_i}^{(i)})/2)\mu_i([\hat b_i,\infty))\\
            &\qquad+\sum_{l=0}^{\hat n_i}\exp(\varepsilon u_{\tilde q_i}(\*{\hat x}_i, \hat I_l^{(i)})/2)\mu_i(\hat I_l)\\
            &=\sum_{l=0}^{\max_{\*x_{[j]}<\hat a_i}j}\exp(-\varepsilon\Gap_{\tilde q_i}(\*{\hat x}_i,I_l\cap(-\infty,\hat a_i])/2)\mu_i(I_l\cap(-\infty,\hat a_i])\\
            &\qquad+\sum_{l=\max_{\*x_{[j]}<\hat b_i}j}^{n}\exp(-\varepsilon\Gap_{\tilde q_i}(\*{\hat x}_i,I_l\cap[\hat b_i,\infty))/2)\mu_i(I_l\cap[\hat b_i,\infty))\\
            &\qquad+\sum_{l=\max_{\*x_{[j]}<\hat a_i}j}^{\max_{\*x_{[j]}<\hat b_i}j}\exp(-\varepsilon\Gap_{\tilde q_i}(\*{\hat x}_i,I_l\cap[\hat a_i,\hat b_i])\mu_i(I_l\cap[\hat a_i,\hat b_i])\\
            &\ge\Psi_{\*x}^{(q_i,\varepsilon_i)}(\mu_i)\exp(-\varepsilon\hat\gamma_i/2)
        \end{split}
    \end{align}
    Substituting into Lemma~\ref{lem:empirical-quantiles} yields the result.
\end{proof}

\newpage
\begin{Lem}\label{lem:shuffle}
    Suppose $q_0<q_1$ are two quantiles and $o_0>o_1$.
    Then 
    \begin{equation}
        \max_{i=0,1}\Gap_{q_i}(\*x,o_i)
        \ge\max_{i=0,1}\Gap_{q_i}(\*x,o_{1-i})
    \end{equation}
\end{Lem}
\begin{proof}
    We consider four cases.
    If $\lfloor q_0|\*x|\rfloor\le\max_{\*x_{[j]}<o_1}j$ and $\lfloor q_1|X|\rfloor\le\max_{\*x_{[j]}<o_0}j$ then 
    \begin{equation}
        \lfloor q_0|\*x|\rfloor
        \le\min\{\lfloor q_1|\*x|\rfloor,\max_{\*x_{[j]}<o_1}j\}
        \le\max\{\lfloor q_1|\*x|\rfloor,\max_{\*x_{[j]}<o_1}j\}
        \le\max_{\*x_{[j]}<o_0}j
    \end{equation}
    and so
    \begin{equation}
        \max_{i=0,1}\Gap_{q_i}(\*x,o_i)
        =\max_{\*x_{[j]}<o_0}j-\lfloor q_0|\*x|\rfloor
        \ge\max_{i=0,1}\Gap_{q_i}(X,o_{i-1})
    \end{equation}
    If $\lfloor q_0|X|\rfloor\le\max_{\*x_{[j]}<o_1}j$ and $\lfloor q_1|\*x|\rfloor>\max_{\*x_{[j]}<o_0}j$ then 
    \begin{equation}
        \lfloor q_0|\*x|\rfloor
        \le\max_{\*x_{[j]}<o_1}j
        \le\max_{\*x_{[j]}<o_0}j
        <\lfloor q_1|\*x|\rfloor
    \end{equation}
    and so both improve after swapping.
    If $\lfloor q_0|\*x|\rfloor>\max_{\*x_{[j]}<o_1}j$ and $\lfloor q_1|\*x|\rfloor>\max_{\*x_{[j]}<o_0}j$ then
    \begin{equation}
        \max_{\*x_{[j]}<o_1}j
        \le\min\{\lfloor q_0|\*x|\rfloor,\max_{\*x_{[j]}<o_0}j\}
        \le\max\{\lfloor q_0|\*x|\rfloor,\max_{\*x_{[j]}<o_0}j\}
        \le\lfloor q_1|\*x|\rfloor
    \end{equation}
    and so
    \begin{equation}
        \max_{i=0,1}\Gap_{q_i}(\*x,o_i)
        =\max_{\*x_{[j]}<o_1}j-\lfloor q_1|\*x|\rfloor
        \ge\max_{i=0,1}\Gap_{q_i}(\*x,o_{i-1})
    \end{equation}
    Finally, if $\lfloor q_0|\*x|\rfloor>\max_{\*x_{[j]}<o_1}j$ and $\lfloor q_1|\*x|\rfloor\le\max_{\*x_{[j]}<o_0}j$ then
    \begin{equation}
        \max_{\*x_{[j]}<o_1}j
        <\lfloor q_0|\*x|\rfloor
        \le\lfloor q_1|\*x|\rfloor
        \le\max_{\*x_{[j]}<o_0}j
    \end{equation}
    so swapping will make the new largest error for each quantile at most as large as the other quantile's current error.
\end{proof}

\newpage
\section{Additional proofs for online learning}

\subsection{Online-to-batch conversion}\label{sec:o2b}

\begin{Thm}\label{thm:o2b}
    Suppose an online algorithm sees a sequence $\ell_{\*x_1}(\cdot),\dots,\ell_{\*x_T}(\cdot):\Theta\mapsto[0,B]$ of convex losses whose data $\*x_1,\dots,\*x_T$ are drawn i.i.d. from some distribution $\D$, and let $\theta_1,\dots,\theta_T$ be its predictions.
    If $\max_{\theta\in\Theta}\sum_{t=1}^T\ell_{\*x_t}(\theta_t)-\ell_{\*x_t}(\theta)\le R_T$, $\hat\theta=\frac1T\sum_{t=1}^T\theta_t$, and $T=\Omega\left(T_\alpha+\frac{B^2}{\alpha^2}\log\frac1{\beta'}\right)$ for $T_\alpha=\min_{2R_T\le T\alpha}T$, then w.p. $\ge1-\beta'$
    \ifdefined\arxiv$\else
    \begin{equation}
    \fi
        \E_{\*x\sim\D}\ell_{\*x}(\hat\theta)
        \le\min_{\theta\in\Theta}\E_{\*x\sim\D}\ell_{\*x}(\theta)+\alpha
    \ifdefined\arxiv$.\else
    \end{equation}
    \fi
\end{Thm}
\begin{proof}
	This is a formalization of a standard procedure;
	we follow the argument in~\citet[Lemma~A.1]{khodak2022awp}.
    Applying Jensen’s inequality, \citet[Proposition~1]{cesa-bianchi2004online2batch}, the assumption that regret is $\le R_T$, and Hoeffding's inequality yields
    \ifdefined\arxiv
    \begin{align}
    \begin{split}
    \E_{\*x\sim\D}\ell_{\*x}(\hat\theta)
    \le\frac1T\sum_{t=1}^T\E_{\*x\sim\D}\ell_{\*x}(\theta_t)
    &\le\frac1T\sum_{t=1}^T\ell_{\*x_t}(\theta_t)+B\sqrt{\frac2T\log\frac2{\beta'}}\\
    &\le\min_{\theta\in\Theta}\frac1T\sum_{t=1}^T\ell_{\*x_t}(\theta)+\frac{R_T}T+B\sqrt{\frac2T\log\frac2{\beta'}}\\
    &\le\min_{\theta\in\Theta}\E_{\*x\sim\D}\ell_{\*x}(\theta)+\frac{R_T}T+2B\sqrt{\frac2T\log\frac2{\beta'}}
    \end{split}
    \end{align}
    \else
    \begin{align}
        \begin{split}
            \E_{\*x\sim\D}\ell_{\*x}(\hat\theta)
            \le\frac1T\sum_{t=1}^T\E_{\*x\sim\D}\ell_{\*x}(\theta_t)
            \le\frac1T\sum_{t=1}^T\ell_{\*x_t}(\theta_t)+B\sqrt{\frac2T\log\frac2{\beta'}}
            &\le\min_{\theta\in\Theta}\frac1T\sum_{t=1}^T\ell_{\*x_t}(\theta)+\frac{R_T}T+B\sqrt{\frac2T\log\frac2{\beta'}}\\
            &\le\min_{\theta\in\Theta}\E_{\*x\sim\D}\ell_{\*x}(\theta)+\frac{R_T}T+2B\sqrt{\frac2T\log\frac2{\beta'}}
        \end{split}
    \end{align}
    \fi
    w.p. $\ge1-\beta'$. 
    Substituting the lower bound on $T$ yields the result.
\end{proof}

\subsection{Negative log-inner-product losses}\label{sec:overlap}

For functions of the form $f_t(\mu)=-\log\int_a^bs_t(o)\mu(o)do$, \citet{balcan2021ltl} showed $\tilde\BigO(T^{3/4})$ regret for the case $s_t(o)\in\{0,1\}~\forall~o\in[a,b]$ using a variant of exponentiated gradient with a dynamic discretization.
Notably their algorithm can be extended to (non-privately) learn $-\log\Psi_{\*x_t}^{(q)}(\mu)$, since $s_t$ in this case is one on the optimal interval and zero elsewhere.
However, the changing discretization and dependence of the analysis on the range of $s_t$ suggests it may be difficult to privatize their approach.
The discretized form $-\log\langle\*s_t,\*w\rangle$ is more heavily studied, arising in portfolio management~\cite{cover1991universal}.
It enjoys the exp-concavity property, leading to $\BigO(d\log T)$ regret using the EWOO method~\cite{hazan2007logarithmic}.
However, EWOO requires maintaining and sampling from a distribution defined by a product of inner products, which is inefficient and similarly difficult to privatize.
Other algorithms, e.g. adaptive FTAL~\cite{hazan2007logarithmic}, also attain logarithmic regret for exp-concave functions, but the only private variant we know of is non-adaptive and only guarantees $\BigO(\sqrt T)$-regret for non-strongly-convex losses~\cite{smith2013optimal}.
The adaptivity, which is itself data-dependent, seems critical for taking advantage of exp-concavity.\looseness-1

\begin{Lem}\label{lem:lipschitz}
	If $f_t(\mu_{\*W})=-\log\sum_{i=1}^m\frac{1/m}{\langle\*s_{t,i},\*W_{[i]}\rangle}$ for $\*s_{t,i}\in\R_{\ge0}^d$ then $\|\nabla_{\*W}f_t(\mu_{\*W})\|_1\le d/\gamma~\forall~\*W\in\triangle_d^m$ s.t. $\*W_{[i,j]}\ge\gamma/d~\forall~i,j$ for some $\gamma\in(0,1]$.
\end{Lem}
\begin{proof}
	\begin{align}
	\begin{split}
	\|\nabla_{\*W}f_t(\mu_{\*W})\|_1
	=\sum_{i=1}^m\|\nabla_{\*W_{[i]}}f_t(\mu_{\*W})\|_1
	&=\left(\sum_{i=1}^m\frac1{\langle\*s_{t,i},\*W_{[i]}\rangle}\right)^{-1}\sum_{i=1}^m\sum_{j=1}^d\frac{\*s_{t,i[j]}}{\langle\*s_{t,i},\*W_{[i]}\rangle^2}\\
	&\le\left(\sum_{i=1}^m\frac1{\langle\*s_{t,i},\*W_{[i]}\rangle}\right)^{-1}\sum_{i=1}^m\frac1{\langle\*s_{t,i},\*W_{[i]}\odot\*W_{[i]}\rangle}
	\le d/\gamma
	\end{split}
	\end{align}
	where the first inequality follows by Sedrakyan's inequality and the second by $\*W_{[i,j]}\ge\gamma/d$.\looseness-1
\end{proof}

\newpage
\subsubsection{Proof of Lemma~\ref{lem:refined-sensitivity} for $m>1$}\label{sec:refined-sensitivity}

\begin{proof}
		Let $\*{\tilde x}_t$ be a neighboring dataset of $\*x_t$ constructed by adding or removing a single element, and let $U_{\*{\tilde x}_t}^{(\varepsilon)}$ be the corresponding loss function.
		We note that changing from $\*x_t$ to $\*{\tilde x}_t$ changes the value of $\Gap_{q_i}(\*x_t,o)$ at any point $o\in[a,b]$ by at most $\pm1$ and so the value of the exponential score at any point $o\in[a,b]$ is changed by at most a multiplicative factor $\exp(-\varepsilon_i/2)$ in either direction.
		Therefore\looseness-1
		\ifdefined\arxiv\vspace{-2mm}\else\fi
		\begin{align}
		\begin{split}
		\*{\tilde s}_{t,i[j]}
		&=\int_{a+\frac{b-a}d(j-1)}^{a+\frac{b-a}dj}\exp(-\varepsilon_i\Gap_{q_i}(\*{\tilde x}_t,o)/2)do\\
		&\in\exp(\pm\varepsilon_i/2)\int_{a+\frac{b-a}d(j-1)}^{a+\frac{b-a}dj}\exp(-\varepsilon_i\Gap_{q_i}(\*x_t,o)/2)do
		=\exp(\pm\varepsilon_i/2)\*s_{t,i[j]}\ifdefined\arxiv\vspace{-3mm}\else\fi
		\end{split}
		\end{align}
		where $\pm$ indicates the interval between values.
		\ifdefined\arxiv\vspace{-1mm}
		\begin{align}
		\begin{split}
		&\|\nabla_{\*W}U_{\*x_t}^{(\varepsilon)}(\*W)-\nabla_{\*W}U_{\*{\tilde x}_t}^{(\varepsilon)}(\*W)\|_F\\
		&\quad=\sqrt{\sum_{i=1}^m\sum_{j=1}^d\left(\left(\sum_{i'=1}^m\frac1{\langle\*s_{t,i'},\*W_{[i']}\rangle}\right)^{-1}\frac{\*s_{t,i[j]}}{\langle\*s_{t,i},\*W_{[i]}\rangle^2}-\left(\sum_{i'=1}^m\frac1{\langle\*{\tilde s}_{t,i'},\*W_{[i']}\rangle}\right)^{-1}\frac{\*{\tilde s}_{t,i[j]}}{\langle\*{\tilde s}_{t,i},\*W_{[i]}\rangle^2}\right)^2}\\
		&\quad=\left(\sum_{i'=1}^m\frac1{\langle\*s_{t,i'},\*W_{[i']}\rangle}\right)^{-1}\sqrt{\sum_{i=1}^m\sum_{j=1}^d\left(\frac{\*s_{t,i[j]}}{\langle\*s_{t,i},\*W_{[i]}\rangle^2}-\frac{\*{\tilde s}_{t,i[j]}}{\langle\*{\tilde s}_{t,i},\*W_{[i]}\rangle^2}\frac{\sum_{i'=1}^m\frac1{\langle\*s_{t,i'},\*W_{[i']}\rangle}}{\sum_{i'=1}^m\frac1{\langle\*{\tilde s}_{t,i'},\*W_{[i']}\rangle}}\right)^2}\\
		&\quad=\left(\sum_{i'=1}^m\frac1{\langle\*s_{t,i'},\*W_{[i']}\rangle}\right)^{-1}\sqrt{\sum_{i=1}^m\sum_{j=1}^d\frac{\*s_{t,i[j]}^2}{\langle\*W_{t,i},\*W_{[i]}\rangle^4}\left(1-\frac{\langle\*g_{t,i},\*x_{[i]}\rangle^2}{\langle\*{\tilde s}_{t,i},\*W_{[i]}\rangle^2}\frac{\sum_{i'=1}^m\frac{\*{\tilde s}_{t,i[j]}}{\langle\*s_{t,i'},\*W_{[i']}\rangle}}{\sum_{i'=1}^m\frac{\*s_{t,i[j]}}{\langle\*{\tilde s}_{t,i'},\*W_{[i']}\rangle}}\right)^2}\\
		&\quad\le\left(\sum_{i'=1}^m\frac1{\langle\*s_{t,i'},\*W_{[i']}\rangle}\right)^{-1}\sum_{i=1}^m\sum_{j=1}^d\frac{\*s_{t,i[j]}}{\langle\*s_{t,i},\*W_{[i]}\rangle^2}|1-\kappa_{i,j}|
		\le\frac d\gamma\max_{i,j}|1-\kappa_{i,j}|\vspace{-2mm}
		\end{split}
		\end{align}
		\else
		\begin{align}
		\begin{split}
		\|\nabla_{\*W}U_{\*x_t}^{(\varepsilon)}&(\*W)-\nabla_{\*W}U_{\*{\tilde x}_t}^{(\varepsilon)}(\*W)\|_F\\
		&=\sqrt{\sum_{i=1}^m\sum_{j=1}^d\left(\left(\sum_{i'=1}^m\frac1{\langle\*s_{t,i'},\*W_{[i']}\rangle}\right)^{-1}\frac{\*s_{t,i[j]}}{\langle\*s_{t,i},\*W_{[i]}\rangle^2}-\left(\sum_{i'=1}^m\frac1{\langle\*{\tilde s}_{t,i'},\*W_{[i']}\rangle}\right)^{-1}\frac{\*{\tilde s}_{t,i[j]}}{\langle\*{\tilde s}_{t,i},\*W_{[i]}\rangle^2}\right)^2}\\
		&=\left(\sum_{i'=1}^m\frac1{\langle\*s_{t,i'},\*W_{[i']}\rangle}\right)^{-1}\sqrt{\sum_{i=1}^m\sum_{j=1}^d\left(\frac{\*s_{t,i[j]}}{\langle\*s_{t,i},\*W_{[i]}\rangle^2}-\frac{\*{\tilde s}_{t,i[j]}}{\langle\*{\tilde s}_{t,i},\*W_{[i]}\rangle^2}\frac{\sum_{i'=1}^m\frac1{\langle\*s_{t,i'},\*W_{[i']}\rangle}}{\sum_{i'=1}^m\frac1{\langle\*{\tilde s}_{t,i'},\*W_{[i']}\rangle}}\right)^2}\\
		&=\left(\sum_{i'=1}^m\frac1{\langle\*s_{t,i'},\*W_{[i']}\rangle}\right)^{-1}\sqrt{\sum_{i=1}^m\sum_{j=1}^d\frac{\*s_{t,i[j]}^2}{\langle\*W_{t,i},\*W_{[i]}\rangle^4}\left(1-\frac{\langle\*g_{t,i},\*x_{[i]}\rangle^2}{\langle\*{\tilde s}_{t,i},\*W_{[i]}\rangle^2}\frac{\sum_{i'=1}^m\frac{\*{\tilde s}_{t,i[j]}}{\langle\*s_{t,i'},\*W_{[i']}\rangle}}{\sum_{i'=1}^m\frac{\*s_{t,i[j]}}{\langle\*{\tilde s}_{t,i'},\*W_{[i']}\rangle}}\right)^2}\\
		&\le\left(\sum_{i'=1}^m\frac1{\langle\*s_{t,i'},\*W_{[i']}\rangle}\right)^{-1}\sum_{i=1}^m\sum_{j=1}^d\frac{\*s_{t,i[j]}}{\langle\*s_{t,i},\*W_{[i]}\rangle^2}|1-\kappa_{i,j}|
		\le\frac d\gamma\max_{i,j}|1-\kappa_{i,j}|
		\end{split}
		\end{align}
		\fi
		where we have
		\ifdefined\arxiv\vspace{-2mm}
		\begin{align}
		\begin{split}
		\kappa_{i,j}
		=\frac{\langle\*s_{t,i},\*W_{[i]}\rangle^2}{\langle\*{\tilde s}_{t,i},\*x_{[i]}\rangle^2}\frac{\sum\limits_{i'=1}^m\frac{\*{\tilde s}_{t,i[j]}}{\langle\*s_{t,i'},\*W_{[i']}\rangle}}{\sum\limits_{i'=1}^m\frac{\*s_{t,i[j]}}{\langle\*{\tilde s}_{t,i'},\*W_{[i']}\rangle}}
		\in\frac{\langle\*s_{t,i},\*W_{[i]}\rangle^2}{\langle\*s_{t,i},\*W_{[i]}\rangle^2e^{\pm\varepsilon_i}}\frac{\sum\limits_{i'=1}^m\frac{\*s_{t,i[j]}e^{\pm\varepsilon_{i'}/2}}{\langle\*s_{t,i'},\*W_{[i']}\rangle}}{\sum\limits_{i'=1}^m\frac{\*s_{t,i[j]}}{\langle\*s_{t,i'},\*W_{[i']}\rangle e^{\pm\varepsilon_{i'}/2}}}
		=e^{\pm2\max_i\varepsilon_i}\ifdefined\arxiv\vspace{-3mm}\else\fi
		\end{split}
		\end{align}
		\else
		\begin{align}
		\begin{split}
		\kappa_{i,j}
		=\frac{\langle\*s_{t,i},\*W_{[i]}\rangle^2}{\langle\*{\tilde s}_{t,i},\*x_{[i]}\rangle^2}\frac{\sum\limits_{i'=1}^m\frac{\*{\tilde s}_{t,i[j]}}{\langle\*s_{t,i'},\*W_{[i']}\rangle}}{\sum\limits_{i'=1}^m\frac{\*s_{t,i[j]}}{\langle\*{\tilde s}_{t,i'},\*W_{[i']}\rangle}}
		\in\frac{\langle\*s_{t,i},\*W_{[i]}\rangle^2}{\langle\*s_{t,i},\*W_{[i]}\rangle^2\exp(\pm\varepsilon_i)}\frac{\sum\limits_{i'=1}^m\frac{\*s_{t,i[j]}\exp(\pm\frac{\varepsilon_{i'}}2)}{\langle\*s_{t,i'},\*W_{[i']}\rangle}}{\sum\limits_{i'=1}^m\frac{\*s_{t,i[j]}}{\langle\*s_{t,i'},\*W_{[i']}\rangle\exp(\pm\frac{\varepsilon_{i'}}2)}}
		=\exp(\pm2\max_i\varepsilon_i)\ifdefined\arxiv\vspace{-3mm}\else\fi
		\end{split}
		\end{align}
		\fi
		Substituting into the previous inequality and taking the minimum with the $\ell_1$ bound on the gradient of the losses from Lemma~\ref{lem:lipschitz} yields the result.
\end{proof}

\subsubsection{Settings of $\gamma$ and $d$ for Corollary~\ref{cor:quantile-regret}}
	
\begin{enumerate}[itemsep=1pt]
	\item $\lambda$-robust and discrete $\mu_{[i]}\in\F_{0,d}^{(\lambda)}$: $\gamma=\lambda$
	\item $\lambda$-robust and $V$-Lipschitz $\mu_{[i]}\in\F_{V,1}^{(\lambda)}$: $\gamma=\lambda$ and $d=\left\lceil\sqrt{\frac{V(b-a)^3}{\bar\psi}\sqrt{\left(1+\frac{\min\{1,\tilde\varepsilon_m\}}{\varepsilon'}\right)T}}\right\rceil$
	\item discrete $\mu_{[i]}\in\F_{0,d}$: $\gamma=\sqrt{md}\sqrt[4]{\frac{1+\min\{1,\tilde\varepsilon_m\}/\varepsilon'}T}$
	\item $V$-Lipschitz $\mu_{[i]}\in\F_{V,1}$:  $\gamma=\sqrt m\sqrt[4]{\frac{V(b-a)^3}{\bar\psi}}\sqrt[8]{\frac{1+\min\{1,\tilde\varepsilon_m\}/\varepsilon'}T}$ and \\ $d=\left\lceil\sqrt{\frac{V(b-a)^3}{\bar\psi}\sqrt{\left(1+\frac{\min\{1,\tilde\varepsilon_m\}}{\varepsilon'}\right)T}}\right\rceil$
\end{enumerate}

\begin{algorithm}[!t]
	\DontPrintSemicolon
	\KwIn{sorted unrepeated data $\*x\in(a,b)^n$, ordered quantiles $q_1,\dots,q_m\in(0,1)$,\\ priors $\mu_1,\dots,\mu_m:\R\mapsto\R_{\ge0}$, prior adaptation rule $r\in$\texttt{\{conditional,edge\}},\\ privacy parameters $\varepsilon_1,\dots,\varepsilon_m>0$, branching factor $K\ge2$}
	\SetKwProg{Method}{Method}{:}{}
	\SetKwFunction{quantile}{quantile}
	\tcp{runs single-quantile algorithm on datapoints $\*{\hat x}$}
	\Method{\quantile{$\*{\hat x}$, $q$, $\varepsilon$, $\mu$}}{
		\KwOut{$o\in(a,b)$ w.p. $\propto\exp(-\varepsilon\Gap_q(\*{\hat x},o)/2)\mu(o)$}
	}
	\SetKwFunction{recurse}{recurse}
	\Method{\recurse{$\*j$, $\underline q$, $\overline q$, $\hat a$, $\hat b$}}{
		\tcp{determines $K-1$ indices $\*i$ whose quantiles to compute at this node}
		\uIf{$|\*j|\ge K$}{
			$\*i\gets\begin{pmatrix}\*j_{[\lceil|\*j|/K\rceil]},\cdots,\*j_{[\lceil(K-1)|\*j|/K\rceil]}\end{pmatrix}$\\
		}
		\Else{
			$\*i\gets\*j$\\
		}
		\tcp{restricts dataset to the interval $(\hat a,\hat b)$}
		$\underline k_{\*i}\gets\min_{\*x_{[k]}>\hat a}k$\\
		$\overline k_{\*i}\gets\max_{\*x_{[k]}<\hat b}k$\\
		$\*{\hat x}_{\*i}\gets\begin{pmatrix}\*x_{[\underline k_{\*i}]},\cdots,\*x_{[\overline k_{\*i}]}\end{pmatrix}$\\
		\tcp{sets relative quantiles $\tilde q_i$ and restricts priors to the interval $[\hat a,\hat b]$}
		\For{$j=1,\dots,|\*i|$}{
			$\tilde q_{\*i_{[j]}}\gets(q_{\*i_{[j]}}-\underline q)/(\overline q-\underline q)$\\
			\uIf{$r=${ \em\texttt{conditional}}}{
				$\hat\mu_{\*i_{[j]}}(o)\gets\frac{\mu_{\*i_{[j]}}(o)}{\mu_{\*i_{[j]}}([\hat a,\hat b])}1_{o\in[\hat a,\hat b]}$\\
			}
			\Else{
				$\hat\mu_{\*i_{[j]}}(o)\gets\mu_{\*i_{[j]}}(o)1_{o\in(\hat a,\hat b)}+\mu_{\*i_{[j]}}((-\infty,\hat a])\delta(o-\hat a)+\mu_{\*i_{[j]}}([\hat b,\infty))\delta(o-\hat b)$\\
			}
		}
		\tcp{computes $K-1$ quantiles $\*o_{\*i}$ and sorts the results}
		$\*o_{\*i}\gets\begin{pmatrix}$\quantile{$\*{\hat x}_{\*i},\tilde q_{\*i_{[1]}},\varepsilon_{\*i_{[1]}}/|\*i|,\hat\mu_{\*i_{[1]}}$} , $\cdots$ , \quantile{$\*{\hat x}_{\*i},\tilde q_{\*i_{[|\*i|]}},\varepsilon_{\*i_{[|\*i|]}}/|\*i|,\hat\mu_{\*i_{[|\*i|]}}$}$\end{pmatrix}$\\
		$\*o_{\*i}\gets$\texttt{sort(}$\*o_{\*i}$\texttt{)}\\
		\tcp{recursively computes remaining indices on the $K$ intervals induced by $\*o_{\*i}$}
		\uIf{$|\*j|<K$}{
			$\*o\gets\*o_i$
		}
		\Else{
			$\*o\gets$ \texttt{concat(}\recurse{$\begin{pmatrix}\*j_{[1]},\cdots,\*j_{[\lceil|\*j|/K\rceil-1]}\end{pmatrix},\underline q,q_{\*i_{[1]}},\hat a,\*o_{[1]}$},$\begin{pmatrix}\*o_{[1]}\end{pmatrix}$\texttt{)}\\
			\For{$j=2,\dots,|\*i|$}{
				$\*o\gets$\texttt{concat(}$\*o,$ \recurse{$\begin{pmatrix}\*j_{[\lceil(j-1)|\*j|/K\rceil+1]},\cdots,\*j_{[\lceil j|\*j|/K\rceil-1]}\end{pmatrix},q_{\*i_{[j-1]}},q_{\*i_{[j]}},\*o_{[j-1]},\*o_{[j]}$}\texttt{)}\\
				$\*o\gets$\texttt{concat(}$\*o,\begin{pmatrix}\*o_{[j]}\end{pmatrix}$\texttt{)}\\
			}
			$\*o\gets$ \texttt{concat(}$\*o$, \recurse{$\begin{pmatrix}\*j_{[\lceil(K-1)|\*j|/K\rceil+1]},\cdots,\*j_{[|\*j|]}\end{pmatrix},q_{\*i_{[K-1]}},\overline q,\*o_{[K-1]},\hat b$}\texttt{)}\\
		}
		\KwOut{$\*o$}
	}
	\KwOut{\recurse{$\begin{pmatrix}1,\cdots,m\end{pmatrix},0,1,-\infty,\infty$}}
	\caption{\label{alg:quantiles}\texttt{ApproximateQuantiles} with predictions}
\end{algorithm}

\end{document}